\documentclass[11pt ]{article}

\usepackage[left=2cm,top=2cm,right=2cm,bottom = 2cm]{geometry}
\usepackage[plainpages,urlcolor=red,linktocpage=true]{hyperref}

\usepackage{amscd}
\usepackage{amssymb}
\usepackage{mathrsfs}
\usepackage{amsmath}
\usepackage{mathabx}
\usepackage{amsthm}
\usepackage[all,cmtip]{xy}
\usepackage{enumerate}
\usepackage{enumitem}
\usepackage{ytableau}
\usepackage{tikz}
\usepackage{tikz-cd} 
\usetikzlibrary{matrix}
\usepackage{mathtools}
\usepackage{color}
\usepackage{wasysym}
\usepackage{comment}

\numberwithin{equation}{section}

\newcommand{\thmref}[1]{Theorem~\ref{#1}}
\newcommand{\secref}[1]{Section~\ref{#1}}

\newcommand{\lemref}[1]{Lemma~\ref{#1}}
\newcommand{\propref}[1]{Proposition~\ref{#1}}
\newcommand{\corref}[1]{Corollary~\ref{#1}}
\newcommand{\defref}[1]{Definition~\ref{#1}}

\newcommand{\eqnref}[1]{(\ref{#1})}
\newcommand{\exref}[1]{Example~\ref{#1}}
\newcommand{\figref}[1]{Figure~\ref{#1}}

\newtheorem{theorem}{Theorem}[section]
\newtheorem{lemma}[theorem]{Lemma}
\newtheorem{proposition}[theorem]{Proposition}
\newtheorem{corollary}[theorem]{Corollary}

\theoremstyle{definition}
\newtheorem{example}[theorem]{Example}

\newtheorem{definition}[theorem]{Definition}

\theoremstyle{remark}
\newtheorem{remark}[theorem]{Remark}

\makeatletter
\newcommand*{\rom}[1]{\expandafter\@slowromancap\romannumeral #1@}
\makeatother

\DeclareMathOperator{\ad}{ad}
\DeclareMathOperator{\Ad}{Ad}
\DeclareMathOperator{\Der}{Der}
\DeclareMathOperator{\Ric}{Ric}
\DeclareMathOperator{\sdim}{sdim}
\DeclareMathOperator{\sspan}{span}
\DeclareMathOperator{\Str}{Str}

\DeclareMathOperator{\GL}{GL}
\DeclareMathOperator{\SL}{SL}
\DeclareMathOperator{\SO}{SO}
\DeclareMathOperator{\SOSp}{SOSp}
\DeclareMathOperator{\Sp}{Sp}
\DeclareMathOperator{\SU}{SU}
\DeclareMathOperator{\UU}{U}

\newcommand{\Cc}{\mathcal{C}}

\newcommand{\Fc}{\mathcal{F}}

\newcommand\bigw{\scalebox{.8}[0.8]{$\bigwedge$}}

\allowdisplaybreaks

\begin{document}

\thispagestyle{empty}

\begin{center}

\scalebox{1}{\textbf{\huge Einstein metrics on homogeneous superspaces}}
\\[0.6cm]
\scalebox{0.9}{\LARGE Yang Zhang, Mark D.~Gould, Artem Pulemotov, J\o rgen Rasmussen}
\\[0.3cm]
\textit{School of Mathematics and Physics, University of Queensland\\ 
St Lucia, Brisbane, Queensland 4072, Australia}
\\[0.3cm] 
\textsf{yang.zhang\!\;@\!\;uq.edu.au\quad\ m.gould1\!\;@\!\;uq.edu.au\quad\ a.pulemotov\!\;@\!\;uq.edu.au\quad\ j.rasmussen\!\;@\!\;uq.edu.au}

\vspace{1.4cm}

{\large\textbf{Abstract}}\end{center}
This paper initiates the study of the Einstein equation on homogeneous supermanifolds. First, we produce explicit curvature formulas for graded Riemannian metrics on these spaces. Next, we present a construction of homogeneous supermanifolds by means of Dynkin diagrams, resembling the construction of generalised flag manifolds in classical (non-super) theory. We describe the Einstein metrics on several classes of spaces obtained through this approach. Our results provide examples of compact homogeneous supermanifolds on which the Einstein equation has no solutions, discrete families of solutions, and continuous families of Ricci-flat solutions among invariant metrics. These examples demonstrate that the finiteness conjecture from classical homogeneous geometry fails on supermanifolds, and challenge the intuition furnished by Bochner's vanishing theorem.

\newpage

\tableofcontents

\newpage

\section{Introduction}

Over the past few decades, considerable efforts have been made to extend Riemannian geometry to the realm of 
supermanifolds; see, for example, the classic book~\cite{DeW92} and the more recent texts~\cite{Rog07,CCF11,Kes19}. 
Notions such as Riemannian metric, Levi-Civita connection and Riemann curvature have thus been generalised to this 
setting. The purpose of the present paper is to address some of the fundamental mathematical aspects of these notions,
building on explicit curvature formulas that we derive for certain types of supermanifolds. 

The appearance of supermanifolds in contemporary physical theories makes their mathematical study timely and warrants the
development of manageable curvature expressions in cases with physically relevant symmetries or properties.
Indeed, supergeometry provides a natural framework for formulating a variety of supersymmetric theories,
including supersymmetric quantum field theories, supergravity, and string theory.
The field of supergeometry flourished in attempts to combine the principles of general relativity theory with those of 
supersymmetry (see the books \cite{GGRS83,CDAF91,WB92,BK98}), and was propelled during the so-called 
``first superstring revolution" because of its fundamental role in the Green--Schwarz formulation of string theory \cite{GS84a,GS84b}.
Later, Calabi--Yau supermanifolds and graded analogues of the Fubini--Study metrics on projective
superspaces \cite{Wit04,RW05,Zho05,ARS16} have been linked to mirror symmetry \cite{Set94,Sch96,NV04,AV04,Ahn05,BDRSS05}.
Moreover, supermanifolds arise in the description of strings on AdS$_5\times S^5$ \cite{MT98,BPR04,AF09},
which is of great importance in the AdS/CFT correspondence and enjoys links to integrable field theories.
For further results on supermanifolds and their applications to physics, we refer to 
\cite{BBHZZ00,Zar10,Wit19a,Wit19b,Kes19} and the references therein.

One of the major goals in Riemannian geometry is to determine which manifolds admit `special' geometric structures. 
Once the foundational 
concepts are established and their basic properties understood, one may search for metrics with distinguished features. 
For example, there is a large body of literature on metrics with positive or non-negative sectional curvature; 
see the collection~\cite{HHL14}. Of particular interest in Riemannian geometry are metrics whose curvature is, in some 
sense, constant (think of round metrics on spheres as model examples). A natural way to capture the idea of constant 
curvature mathematically is through the Einstein equation
\begin{align}\label{intro: Einstein}
 \Ric(g)=cg,
\end{align}
where $\Ric$ denotes the Ricci curvature, $g$ the metric, and $c$ the Einstein constant. Solutions to this equation are 
known as Einstein metrics, and many of the classical results on the associated Einstein manifolds are collected in the 
celebrated monograph~\cite{Bes87}. The analysis of~\eqref{intro: Einstein} also underpins the classification of 
three-manifolds via Thurston's geometrisation conjecture~\cite{MT06}, and plays a pivotal role in a raft of other 
breakthrough insights and constructions in geometry, analysis and theoretical physics.

K\"ahler geometry provides an important class of Einstein metrics. These include Ricci-flat metrics on Calabi--Yau 
manifolds appearing in string theory; see~\cite{CHSW85,Gre97,Pol98} and references therein. 
Outside the K\"ahler setting, the majority 
of known Einstein metrics have been found through symmetry reduction. One typically assumes that the manifold admits an 
action of a Lie group that either is transitive or has principal orbits of co-dimension one. 
This helps reduce~\eqref{intro: Einstein} from a system of PDEs to a system of algebraic equations or ODEs. The theory of 
Einstein metrics on homogeneous spaces is particularly fruitful, and includes a plethora of infinite families and isolated 
examples~\cite{WZ86,Boh04,BWZ04,Arv15}. The constructions in the present paper are inspired by that theory.

Another goal in Riemannian geometry is the establishment and exploration of links between the geometry and the topology of a manifold. 
For instance, the Bonnet--Myers theorem states that the positivity of the Ricci curvature implies the 
finiteness of the fundamental group. In the presence of symmetry, the links are particularly strong. 
A compact homogeneous space, for example, typically does not admit a metric with non-positive Ricci 
curvature. According to the Bochner theorem (see~\cite[Theorem~1.84]{Bes87}), if $G_0/K_0$ is a compact homogeneous 
space with non-positive Ricci curvature, then the connected component of the identity of $G_0$ must be abelian.

While supergeometric generalisations of Ricci-flat K\"ahler metrics on Calabi--Yau manifolds have been considered in the 
literature (see, e.g.,~\cite{RW05,Zho05,ARS16}), there seem to be no known results regarding metrics on supermanifolds 
with distinguished curvature properties but without an immediate connection to a K\"ahler-type structure. In this paper, we 
initiate the study of the Einstein equation~\eqref{intro: Einstein} on homogeneous superspaces.
Our work is motivated by the theory of Einstein metrics on homogeneous spaces 
$G_0/K_0$, which is a particularly rich theory
for $G_0$~\emph{compact}. For example, it was shown in~\cite{BL23} that every \emph{non-compact} manifold with a 
homogeneous Einstein metric of negative scalar curvature must be diffeomorphic to~$\mathbb R^n$. 

Our key initial step is to compute the curvature of a metric on a homogeneous superspace $G/K$. Building on this, we 
solve~\eqref{intro: Einstein} on a large class of such spaces, and find examples of homogeneous supermanifolds 
on which~\eqref{intro: Einstein} has no solutions, (up to four) discrete families of solutions, and continuous families of 
solutions among invariant metrics. The solutions we produce appear to be the first nontrivial Einstein metrics obtained in the 
framework of supergeometry without exploiting methods or ideas from K\"ahler geometry. The existence of the discrete 
families is to be expected in light of the results available in the non-super setting. However, the presence of continuous 
families (constructed on $G/K$ with~$G=\SU(m|n)$) shows that the well-known finiteness conjecture from homogeneous 
geometry (see~\cite{BWZ04}) fails on supermanifolds. Remarkably, all metrics in the  continuous  families are Ricci-flat, 
while Bochner's theorem tells us that no Ricci-flat metrics exist on $G_0/K_0$ 
if $G_0$ is a compact Lie group whose identity component is not a torus. Thus, the 
exploration of links between the geometry of a supermanifold and the topology of its associated reduced manifold is an 
intriguing research direction that is likely to challenge intuition based on classical results only.

Often motivated by physical theories, homogeneous superspaces, including special types such as flag 
supermanifolds, have been studied before; see~\cite{BV91,Cor00,Goe08,San10,CCF11,Vis15,Vis22}, and references 
therein. In fact, one may view our work as a continuation of these papers, offering a further but distinct study of the nature 
of homogeneous superspaces with distinguished structures.

The paper is organised as follows. In \secref{sec: prelim}, we recall some definitions and properties of Lie superalgebras 
and their real forms, Lie supergroups, Harish-Chandra pairs, and homogeneous superspaces. In \secref{sec: gradedR} and 
\secref{sec: RRs}, we collect definitions and basic results from Riemannian supergeometry needed in subsequent sections. 
We also show that every left-invariant metric on~$\SL(1|1)$ satisfies \eqref{intro: Einstein}. This gives us, arguably, the 
simplest nontrivial examples of Einstein metrics on supermanifolds and a toy model for more complicated cases. In 
\secref{subsec: curvs}, we present our formulas for the curvature tensors associated with a metric on a
homogeneous superspace, with the proof of the Ricci curvature formula 
deferred until \secref{sec: pfmain}. While the expression for the Riemann curvature is considerably bulkier than 
its non-super counterpart, we find that the expression for the Ricci curvature is comparatively simple. In both expressions, 
a key difference from the non-super setting stems from the use of dual bases. Specialisations and simplifications for 
naturally reductive and diagonal metrics are presented in \secref{sec: nat red} and \secref{sec: diag_curv}, respectively. 
The possible non-positivity of superdimensions and structure constants presents additional challenges in the diagonal case.

\secref{sec: gen_flag} is devoted to Einstein metrics on homogeneous superspaces~$G/K$, where $G$ is a connected Lie 
supergroup with Lie superalgebra~$\mathfrak g$. We assume that $\mathfrak g$ is basic classical to ensure that 
$\mathfrak g^{\mathbb C}$ admits a non-degenerate $\ad_{\mathfrak{g}^{\mathbb{C}}}$-invariant supersymmetric even 
bilinear form, which we use to obtain our background metric on~$G/K$. Furthermore, we demand that $\mathfrak g$ be the 
compact real form of $\mathfrak g^{\mathbb C}$ induced by a so-called star operation, as is the case in the non-super 
setting when $G_0$ is compact. This leaves $\mathfrak g^{\mathbb C}$ to be of type~$A$ or~$C$, whence $G=\SU(m|n)$ 
or $G=\SOSp(2|2n)$; see~\cite{SNR77,AN15,FG23}. 
Imitating the structure of so-called generalised flag manifolds in classical Riemannian geometry (see~\cite{Arv03}),
we will assume that the Dynkin diagram of $K$ is obtained from that of $G$ by removing one or more nodes. 
Here, we restrict ourselves to removing one or two nodes only but intend to discuss elsewhere the
removal of an arbitrary number of nodes.
\secref{sec: flag} thus outlines the construction of homogeneous superspaces by means of Dynkin diagrams, while 
\secref{sec: examflag} concerns the corresponding flag supermanifolds and their compact real forms. In 
\secref{sec: flagA1} and \secref{sec: flagA}, we obtain our classification results for Einstein metrics on $G/K$ 
with $G=\SU(m|n)$. Notably, some of the metrics we find are Ricci-flat, 
and our construction in \secref{sec: flagA1} produces supergeometric analogues of Fubini--Study metrics on projective spaces. 
\secref{sec: flagC} concerns Einstein metrics on $G/K$ with $G=\SOSp(2|2n)$.

\subsubsection*{Notation}

Throughout, the ground field is the field $\mathbb{R}$ of real numbers, unless otherwise stated. We let $\mathbb{C}$ 
denote the field of complex numbers, $\mathbb{H}$ the quaternions (appearing in Section~\ref{subsubsec: typeC}), 
$\mathbb{Z}$ the ring of integers, $\mathbb{Z}_2=\{\bar{0},\bar{1} \}$ the ring of integers modulo~2,  
$\mathbb{R}^+$ (respectively $\mathbb{R}^-$) the set of positive (respectively negative) real numbers, and $\imath$ the 
imaginary unit, and set $\mathbb{R}^\times=\mathbb{R}-\{0\}$. 
Let $\langle-, -\rangle$ denote a scalar superproduct on a vector superspace~$V$. 
If $\{v_1,\ldots,v_d\}$ is a basis for $V$, we write $\{\bar v_1,\ldots,\bar v_d\}$ for the corresponding right dual basis, 
i.e., $\langle v_i,\bar v_j\rangle=\delta_{ij}$ for all $i,j\in\{1,\ldots,d\}$. 
We denote by $g$ a graded Riemannian metric on a Riemannian supermanifold~$M$. 
We write $\mathfrak{g}^{\mathbb{C}}$ for a complex Lie superalgebra, with $\mathfrak{g}$ its real form over $\mathbb{R}$.  
If $\mathfrak{g}^{\mathbb{C}}$ is a basic classic complex Lie superalgebra, 
we denote by $Q$ a fixed, non-degenerate, supersymmetric, invariant, even bilinear form on $\mathfrak{g}^{\mathbb{C}}$.

\section{Supermanifolds}
\label{sec: prelim}

Here, we recall some basic notions from supergeometry. We refer 
to~\cite{Kos77,Bat79,Lei80,BBH91,DeW92,DM99,Var04,Goe08,CCF11} for further background, 
noting (as in~\cite{BCF10}) that these references employ different approaches to supergeometry. 
Focus here is on notions and results needed in subsequent sections.

\subsection{Lie superalgebras and real forms}
\label{subsec: Lie}

A \emph{vector superspace} of dimension $m|n$ is a $\mathbb{Z}_2$-graded vector space $V=V_{\bar{0}}\oplus V_{\bar{1}}$, 
where $m=\dim(V_{\bar{0}})$ and $n=\dim(V_{\bar{1}})$.  The \emph{superdimension} of $V$ is defined as $\sdim(V)=m-n$. 
A nonzero element of $V_{\bar{0}}\cup V_{\bar{1}}$ is said to be 
\emph{homogeneous}, with elements of $V_{\bar{0}}$ called \emph{even} and those of $V_{\bar{1}}$ \emph{odd},
and a basis $\{v_1, \dots, v_{m+n}\}$ for $V$ is said to be \emph{homogeneous} if every $v_i$ is homogeneous. 
A homogeneous vector $v\in V_a$ is said to have \emph{parity} $[v]=a$. 
For the tensor product of vector superspaces $V$ and $W$, we have
\begin{align}\label{VW}
 (V\otimes W)_{\bar{0}}= (V_{\bar{0}}\otimes W_{\bar{0}})\oplus(V_{\bar{1}}\otimes W_{\bar{1}}),\qquad 
 (V\otimes W)_{\bar{1}}= (V_{\bar{0}}\otimes W_{\bar{1}})\oplus(V_{\bar{1}}\otimes W_{\bar{0}}).
\end{align}
Throughout this paper, definitions and calculations 
for vector superspaces are given in terms of homogeneous elements and then extended linearly. Moreover, the ground field 
$\mathbb{R}$ is considered an even vector space over $\mathbb{R}$, 
and a linear map $f\in \mathrm{Hom}_{\mathbb{R}}(V,W)$ is said to be 
\emph{even} if $f(V_{a})\subseteq W_{a}$ for $a\in \mathbb{Z}_2$. This notion extends readily to bilinear maps.

A \emph{Lie superalgebra} is a $\mathbb{Z}_2$-graded vector space (over $\mathbb{R}$ or $\mathbb{C}$) equipped with 
a Lie superbracket $[-,-]$ which is super-skew-symmetric and satisfies the super-Jacobi identity. 
The (even bilinear) \emph{Killing form} is defined for every finite-dimensional Lie superalgebra $\mathfrak{g}$, by
\begin{align*}
 B(X,Y):=\Str_{\mathfrak{g}}\!\big(\!\ad(X)\ad(Y)\big),\qquad X,Y\in\mathfrak{g},
\end{align*}
where $\Str_{\mathfrak{g}}$ is the supertrace over $\mathrm{End}(\mathfrak{g})$. The Killing form may be identically zero.

A finite-dimensional complex Lie superalgebra 
$\mathfrak{g}^{\mathbb{C}}=\mathfrak{g}_{\bar{0}}^{\mathbb{C}}\oplus\mathfrak{g}_{\bar{1}}^{\mathbb{C}}$ 
is called \emph{basic classical} if $\mathfrak{g}^{\mathbb{C}}$ is simple, $\mathfrak{g}_{\bar{0}}^{\mathbb{C}}$ is reductive, 
and $\mathfrak{g}^{\mathbb{C}}$ admits a non-degenerate $\ad_{\mathfrak{g}^{\mathbb{C}}}$-invariant supersymmetric 
even bilinear form. Following Kac~\cite{Kac77,Kac78}, there are two types of basic classical Lie superalgebras, 
characterised by whether $\mathfrak{g}_{\bar{1}}^{\mathbb{C}}$ is reducible (type \rom{1}) or irreducible (type \rom{2)} as a 
$\mathfrak{g}_{\bar{0}}^{\mathbb{C}}$-module.
Thus, $A(m,n)$ and $C(n)$ comprise the basic classical Lie superalgebras of type \rom{1}, while $B(m,n)$, $D(m,n)$, 
$F(4)$, $G(3)$ and $D(2,1,\alpha)$ comprise the ones of type \rom{2}.
Concretely, a basic classical Lie superalgebra of type \rom{1} is isomorphic to one of the following Lie superalgebras: 
$A(m,n)=\mathfrak{sl}(m+1|n+1)^{\mathbb{C}}$ with $m>n\geq0$; 
$A(n|n)=\mathfrak{psl}(n+1|n+1)^{\mathbb{C}}$ with $n\geq1$; 
or $C(n)=\mathfrak{osp}(2|2n-2)^{\mathbb{C}}$ with $n\geq3$ (noting that $C(2)\cong A(1|0)$).

A real Lie superalgebra $\mathfrak{g}$ is called a \emph{real form} of a complex Lie superalgebra 
$\mathfrak{g}^{\mathbb{C}}$ if $\mathfrak{g}^{\mathbb{C}}\cong \mathfrak{g}\otimes_{\mathbb{R}} \mathbb{C}$, while 
a \emph{star operation} $*: \mathfrak{g}^{\mathbb{C}}\to \mathfrak{g}^{\mathbb{C}}$ is an even anti-linear map such that 
\begin{align*}
 [X,Y]^*=[Y^*, X^*],\qquad (X^*)^*=X,
\end{align*}
for any $X,Y\in \mathfrak{g}^{\mathbb{C}}$; see~\cite{SNR77} or~\cite[\S 4.1--4.3]{DM99}. 
Here, anti-linearity means that $(cX)^{*}=\bar{c}X^*$, where $\bar{c}$ is the complex conjugate of $c\in\mathbb{C}$. 
With $\sqrt{\imath}= e^{\frac{\pi}{4}\imath}$, a real form $\mathfrak{g}=\mathfrak{g}_{\bar{0}}\oplus \mathfrak{g}_{\bar{1}}$
can now be constructed as 
\begin{align}\label{eq: real}
\begin{gathered}
 \mathfrak{g}_{\bar{0}}=\sspan_{\mathbb{R}} \{ X\in \mathfrak{g}_{\bar{0}}^{\mathbb{C}}\,|\, X^*=-X\},\\[.1cm]
 \mathfrak{g}_{\bar{1}}=\sspan_{\mathbb{R}} \{ \sqrt{\imath}X\in \mathfrak{g}_{\bar{1}}^{\mathbb{C}}\,|\, X^*=-X\}
  =\sspan_{\mathbb{R}} \{ X\in \mathfrak{g}_{\bar{1}}^{\mathbb{C}}\,|\, X^*=\imath X\}.
\end{gathered}
\end{align}

Note that $\mathfrak{g}_{\bar{0}}$ is a real form of the complex Lie algebra $\mathfrak{g}_{\bar{0}}^{\mathbb{C}}$. 
We say that a real form of $\mathfrak{g}^{\mathbb{C}}$ is \emph{compact} if $\mathfrak{g}_{\bar{0}}$ is a compact Lie 
algebra; cf.~\cite{AN15,FG23}. It was shown in~\cite{SNR77} that among basic classical Lie superalgebras, only type I Lie 
superalgebras admit star operations. In that case, there are actually two star operations, and they are related by duality.
Throughout this paper, we are only concerned with the so-called type (1) star 
operation; see, e.g.,~\cite{GZ90a, GZ90b} and \secref{sec: examflag}.

\subsection{Lie supergroups}
\label{sec: Lie}

The approach taken when introducing the notion of a supermanifold is typically context-dependent.
In the following, we largely follow~\cite{Kos77,Lei80}.

An archetypal supermanifold is the superspace
\begin{align*}
 \mathbb{R}^{m|n}=(\mathbb{R}^m, \mathcal{O}_{\mathbb{R}^{m|n}}),
\end{align*}
the topological space $\mathbb{R}^m$ equipped with the sheaf of supercommutative superalgebras
\begin{align*}
 \mathcal{O}_{\mathbb{R}^{m|n}}(U)
 =\mathcal{C}_{\mathbb{R}^m}^{\infty}(U) \otimes \Lambda_{\mathbb{R}}[\xi_1, \dots, \xi_n]
\end{align*}
for any open set $U\subseteq \mathbb{R}^m$. Here, $\mathcal{C}_{\mathbb{R}^m}^{\infty}$ denotes the sheaf of smooth (even) 
functions on $\mathbb{R}^m$, and $\Lambda_{\mathbb{R}}[\xi_1, \dots, \xi_n]$ is the Grassman algebra in the odd 
variables $\xi_1, \dots, \xi_n$. 

A (real smooth) \emph{supermanifold} of dimension $m|n$ is a ringed space 
$M=(|M|,\mathcal{O}_M)$, where the underlying topological space $|M|$ is Hausdorff and second-countable, 
and $\mathcal{O}_M= (\mathcal{O}_M)_{\bar{0}}\oplus (\mathcal{O}_M)_{\bar{1}}$ is a sheaf of \emph{superfunctions}, 
i.e., for each $x\in |M|$ there exists an open neighbourhood $V_x\subseteq |M|$ of $x$ and an open set 
$U\subseteq \mathbb{R}^{m|n}$ such that $\mathcal{O}_M(V_x)\cong \mathcal{O}_{\mathbb{R}^{m|n}}(U)$.
Let $\mathcal{I}_M=\langle (\mathcal{O}_M)_{\bar{1}}\rangle$ be the two-sided ideal of $\mathcal{O}_M$ generated by all nilpotent 
sections in $(\mathcal{O}_M)_{\bar{1}}$. The corresponding subspace $M_0=(|M|, \mathcal{O}_M/\mathcal{I}_M)$ 
is called the \emph{reduced manifold} associated with $M$, where $\mathcal{O}_M/\mathcal{I}_M\cong \mathcal{C}^{\infty}_{|M|}$ is 
regarded as the sheaf of smooth functions on $|M|$. Hence, $M_0$ lies in the category of ordinary smooth manifolds. 
The canonical projection 
$\mathrm{ev}: \mathcal{O}_M \to \mathcal{O}_M/\mathcal{I}_M $
is the corresponding \emph{evaluation map}. For any point $p\in |M|$, let $\mathcal{O}_{M,p}$ denote the stalk at $p$. 
The \emph{evaluation at $p$} is then defined as
\begin{align*}
 \mathrm{ev}_p: \mathcal{O}_{M,p} \to \mathbb{R}, \qquad f \mapsto f(p):=\mathrm{ev}(f)(p).
\end{align*}

For each supermanifold $M=(|M|, \mathcal{O}_M)$, the tangent sheaf $\mathcal{T}_M$ is defined by
\begin{align*}
\mathcal{T}_M(U):=\Der(\mathcal{O}_{M}(U))
\end{align*} 
for any open subset $U\subseteq |M|$. Here, $\Der(\mathcal{O}_{M}(U))$ denotes the $\mathcal{O}_M(U)$-module of (left) 
super-derivations $\varphi$, the ones satisfying the super-Leibniz rule 
\begin{align*}
 \varphi(fg)= \varphi(f)g+(-1)^{[\varphi][f]} f\varphi(g), \qquad f,g\in \mathcal{O}_M(U).
\end{align*}
The sections of $\mathcal{T}_M$ are called \emph{vector fields}. For every point $p\in |M|$, the tangent superspace 
$T_pM$ at $p$ is the vector superspace of dimension $m|n$, defined by 
\begin{align*}
 T_pM:= \{\psi \in \mathrm{Hom}_{\mathbb{R}}(\mathcal{O}_{M,p}, \mathbb{R})\,|\,\psi(fg)= \psi(f)g(p)+ (-1)^{[\psi][f]}f(p)\psi(g)\}.	
\end{align*}
Elements of $T_pM$ are called \emph{tangent vectors}. 
For any $p\in U$ and vector field $X\in \mathcal{T}_M(U)$,  the \emph{value} of $X$ at $p$ is the 
tangent vector $X_p: \mathcal{O}_{M,p} \to \mathbb{R}$ defined by $X_p(f):= \mathrm{ev}_p(X(f))$ for all $f\in \mathcal{O}_{M,p}$.
A vector field is only determined by its values at all points if the odd dimension of $M$ is zero.
We also introduce the dual space
\begin{align*}
 \mathcal{T}_M ^*:=\mathrm{Hom}_{\mathcal{O}_M}(\mathcal{T}_M, \mathcal{O}_M).
\end{align*}

A \emph{Lie supergroup} $G$ is a group object in the category of supermanifolds, that is, it is a real smooth supermanifold 
together with three morphisms
\begin{align*}
 m: G\times G\to G, \qquad i: G\to G,\qquad e: \mathbb{R}^{0|0}\to G,
\end{align*}
called \emph{multiplication}, \emph{inverse} and \emph{unit}, respectively, satisfying the usual group axioms that encode 
associativity and the existence of identity and of inverse (see~\cite[Definition 7.1]{CCF11}). The reduced manifold $G_0$ is 
an ordinary Lie group. A Lie supergroup $K$ is a \emph{Lie subsupergroup} of $G$ if the reduced group $K_0$ is a Lie 
subgroup of $G_0$ and the inclusion morphism $K \to G$ is an even immersion. Moreover, $K$ is said to be 
\emph{closed}, respectively \emph{connected}, if $K_0$ is closed, respectively connected, in $G_0$.

A vector field $X$ on $G$ is \emph{left-invariant} if and only if
\begin{align*}
 (1\otimes X)\circ m^*= m^*\circ X,
\end{align*}
where $m^*: \mathcal{O}_{G}\to \mathcal{O}_{G\times G}$ is a morphism on sheaves.
The \emph{Lie superalgebra} $\mathfrak{g}$ of $G$ consists of all left-invariant vector fields on $G$. 
We can identify $\mathfrak{g}$ with the tangent superspace of $G$ at the identity $e$, i.e., the linear map 
$\mathfrak{g} \to T_{e}G$, $X \mapsto X_e$, is an isomorphism of vector superspaces, with inverse map given by 
\begin{align*}
 v \mapsto X_v= (1\otimes v)\circ m^*,
\end{align*} 
where $X_v$ is the left invariant vector field with $(X_v)_e=v$. 
We refer to~\cite{CCF11} for more details.

\subsection{Harish-Chandra pairs}

Complementary to the introduction in \secref{sec: Lie},
the theory of Lie supergroups can be approached equivalently through Harish-Chandra pairs~\cite{Kos77,DM99,CCF11}. 

\begin{definition}\label{def: HCpair}
A (real) \emph{Harish-Chandra pair} $(G_0,\mathfrak{g})$ consists of an ordinary Lie group $G_0$; a Lie superalgebra 
$\mathfrak{g}=\mathfrak{g}_{\bar{0}}\oplus \mathfrak{g}_{\bar{1}}$, where $\mathfrak{g}_{\bar{0}}$ is the Lie algebra of 
$G_0$; and a representation $\sigma$ of $G_0$ on $\mathfrak{g}$ such that
(i) $\sigma$ induces the adjoint representation of $G_0$ on $\mathfrak{g}_{\bar{0}}$, and
(ii) the differential $(d\sigma)_e$ at the identity $e\in G_0$ acts on $\mathfrak{g}$ as the adjoint representation of 
$\mathfrak{g}_{\bar{0}}$.
\end{definition}

Given a Lie supergroup $G$, we can define a Harish-Chandra pair $(G_0, \mathfrak{g})$ by taking $G_0$ to be the 
underlying reduced Lie group of $G$, and $\mathfrak{g}$ to be the Lie superalgebra of $G$. It is well known that the 
category of (real) Harish-Chandra pairs is equivalent to the category of (real) Lie supergroups; see, 
e.g.,~\cite[Theorem~7.4.5]{CCF11}. Thus, a Lie supergroup only depends on the underlying ordinary 
Lie group and Lie superalgebra, together with the compatibility conditions given in \defref{def: HCpair}. 

As an example, the \emph{general linear supergroup} $\GL(m|n)$ is the Lie supergroup associated with the Harish-Chandra 
pair $(\GL(m)\times \GL(n), \mathfrak{gl}(m|n))$, where $\mathfrak{gl}(m|n)$ is the real general linear superalgebra 
comprising the endomorphisms on $\mathbb{R}^{m|n}$. The \emph{special linear supergroup} 
$\SL(m|n)=(\SL(m|n)_0, \mathfrak{sl}(m|n))$ is the Lie subsupergroup of $\GL(m|n)$ with even subsupergroup 
\begin{align*}
 \SL(m|n)_0=\{(A,B)\in \GL(m)\times \GL(n)\,|\,\mathrm{det}(A)=\mathrm{det}(B)>0 \} 
\end{align*}
and Lie superalgebra
\begin{align*}
 \mathfrak{sl}(m|n)=\{X\in \mathfrak{gl}(m|n)\,|\, \Str(X)=0\}.
\end{align*}

A representation of a Harish-Chandra pair $G=(G_0, \mathfrak{g})$ on a vector superspace $V$ consists of a 
representation of $G_0$ on $V$ and one of $\mathfrak{g}$ on $V$, such that the differential of the representation of $G_0$ 
agrees with the restriction to $\mathfrak{g}_0$ of the $\mathfrak{g}$-representation.

\subsection{Homogeneous superspaces}
\label{sec: homosp}

Let $G$ be a Lie supergroup in the category of real smooth supermanifolds and $K$ a closed Lie subsupergroup of $G$. 
Classically, the set $G_0/K_0$ is a topological space with the quotient topology induced by the natural projection 
$\pi_0: G_0\to G_0/K_0$. Moreover, it has a unique classical manifold structure compatible with the action of $G_0$.
As discussed in the following, this classical construction can be extended to the super-setting, thereby giving rise to 
homogeneous superspaces.

Following \cite[Theorem 3.9]{Kos77}, the super-structure sheaf $\mathcal{O}_{G/K}$ on $G_0/K_0$ is defined as follows. 
Let $\mathrm{pr}_1: G\times K \to G$ denote the projection onto the first factor, and $\psi: G\times K \to G$ the morphism 
induced by the multiplication $m$. For any open subset $U\subseteq G_0/K_0$, define the set of $K$-invariant sections by
\begin{align*}
 \mathcal{O}_{G/K}(U):= \{ f\in\mathcal{O}_G(\pi_0^{-1}(U))\,|\, \psi^*f= \mathrm{pr}_1^*f \}.
\end{align*}
This gives rise to a supermanifold,
\begin{align*}
 G/K:=(G_0/K_0, \mathcal{O}_{G/K}),
\end{align*}
referred to as a \emph{homogeneous superspace}. Let $\pi: G \to G/K$ be the corresponding canonical surjective 
homomorphism of supermanifolds, where the sheaf morphism is given by inclusion. The differential of $\pi$ at the identity 
$e\in G$ yields a surjective linear map $\mathfrak{g}\cong T_eG\to T_{K}(G/K)$ with kernel $\mathfrak{k}$, the Lie 
superalgebra of~$K$. This induces a canonical identification: $T_K (G/K)\cong \mathfrak{g}/\mathfrak{k}$. 

A Harish-Chandra pair $G = (G_0, \mathfrak{g})$ acts on a supermanifold $M$ if the reduced Lie group $G_0$ acts on 
$|M|$, and there exists a Lie superalgebra anti-homomorphism from $\mathfrak{g}$ to the Lie superalgebra $\mathcal{T}_M(M)$ 
of vector fields on $M$ such that its restriction to $\mathfrak{g}_{\bar{0}}$ agrees with the differential of the $G_0$-action 
at the identity~\cite[\S 3.8]{DM99}. Indeed, with $\alpha: G\times M\to M$ denoting the Lie supergroup action, 
the anti-homomorphism is given by 
\begin{align}\label{eq: gTMM}
 \mathfrak{g}\to \mathcal{T}_M(M),\qquad
 X\mapsto (X_e\otimes 1)\circ \alpha^*.
\end{align}
A $G$-action is said to be \emph{transitive} if (i) the reduced $G_0$-action is transitive, and (ii) for every $p\in |M|$, 
the even linear map 
\begin{align}\label{eq: trans} 
 \mathfrak{g} \to T_pM, \qquad X\mapsto\mathrm{ev}_p\circ (X_e\otimes 1)\circ \alpha^*,
\end{align}
is surjective~\cite[Section 4.11]{Goe08}.

A supermanifold $M$ is said to be \emph{$G$-homogeneous} if the Lie supergroup $G$ acts on it transitively. 
Every $G$-homogeneous supermanifold $M$ is isomorphic to $G/G_p$ with $G_p$ being the isotropy subsupergroup of 
$G$ at $p\in |M|$~\cite[Proposition 3.10.3]{Kos77}. By definition, the isotropy subsupergroup $G_p$ is the Harish-Chandra 
pair $((G_0)_p, \mathfrak{g}_p)$, where $(G_0)_p$ denotes the ordinary isotropy subgroup of $G_0$ at $p$, and 
$\mathfrak{g}_p$ consists of all left-invariant vector fields $X$ on $G$ such that the infinitesimal action of $X$ at $p$ is 
trivial; see~\cite[Section 4.11]{Goe08} and~\cite[\S 8.4]{CCF11}. 
The \emph{isotropy representation} of $G_p$ on $T_pM$ is equivalent to the adjoint representation on $\mathfrak{g}/\mathfrak{g}_p$ 
via the natural isomorphism $\mathfrak{g}/\mathfrak{g}_p \cong T_pM$ \cite[Proposition 12]{Goe08}.

\section{Curvature}
\label{sec: met_curv}

\subsection{Graded Riemannian metrics and curvatures}
\label{sec: gradedR}

Let $V=V_{\bar{0}}\oplus V_{\bar{1}}$ be a vector superspace over $\mathbb{R}$. A \emph{scalar superproduct} on $V$ is a 
non-degenerate, supersymmetric, even bilinear form $\langle-, -\rangle: V\times V\rightarrow \mathbb{R}$. 
By the non-degeneracy, $\dim(V_{\bar{1}})$ must be even. The supersymmetry implies that $\langle-,-\rangle$ is symmetric on 
$V_{\bar{0}}\times V_{\bar{0}}$, and skew-symmetric on $V_{\bar{1}}\times V_{\bar{1}}$, while the evenness implies that it is  zero on 
$V_{\bar{0}}\times V_{\bar{1}}$ and $V_{\bar{1}}\times V_{\bar{0}}$. 
Moreover, the restriction of $\langle-,-\rangle$ to $V_{\bar{0}}$ may be indefinite. 

\begin{definition} \cite{Goe08} 
A \emph{graded Riemannian metric} on a supermanifold $M=(|M|, \mathcal{O}_M)$ is an $\mathcal{O}_M$-linear morphism of sheaves 
$g: \mathcal{T}_M \otimes \mathcal{T}_M \to \mathcal{O}_M$ such that for any $X,Y\in \mathcal{T}_M(U)$, 
where $U$ is any open subset of $|M|$, 
\begin{enumerate}
\item [(1)] (non-degeneracy) the map $X\mapsto g(X, -)$ is an isomorphism $\mathcal{T}_M \rightarrow \mathcal{T}_M ^*$;
\item [(2)] (supersymmetry) $ g(X,Y) = (-1)^{[X][Y]} g(Y,X)$;
\item [(3)] (evenness) $[g(X, Y)]= [X]+[Y]$. 
\end{enumerate}
\end{definition}

For each $p\in |M|$, the sheaf morphism $g$ induces a scalar superproduct $\langle-,-\rangle_p:=g_p$ 
on the tangent superspace $T_pM$,
\begin{align*}
 \langle-,-\rangle_p: T_pM\times T_pM \to \mathbb{R}.
\end{align*}
If the restriction of $\langle-,-\rangle_p$ to $(T_pM)_{\bar{0}}$ is indefinite, it induces a pseudo-Riemannian metric on the reduced 
manifold $M_0$. We stress that the scalar superproduct $\langle -, -\rangle_p$ does not determine the graded Riemannian metric unless 
the odd dimension of $M$ is zero. A supermanifold $M$ equipped with a graded Riemannian metric $g$ 
is called a \emph{Riemannian supermanifold}, and is denoted by $(M,g)$. 

An affine connection on a supermanifold $M$ is an even morphism 
$\nabla: \mathcal{T}_M \otimes_{\mathbb{R}}\mathcal{T}_M\rightarrow  \mathcal{T}_M$ of sheaves of vector superspaces such that, 
for every open subset $U\subseteq |M|$, 
\begin{align*}
 \nabla_{fX}Y= f \nabla_{X}Y, \qquad \nabla_{X}(fY)= X(f)Y + (-1)^{[X][f]}f \nabla_{X}Z,
\end{align*}
where $f\in \mathcal{O}_{M}(U)$ and $X,Y\in \mathcal{T}_{M}(U)$. 
If $(M, g)$ is a Riemannian supermanifold, there is a unique  affine connection, called the \emph{Levi-Civita connection}, satisfying
\begin{align}
 X g(Y, Z) &= g(\nabla_{X}Y, Z) +(-1)^{[X][Y]}  g(Y, \nabla_{X}Z) , \label{eq: LC1}
 \\[.1cm]
 [X,Y]&=\nabla_{X}Y- (-1)^{[X][Y]} \nabla_{Y}X, \label{eq: LC2}
\end{align}
for all vector fields $X,Y,Z$ on $M$, meaning that the connection is \emph{metric} and \emph{torsion-free}. 
As in the non-super case, the Levi-Civita connection is  determined by the following super analogue of the  Koszul formula.
\begin{proposition} \cite[Theorem 4.2]{MS96}\label{prop: Kos}
Let $(M, g)$ be a Riemannian supermanifold. Then, the Levi-Civita  connection $\nabla$ is  given implicitly by the Koszul formula
\begin{align*}
 2\, g(\nabla_{X}Y, Z) 
 &= X g(Y,Z) +(-1)^{[X][Y]} Y g(X,Z) -(-1)^{([X]+[Y])[Z]} Z g(X,Y)
 \\[.1cm]
 &\quad - g(X, [Y,Z]) -(-1)^{[X][Y]} g(Y, [X,Z]) +(-1)^{([X]+[Y])[Z]}  g(Z, [X,Y]),
\end{align*}
with $X,Y,Z$ any vector fields on $M$.
\end{proposition} 

A \emph{graded Killing vector field} on a Riemannian supermanifold is a vector field $X$ that preserves the metric, that is,
\begin{align}\label{eq: Kfield}
 Xg(Y,Z)= g([X,Y],Z) +(-1)^{[X][Y]} g(Y, [X,Z]),
\end{align}
where $Y,Z$ are any vector fields, and we note that the set of graded Killing vector fields forms a Lie superalgebra. 
Using \eqref{eq: LC1} and \eqref{eq: LC2}, a vector field $X$ is seen to be a graded Killing vector field if and only if 
\begin{align}\label{eq: KVec}
  g(\nabla_YX, Z) +(-1)^{[X][Z]} g(Y, \nabla_ZX) =0
\end{align}
for all vector fields $Y,Z$. For graded Killing vector fields, the Koszul formula in \propref{prop: Kos} readily simplifies as follows. 
\begin{lemma}\label{lem: Kos}
	Let $X,Y,Z$ be graded Killing vector fields on a Riemannian supermanifold $(M,g)$. Then,
	\begin{align*}
	 2 g(\nabla_XY,Z)= g(X,[Y,Z])+X g(Y,Z)
	   = g(X,[Y,Z])+g([X,Y],Z)+(-1)^{[X][Y]} g(Y, [X,Z]).
	\end{align*}
\end{lemma}

Following~\cite{DeW92}, the \emph{Riemann curvature tensor} on a Riemannian supermanifold $(M, g)$  is defined by
\begin{align*}
 R(X,Y)Z:= \nabla_{[X,Y]}Z- \nabla_{X}\nabla_{Y} Z+ (-1)^{[X][Y]} \nabla_{Y}\nabla_{X}Z,
\end{align*}
where $X,Y,Z\in \mathcal{T}_M(U)$ are any vector fields and $U$ is any open subset of $|M|$. 
Furthermore, the \emph{Ricci curvature tensor} is defined as
\begin{align*} 
 \mathrm{Ric}(Y,Z)= \mathrm{Str}\big(X \mapsto (-1)^{[X][Z]} R(Y,X)Z\big),
\end{align*}
where the supertrace $\mathrm{Str}$ is taken over the $\mathcal{O}_M(U)$-linear superspace $\mathcal{T}_{M}(U)$. 
Similarly, the \emph{scalar curvature} is defined as the  supertrace of the Ricci curvature tensor with respect to the metric $g$.  
  
In the remainder of this paper, we will work with these curvature tensors for homogeneous superspaces.
To present the curvature tensors evaluated at a point $p\in |M|$, let $\{X_{1},\dots,X_d\}$ be an ordered homogeneous basis for 
the tangent superspace $T_p(M)$, with $\{\widebar{X}_1,\ldots,\widebar{X}_d\}$ the corresponding right dual basis. 
For a vector $X\in T_p(M)$, we thus have
\begin{align}\label{XXX}
 X=\sum_{i=1}^d\,\langle X, \widebar{X}_i\rangle_p X_i= \sum_{i=1}^d\,\langle X_i, X\rangle_p \widebar{X}_i. 
\end{align}
At $p\in |M|$, the \emph{Ricci curvature} is given by 
\begin{align}\label{RicR}
 \Ric(X_i, X_k)_p:=\sum_{j=1}^d (-1)^{[X_j]+[X_j][X_k]} \langle R(X_i,X_j)X_k, \widebar{X}_j\rangle_p.
\end{align}
As $\Ric(-,-)_p$ is even, we have $\Ric(X_i,X_j)_p=0$ if $[X_i]\neq[X_j]$. 
Finally, the scalar curvature at $p$ is 
\begin{align}\label{eq: scalar}
 S(p):= \sum_{i=1}^d (-1)^{[X_i]} \Ric(X_i, \widebar{X}_i)_p.
\end{align}

\subsection{Invariant metrics}\label{sec: RRs}

Let $G=(G_0, \mathfrak{g})$ be a Lie supergroup. A graded Riemannian metric $g$ 
on $G$ is called \emph{left-invariant} if $g(X,Y)$ is a constant function on $G$ for all left-invariant vector fields 
$X,Y$. Recall that we  identify the Lie superalgebra $\mathfrak{g}$ of left-invariant vector fields with the tangent space $T_eG$ at the 
identity $e\in G$ via the isomorphism $X\mapsto X_e$. Similarly to the non-super case, the left-invariant graded Riemannian metrics $g$ 
on $G$ are in natural bijection with 
the scalar superproducts $\langle -,-\rangle:=g_e$  on $\mathfrak{g}\cong T_eG$ \cite[Section 4.9]{Goe08}. 

\begin{lemma}\label{lem: leftinv}
Let $g$ be a left-invariant graded Riemannian metric on $(G_0, \mathfrak{g})$. Then,
\begin{align*}
 2\,g(\nabla_{X}Y, Z) 
   = g([X,Y],Z)-g(X, [Y,Z]) -(-1)^{[X][Y]}g(Y, [X,Z]),\qquad X,Y,Z\in\mathfrak{g}.
\end{align*}
\end{lemma}
\begin{proof}
Since $g$ is left-invariant, $Xg(Y,Z)=0$ for all $X,Y,Z\in \mathfrak{g}$, so the result follows from \propref{prop: Kos}.
\end{proof}

A scalar superproduct $\langle-,-\rangle $ on a Lie superalgebra $\mathfrak{g}\cong T_eG$ is called \emph{$\Ad_G$-invariant} if it is both 
$\Ad_{G_0}$- and $\ad_{\mathfrak{g}}$-invariant, i.e., for every $X, Y, Z\in \mathfrak{g}$ and every $a\in G_0$, we have
\begin{align*}
	\langle \Ad_a(X), \Ad_a(Y) \rangle =\langle X, Y\rangle,\qquad
	\langle [X,Y],Z\rangle = \langle X, [Y,Z]\rangle.
\end{align*}
Note that if $G$ is connected, then $\ad_{\mathfrak{g}}$-invariance implies $\Ad_{G_0}$-invariance.

\begin{definition}\label{def: G-inv}
A graded Riemannian metric on a $G$-homogeneous superspace $M$ is \emph{$G$-invariant} if (i) every element 
of $G_0$ acts on $M$ by isometries, and (ii) the image of the
anti-homomorphism \eqref{eq: gTMM} is contained in the Lie superalgebra of graded Killing vector fields.
\end{definition}
\begin{theorem}\label{thm: invmetric}
\cite[Theorem 3]{Goe08}
Let $G=(G_0, \mathfrak{g})$ be a Lie supergroup and $K=(K_0, \mathfrak{k})$ a closed Lie subsupergroup of $G$. 
The map $g \mapsto g_{K}$ is a bijective correspondence between $G$-invariant graded Riemannian metrics on $G/K$ 
and $\Ad_K$-invariant scalar superproducts on the tangent superspace $T_{K}(G/K)$ at the base point $K$. 
\end{theorem}

\thmref{thm: invmetric} implies that every $G$-invariant graded Riemannian metric $g$ on $G/K$ is uniquely determined by its evaluation 
$g_K$ at the base point $K$. Using  the method in the proof of this result in~\cite{Goe08}, one shows that $\mathrm{Ric}(g)$ on $G/K$ is 
uniquely determined by  $\mathrm{Ric}(g)_K$, even though the Ricci tensor may be degenerate. Here, $\mathrm{Ric}(g)_K$ is a 
supersymmetric, $\mathrm{Ad}_K$-invariant even bilinear form on~$T_K(G/K)$.  More generally \cite[Theorem 4.16]{San10},
there is a bijective correspondence between $G$-invariant tensor fields of type $(r,s)$ on $G/K$ and isotropy-invariant tensors 
of type $(r,s)$ on $T_{K}(G/K)$, via evaluation at~$K$. In our case, both $g$ and $\mathrm{Ric}(g)$ are $G$-invariant tensor fields 
of type $(0,2)$. Thus, to solve the Einstein equation $\mathrm{Ric}(g) = c g$
on a homogeneous Riemannian supermanifold $(G/K,g)$, it is  necessary and sufficient to solve the corresponding equation at the 
base point $K$, that is, to find $g_K$ satisfying $\mathrm{Ric}(g)_K = c g_K$.
The bijective correspondence ensures that a solution at $K$ uniquely extends to a $G$-invariant solution on the entire supermanifold.

\begin{example}\label{sec: sl11}
In the spirit of~\cite{Mil76} on left-invariant metrics on Lie groups, we now compute the Ricci curvature of the 
three-dimensional Lie supergroup $\SL(1|1)$, and demonstrate that every left-invariant graded Riemannian metric on 
$\SL(1|1)$ is an Einstein metric. Although $\SL(1|1)$ is not part of our $A$-type classification in \secref{subsubsec: typeA} 
and \secref{sec: flagA}, we find that the analysis below serves to illustrate some of the key notions.

The Lie superalgebra $\mathfrak{sl}(1|1)$ of the Lie supergroup $\SL(1|1)$ has basis
\begin{align*}
 X_1=\begin{pmatrix}
  	1&0\\ 0&1
  \end{pmatrix},\qquad 
  X_2= \begin{pmatrix}
  	0&1\\ 0&0
  \end{pmatrix},\qquad
  X_3=\begin{pmatrix}
  	0&0\\ 1&0
  \end{pmatrix},
\end{align*}
with graded commutation relations
\begin{align*}
 [X_1,X_2]=[X_1,X_3]=0, \qquad [X_2,X_3]=X_1,
\end{align*}
where $X_1$ is even and $X_2,X_3$ odd. Clearly, the centre of $\mathfrak{sl}(1|1)$ is spanned by $X_1$. 

Let $g$ be a left-invariant graded Riemannian metric on $\SL(1|1)$. 
Then, $g$ can be identified with a scalar superproduct $\langle-, -\rangle$ on $\mathfrak{sl}(1|1)$, which is a non-degenerate 
supersymmetric even bilinear form. This implies that 
$\langle X_1, X_2\rangle =\langle X_1, X_3\rangle=\langle X_2, X_2\rangle=\langle X_3, X_3\rangle= 0$ and
$\langle X_2, X_3\rangle =-\langle X_3, X_2\rangle$, so relative to the ordered basis $\{X_1,X_2,X_3\}$, we may assume 
that the metric $g$ has the matrix form
\begin{align*}
	(g(X_i,X_k))_{3\times 3}=\begin{pmatrix}
	x_1&0&0\\ 
	0&0&x_2\\ 
	0&-x_2&0
\end{pmatrix}, \qquad x_1,x_2\in\mathbb{R}^\times.
\end{align*}
The corresponding right dual basis is thus given by 
$\{\widebar{X}_1=\frac{1}{x_1}X_1,\widebar{X}_2=\frac{1}{x_2}X_3,\widebar{X}_3=-\frac{1}{x_2}X_2\}$.

Using \lemref{lem: leftinv}, we find
$\nabla_{X_1}X_1=\nabla_{X_2}X_2=\nabla_{X_3}X_3=0$ and
\begin{align*}
 \nabla_{X_1}X_2=\nabla_{X_2}X_1=-\frac{x_1}{2x_2}X_2,\qquad
 \nabla_{X_1}X_3=\nabla_{X_3}X_1=\frac{x_1}{2x_2}X_3,\qquad 
 \nabla_{X_2}X_3=\nabla_{X_3}X_2=\frac{1}{2}X_1,
\end{align*}
from which it follows that the Ricci curvature at the identity $e\in\SL(1|1)$ is given by
\begin{align*}
 \Ric(X_i,X_k)=\frac{x_1}{2x_2^2}\,g(X_i,X_k).
\end{align*}
We thus conclude that every left-invariant graded Riemannian metric on $\SL(1|1)$ is an Einstein metric.
\end{example}

\subsection{Curvature formulas}
\label{subsec: curvs}

Let $G=(G_0,\mathfrak{g})$ be a connected Lie supergroup, with $\mathfrak{g}$ a compact real form of a basic classical 
Lie superalgebra~$\mathfrak{g}^{\mathbb{C}}$. As $\mathfrak{g}^{\mathbb{C}}$ is basic classical, there exists a non-degenerate 
supersymmetric even bilinear form  $Q$ on 
$\mathfrak{g}^{\mathbb{C}}$, with  $\ad_{\mathfrak{g}^{\mathbb{C}}}$-invariance given by
\begin{align*}
 Q([X,Y],Z)=Q(X,[Y,Z]), \qquad X,Y,Z\in \mathfrak{g}^{\mathbb{C}}.
\end{align*}
Let $K=(K_0,\mathfrak{k})$ be a closed connected Lie subsupergroup of $G$,
and let the homogeneous superspace $M=G/K$ be equipped with a $G$-invariant graded Riemannian metric 
$g$. Our goal here is to give explicit formulas for the various curvatures of this metric at the point $K$ of $M$. 
As discussed following \thmref{thm: invmetric}, these formulas determine the full curvature tensors.

Since $M$ is equipped with a $G$-invariant metric, the image under the anti-homomorphism \eqref{eq: gTMM}
of the Lie superalgebra $\mathfrak{g}$ of left-invariant vector fields is contained in 
the Lie superalgebra of graded Killing fields, cf.~\defref{def: G-inv}. As $\mathfrak{g}$ is simple, the anti-homomorphism,
mapping $X\mapsto X^{\dagger}$, is injective, so we can identity the left-invariant vector field $X\in \mathfrak{g}$ 
with the Killing vector field $X^{\dagger}$ on $(M,g)$.
With $[-,-]$ denoting the Lie superbracket of vector fields on $M$, and $[-,-]_{\mathfrak{g}}$ the Lie superbracket 
of~$\mathfrak{g}$, the anti-homomorphicity then implies that
\begin{align}\label{eq: opp}
 [X^\dagger, Y^\dagger]=-[X,Y]_{\mathfrak{g}}^\dagger,\qquad X,Y\in\mathfrak{g}.
\end{align}

We now identify the Lie subsuperlagebra $\mathfrak{k}$ of $\mathfrak{g}$ with the subsuperalgebra of the Killing vector fields 
$X^{\dagger}$ which vanish at the base point~$K$, and choose once and for all a $Q$-orthogonal complement $\mathfrak{m}$ 
of $\mathfrak{k}$. We thus have the vector superspace decomposition $\mathfrak{g}=\mathfrak{k}\,\oplus\,\mathfrak{m}$ 
and ditto isomorphisms $\mathfrak{g}/\mathfrak{k}\cong \mathfrak{m}\cong T_KM $, where we identify $\mathfrak{m}$ with $T_KM$ 
by evaluating the corresponding graded Killing vector field at the point $K$.  In this way, the isotropy representation of $K$ in $T_KM$ 
is identified with the restriction of the $\Ad_G$-representation of $K$ to~$\mathfrak{m}$. 

We denote by $\langle-,-\rangle$ the scalar superproduct on $\mathfrak{m}$ induced by $g_K$ through the identification of 
$\mathfrak{m}$ with~$T_KM$. By \thmref{thm: invmetric}, $\langle-,-\rangle$ is an $\mathrm{Ad}_K$-invariant, supersymmetric, 
even bilinear form on $\mathfrak{m}$. As $K$ is connected,  the $\mathrm{Ad}_K$-invariance is equivalent to 
$\ad_{\mathfrak{k}}$-invariance, that is,
\begin{align*}
 \langle X,[Y,Z]_{\mathfrak{m}}\rangle=\langle [X,Y]_{\mathfrak{m}},Z\rangle,
\end{align*}
for all $X,Z\in\mathfrak{m}$ and $Y\in\mathfrak{k}$. 
We extend the superproduct $\langle-,-\rangle$ on $\mathfrak{m}$ to an $\Ad_K$-invariant superproduct 
on $\mathfrak{g}$ by choosing on $\mathfrak{k}$ an $\Ad_{K}$-invariant superproduct, such as $Q|_{\mathfrak{k}}$,
and setting $\langle\mathfrak{k},\mathfrak{m}\rangle=0$. In the following, we primarily use $X,Y,Z,\ldots$ to denote elements of 
$\mathfrak{g}$, despite the use of the same symbols to represent vector fields in \secref{sec: gradedR} and \secref{sec: RRs}. 
We hope this practice will not cause any confusion.

The next result evaluates the Levi-Civita connection on $(M,g)$ at the base point $K$. We denote by $Z_{\mathfrak{k}}$, 
respectively $Z_{\mathfrak{m}}$, the $\mathfrak{k}$- and $\mathfrak{m}$-components of $Z\in\mathfrak{g}$, so
\begin{align}\label{gkm}
 [-,-]_{\mathfrak{g}}= [-,-]_{\mathfrak{k}}+ [-,-]_{\mathfrak{m}}.
\end{align}
We follow~\cite[Proposition 7.28]{Bes87} on classical homogeneous spaces and introduce the even bilinear map
$U: \mathfrak{m}\times \mathfrak{m}\to \mathfrak{m}$ given implicitly by
\begin{align}\label{eq: defU}
 \langle U(X,Y), Z\rangle 
  = -\langle X, [Y,Z]_{\mathfrak{m}}\rangle -(-1)^{[X][Y]} \langle Y, [X,Z]_{\mathfrak{m}}\rangle,\qquad X,Y,Z\in\mathfrak{m}.
\end{align}
It readily follows that $U(X,Y)=(-1)^{[X][Y]} U(Y,X)$ for any homogeneous elements $X,Y\in \mathfrak{m}$.
\begin{remark}
To reduce the number of fractions in the ensuing expressions, especially in \propref{prop: Rie} and \secref{sec: pfmain},
we have opted to scale $U$ by $2$ relative to the similar function in~\cite{Bes87}.
\end{remark}
\begin{proposition}\label{prop: con}
Let $M=G/K$ be a homogeneous superspace with a $G$-invariant graded Riemannian metric $g$,
and let $X^\dagger, Y^\dagger$ be graded Killing vector fields associated with $X,Y\in \mathfrak{m}$. Then,
\begin{align*} 
 2(\nabla_{X^\dagger} Y^\dagger)_K=-[X,Y]_{\mathfrak{m}}+ U(X,Y).
\end{align*}
\end{proposition}
\begin{proof}
By \lemref{lem: Kos} and $\eqnref{eq: opp}$, for any $Z\in \mathfrak{m}$ (with associated graded Killing vector field 
$Z^\dagger$), we have 
\begin{align*}
    2g(\nabla_{X^\dagger}Y^\dagger, Z^\dagger) 
    =-g(X^\dagger, [Y,Z]_{\mathfrak{g}}^\dagger)
      -g([X,Y]^\dagger_{\mathfrak{g}}, Z^\dagger)
      -(-1)^{[X^{\dagger}][Y^{\dagger}]}g(Y^\dagger, [X,Z]^{\dagger}_{\mathfrak{g}}).
\end{align*}
Evaluating at $K$ and using \eqref{eq: defU}, having identified $\mathfrak{m}$ with $T_KM$, we obtain
\begin{align*}
 2\,\langle (\nabla_{X^\dagger}Y^\dagger)_K,Z\rangle= - \langle [X,Y]_{\mathfrak{m}}, Z\rangle +\langle U(X,Y),Z\rangle.
\end{align*}
This completes the proof.
\end{proof}

We note that the restriction of $\langle-,-\rangle$ to $\mathfrak{m}_{\bar{0}}\times \mathfrak{m}_{\bar{0}}$ is symmetric 
but may be indefinite, that the restriction to $\mathfrak{m}_{\bar{1}}\times \mathfrak{m}_{\bar{1}}$ is skew-symmetric, 
that $\dim(\mathfrak{m}_{\bar{1}})$ is even, and that
$\langle\mathfrak{m}_{\bar{0}},\mathfrak{m}_{\bar{1}}\rangle=\langle\mathfrak{m}_{\bar{1}},\mathfrak{m}_{\bar{0}}\rangle=0$.
It follows that  $\mathfrak{m}$ admits an ordered \emph{$g$-normalised homogeneous basis} 
$\{X_1, \dots, X_c, X_{c+1},\dots,X_d\}$, where $[X_i]=\bar{0}$ for $i\leq c$ and $[X_i]=\bar{1}$ for $i>c$, such that
the corresponding matrix representation of $\langle-, -\rangle$ is of the form
 \begin{align}\label{eq: normbasis}
     (\langle X_i, X_j\rangle)_{d\times d}=\begin{pmatrix}
      I_p& 0& 0& 0\\ 
      0& -I_q& 0&0\\ 
      0&0&0&I_r\\ 
      0&0&-I_r&0
    \end{pmatrix},
 \end{align}
where $p+q=c$ and $2r=d-c$. The right dual basis elements are given by
\begin{align}\label{eq: dualX}
\begin{array}{rllll}
 &\widebar{X}_i= X_i, \quad &1\leq i\leq p; \qquad\  &\widebar{X}_{i}=X_{i+r}, \quad &c< i\leq c+r,
  \\[.2cm] 
  &\widebar{X}_i=-X_i, \quad  &p< i\leq c; \qquad\  &\widebar{X}_{i}=-X_{i-r},  \quad &d-r<i\leq d.
\end{array}
\end{align}
We have the following useful observation.

\begin{lemma}\label{lem: fXX}
Let $f$ be an even bilinear map on $\mathfrak{m}\times\mathfrak{m}$. Then,
\begin{align*}
 \sum_{j=1}^df(\widebar{X}_j,X_j)=\sum_{j=1}^d(-1)^{[X_j]}f(X_j,\widebar{X}_j).
\end{align*}
\end{lemma}
\begin{proof}
Using \eqref{eq: dualX}, we have
\begin{align*}
 \sum_{j=1}^df(\widebar{X}_j,X_j)
 &=\sum_{j=1}^pf(X_j,X_j)
   +\sum_{j=p+1}^cf(-X_j,X_j)
   +\sum_{j=c+1}^{c+r}f(X_{j+r},X_j)
   +\sum_{j=c+r+1}^df(-X_{j-r},X_j)
 \\[.1cm]
 &=\sum_{j=1}^pf(X_j,X_j)
   +\sum_{j=p+1}^cf(X_j,-X_j)
   -\sum_{j=c+r+1}^{d}f(X_j,-X_{j-r})
   -\sum_{j=c+1}^{d-r}f(X_j,X_{j+r}),
\end{align*}
which is recognised as the righthand side of the stated relation.
\end{proof}
\begin{remark}
Using the super-skew-symmetry of $[-,-]_{\mathfrak{m}}$, it readily follows from \lemref{lem: fXX} that
\begin{align}\label{XX0}
 \sum_{j=1}^d\,[\widebar{X}_j,X_j]_{\mathfrak{m}}=\sum_{j=1}^d(-1)^{[X_j]}[X_j, \widebar{X}_j]_{\mathfrak{m}}=0.
\end{align}
Using the $\ad_{\mathfrak{k}}$-invariance of $\langle-,-\rangle$, it also follows that, for every $Y\in\mathfrak{k}_{\bar{0}}$,
\begin{align}\label{YXX}
 \sum_{j=1}^d\,\langle[Y,\widebar{X}_j]_{\mathfrak{m}},X_j\rangle=0.
\end{align}
\end{remark}

We are now in a position to state and prove our main curvature formulas: \propref{prop: Rie}, \thmref{thm: Ric} 
and \corref{coro: scur}. The proof of the expression for the Ricci curvature given in \thmref{thm: Ric} is lengthy and deferred 
until \secref{sec: pfmain}. We emphasise that the formulas below only give the evaluations of the corresponding curvature tensors at the 
base point~$K$. However, by~\cite[Theorem~4.16]{San10} (cf.~the discussion following \thmref{thm: invmetric}), 
this data determines the full curvature tensors.
\begin{proposition}\label{prop: Rie}
Let $M=G/K$ be a homogeneous superspace with $G$-invariant graded Riemannian metric $g$, where 
$G=(G_0, \mathfrak{g})$ is a connected Lie supergroup, $\mathfrak{g}$ a compact real form of a basic classical Lie 
superalgebra, and $K=(K_0,\mathfrak{k})$ a connected closed Lie subsupergroup of $G$. Let $\mathfrak{m}$ be a 
$Q$-orthogonal complement of $\mathfrak{k}$, where $Q$ is an $\ad_{\mathfrak{g}^{\mathbb{C}}}$-invariant 
non-degenerate supersymmetric even bilinear form on $\mathfrak{g}^{\mathbb{C}}$, and let 
$X_i,X_j,X_k,X_\ell\in\mathfrak{m}$. Then, at the base point $K$, 
\begin{align*}
   4\, \langle (R &(X_i, X_j)X_k, X_{\ell}\rangle
   =-2\,\langle [X_i,X_j]_{\mathfrak{m}}, [X_k, X_{\ell}]_{\mathfrak{m}}\rangle 
    +(-1)^{[X_i]([X_j]+[X_k])} \langle [X_j, X_k]_{\mathfrak{m}}, [X_i, X_{\ell}]_{\mathfrak{m}}\rangle
    \\[.15cm] 
   &-(-1)^{[X_j][X_k]} \langle [X_i,X_k]_{\mathfrak{m}}, [X_j,X_{\ell}]_{\mathfrak{m}}\rangle
    -\langle X_i, [[X_j,X_k]_{\mathfrak{m}}, X_{\ell}]_{\mathfrak{m}}\rangle
    +\langle [X_i,[X_j, X_k]_{\mathfrak{m}}]_{\mathfrak{m}},X_\ell\rangle
    \\[.15cm] 
   &+(-1)^{[X_i][X_j]} \langle X_j, [[X_i,X_k]_{\mathfrak{m}}, X_{\ell}]_{\mathfrak{m}}\rangle
     +(-1)^{[X_j][X_k]}\langle [[X_i, X_k]_{\mathfrak{m}}, X_j]_{\mathfrak{m}},X_\ell\rangle
    \\[.15cm] 
   &+(-1)^{[X_i]([X_j]+[X_k])} \langle X_j, [X_k, [X_i, X_{\ell}]_{\mathfrak{m}}]_{\mathfrak{m}}\rangle
     +(-1)^{[X_i][X_j]+[X_k][X_\ell]} \langle [X_j, [X_i, X_{\ell}]_{\mathfrak{m}}]_{\mathfrak{m}},X_k\rangle
   \\[.15cm]
   &-(-1)^{[X_j][X_k]} \langle X_i, [X_k, [X_j, X_{\ell}]_{\mathfrak{m}}]_{\mathfrak{m}}\rangle 
     -(-1)^{[X_k][X_\ell]} \langle [X_i, [X_j, X_{\ell}]_{\mathfrak{m}}]_{\mathfrak{m}},X_k\rangle
   \\[.15cm]
   &
     +2\,\langle X_i, [X_j, [X_k, X_{\ell}]_{\mathfrak{g}}]_{\mathfrak{m}}\rangle 
    -2(-1)^{[X_i][X_j]} \langle X_j, [X_i, [X_k, X_{\ell}]_{\mathfrak{g}}]_{\mathfrak{m}}\rangle
   \\[.15cm] 
   &+(-1)^{[X_i]([X_j]+[X_k])} \langle U(X_j, X_k), U(X_i, X_{\ell})\rangle 
     -(-1)^{[X_j][X_{k}]} \langle U(X_i, X_k), U(X_j, X_{\ell}) \rangle.
\end{align*}
\end{proposition}
\begin{proof}
With the shorthand notation $\nabla_i=\nabla_{X_i}$ and $[i]=[X_i]$ for each $i$, and
applying \eqref{eq: LC1} and \eqref{eq: KVec} to the defining expression for the Riemannian curvature, we get
\begin{align*} 
   \langle R(X_i, X_j)X_k, X_{\ell}\rangle 
   &= \langle \nabla_{[X_i,X_j]}X_k, X_{\ell}\rangle- \langle \nabla_{i}\nabla_{j} X_k, X_{\ell}\rangle 
   +(-1)^{[i][j]} \langle \nabla_{j}\nabla_{i}X_k, X_{\ell}\rangle
   \\[.1cm]  
   &= -(-1)^{[k][\ell]} \langle [X_i,X_j], \nabla_{\ell} X_k\rangle 
    - X_i\,\langle \nabla_{j}X_{k}, X_{\ell}\rangle + (-1)^{[i]([j]+[k])} \langle \nabla_{j}X_k, \nabla_{i}X_{\ell}\rangle
   \\[.1cm] 
   &\quad + (-1)^{[i][j]} X_j\,\langle \nabla_{i}X_k,X_{\ell} \rangle -(-1)^{[j][k]} \langle \nabla_{i}X_k, \nabla_{j}X_{\ell}\rangle.
\end{align*}
Applying \eqnref{eq: LC2} and \lemref{lem: Kos} to this expression, we obtain
\begin{align*} 
   \langle R(X_i, X_j)X_k, X_{\ell}\rangle 
   =& -(-1)^{[k][\ell]}\langle \nabla_{i}X_j, \nabla_{\ell}X_k\rangle 
     +(-1)^{[i][j]+[k][\ell]} \langle \nabla_{j}X_i, \nabla_{\ell}X_k\rangle
   \\[.1cm] 
     & +(-1)^{[i]([j]+[k])}\langle \nabla_{j}X_k, \nabla_{i}X_{\ell}\rangle -(-1)^{[j][k]}\langle \nabla_{i}X_k, \nabla_{j}X_\ell\rangle
     \\[.1cm] 
     & - \tfrac{1}{2}X_i\,\langle X_j, [X_k, X_{\ell}]\rangle 
        -\tfrac{1}{2}(-1)^{[j][k]}X_i\,\langle X_k, [X_j,X_{\ell}]\rangle
     \\[.1cm] 
     & -\tfrac{1}{2}(-1)^{[\ell]([j]+[k])} X_i\,\langle X_{\ell}, [X_j, X_{k}]\rangle 
       + \tfrac{1}{2}(-1)^{[i][j]} X_j\,\langle X_{i}, [X_k, X_{\ell}]\rangle 
    \\[.1cm] 
     &+ \tfrac{1}{2}(-1)^{[i]([j]+[k])} X_j\,\langle X_k, [X_i, X_{\ell}]\rangle
        + \tfrac{1}{2}(-1)^{[\ell]([i]+[k])+[i][j]}X_j\,\langle X_{\ell}, [X_i,X_{k}]\rangle.
\end{align*}
Applying \propref{prop: con} to each of the first four terms, and \eqref{eq: Kfield} to each of the last six terms, yields
\begin{align*}
   4\,\langle R(X_i, X_j)X_k, X_{\ell}\rangle 
   &=2\,\langle [X_i,X_j]_{\mathfrak{m}}, U(X_k, X_{\ell})\rangle 
     -(-1)^{[i]([j]+[k])}\langle [X_j,X_k]_{\mathfrak{m}}, U(X_i,X_{\ell})\rangle
   \\[.1cm]
   &\quad-(-1)^{[i]([j]+[k])}\langle U(X_j, X_k), [X_i,X_{\ell}]_{\mathfrak{m}}\rangle 
     +(-1)^{[j][k]} \langle [X_i, X_k]_{\mathfrak{m}}, U(X_j, X_{\ell})\rangle
   \\[.1cm] 
   &\quad +(-1)^{[j][k]}  \langle U(X_i, X_k), [X_j, X_{\ell}]_{\mathfrak{m}}\rangle 
     -2\,\langle [X_i,X_j]_{\mathfrak{m}}, [X_k,X_{\ell}]_{\mathfrak{m}}\rangle
   \\[.1cm] 
   &\quad +(-1)^{[i]([j]+[k])} \langle [X_j,X_k]_{\mathfrak{m}}, [X_i,X_{\ell}]_{\mathfrak{m}}\rangle    
     + (-1)^{[i]([j]+[k])} \langle U(X_j,X_k), U(X_i, X_{\ell}) \rangle
   \\[.1cm]
   &\quad -(-1)^{[j][k]} \langle [X_i,X_k]_{\mathfrak{m}}, [X_j, X_{\ell}]_{\mathfrak{m}}\rangle   
     -(-1)^{[j][k]} \langle U(X_i, X_k), U(X_j, X_{\ell})\rangle
   \\[.1cm] 
   &\quad -2(-1)^{([i]+[j])[k]}\langle X_k, [[X_i,X_j]_{\mathfrak{g}}, X_{\ell}]_{\mathfrak{m}}\rangle  
      -2(-1)^{[\ell]} \langle X_{\ell}, [[X_i,X_j]_{\mathfrak{g}}, X_k]_{\mathfrak{m}}\rangle
   \\[.1cm] 
   &\quad -2(-1)^{[i][j]} \langle X_j, [X_i, [X_k,X_{\ell}]_{\mathfrak{g}}]_{\mathfrak{m}}\rangle 
       +2\,\langle X_i, [X_j, [X_k,X_{\ell}]_{\mathfrak{g}}]_{\mathfrak{m}}\rangle.
\end{align*}
The desired formula now follows by applying \eqref{eq: defU} to the first five terms, and using
\begin{align*}
 (-1)^{[X_k][X_\ell]} \langle [[X_i,X_j]_{\mathfrak{k}}, X_{\ell}]_{\mathfrak{m}},X_k\rangle
  &=-\langle [[X_i,X_j]_{\mathfrak{k}}, X_k]_{\mathfrak{m}},X_\ell\rangle,
\end{align*}
which follows from the $\ad_{\mathfrak{k}}$-invariance.
\end{proof}
\begin{theorem}\label{thm: Ric}
Let $M=G/K$ be a homogeneous superspace with $G$-invariant graded Riemannian metric $g$,
where $G=(G_0, \mathfrak{g})$ is a connected Lie supergroup, $\mathfrak{g}$ a compact real form of a basic classical Lie 
superalgebra, and $K$ a connected closed Lie subsupergroup of $G$. Let $\{X_1,\ldots,X_d\}$ be a $g$-normalised 
homogeneous basis for $\mathfrak{m}=T_K(M)$, with right dual basis $\{\widebar{X}_1,\ldots,\widebar{X}_d\}$. 
Then, at the base point $K$, for every $i\in\{1,\ldots,d\}$,
\begin{align*}
  \Ric(X_i, \widebar{X}_i)
  =-\frac{1}{2}B(X_i, \widebar{X}_i)
   +\frac{1}{2}\sum_{j=1}^d\,\langle [X_i,\widebar{X}_j]_{\mathfrak{m}}, [X_j,\widebar{X}_i]_{\mathfrak{m}}\rangle 
   -\frac{1}{4} \sum_{j,k=1}^d\,\langle X_i, [\widebar{X}_j,\widebar{X}_k]_{\mathfrak{m}}\rangle\,
      \langle [X_k,X_j]_{\mathfrak{m}}, \widebar{X}_i\rangle,
\end{align*}
where $B$ denotes the Killing form on $\mathfrak{g}$.
\end{theorem}
\begin{corollary}\label{coro: scur}
Let the notation be as in \thmref{thm: Ric}. Then, at the base point $K$, the scalar curvature of $g$ satisfies
\begin{align*}
     S(K)= -\frac{1}{2} \sum_{i=1}^dB(\widebar{X}_i,X_i) 
     +\frac{1}{4} \sum_{i,j=1}^d \langle [ \widebar{X}_i, \widebar{X}_j]_{\mathfrak{m}}, [X_j,X_i]_{\mathfrak{m}}\rangle. 
\end{align*}
\end{corollary}
\begin{proof}
First apply the Ricci curvature formula from \thmref{thm: Ric} to \eqref{eq: scalar}. Use \lemref{lem: fXX} and apply
$\sum_{i=1}^d\langle [\widebar{X}_j,\widebar{X}_k]_{\mathfrak{m}},\widebar{X}_i\rangle\,X_i=[\widebar{X}_j,\widebar{X}_k]_{\mathfrak{m}}$ 
to the ensuing triple sum. Change summation variables to obtain the desired formula.
\end{proof}

\subsection{Proof of the Ricci curvature formula}
\label{sec: pfmain}

This section is devoted to proving the Ricci curvature formula given in \thmref{thm: Ric}. Although our general strategy is 
similar to that in~\cite[Chapter 7]{Bes87}, the super-setting complicates matters considerably.
Throughout, $\{X_1,\ldots,X_d\}$ denotes a $g$-normalised homogeneous basis for $\mathfrak{m}=T_K(M)$, 
with right dual basis $\{\widebar{X}_1,\ldots,\widebar{X}_d\}$, where $[\widebar{X}_i]=[X_i]$ for all $i$.
As in the proof of \propref{prop: Rie}, we will use the shorthand notation $\nabla_i=\nabla_{X_i}$ and $[i]=[X_i]$.
Since the extended superproduct $\langle-, -\rangle:\mathfrak{g}\times\mathfrak{g}\to\mathbb{R}$ is an even bilinear map, 
an expression like $\langle Y_1,[Y_2,Y_3]\rangle$, with $Y_1,Y_2,Y_3\in\mathfrak{g}$, is zero if $[Y_1]\neq[Y_2]+[Y_3]$.
We first establish a sequence of technical results: \lemref{lem: Z0} to \lemref{lem: Killing}. 

To set the stage, note that
\begin{align}\label{UX}
 \begin{gathered}U(\widebar{X}_i, X_j)
 =\sum_{k=1}^d\,\langle U(\widebar{X}_i, X_j), \widebar{X}_k\rangle\,X_k
 =-\sum_{k=1}^d\Big(\langle \widebar{X_i}, [X_j, \widebar{X}_k]_{\mathfrak{m}}\rangle
      +(-1)^{[i][j]} \langle X_j, [\widebar{X}_i, \widebar{X}_k]_{\mathfrak{m}}\rangle\Big)X_k,
 \\
 U(X_i, \widebar{X}_j)
 =\sum_{k=1}^d\,\langle X_k,U(X_i,\widebar{X}_j)\rangle\,\widebar{X}_k
 =\sum_{k=1}^d\Big(\langle [X_k,X_i]_{\mathfrak{m}},\widebar{X}_j]\rangle
    +(-1)^{[i][j]}\langle [X_k,\widebar{X}_j]_{\mathfrak{m}},X_i\rangle\Big)\widebar{X}_k,
\end{gathered}
\end{align}
where the second equality in each line follows from \eqref{eq: defU}.

For clarity, we also note that
\begin{align}\label{StrXX}
 \Str_{\mathfrak{m}}(\ad_{\mathfrak{m}}(X))
 =\sum_{i=1}^d\,\langle \widebar{X}_i,[X,X_i]_{\mathfrak{m}}\rangle
 =\sum_{i=1}^d(-1)^{[i]}\langle X_i,[X,\widebar{X}_i]_{\mathfrak{m}}\rangle,\qquad X\in\mathfrak{g},
\end{align}
with the second equality following from \lemref{lem: fXX},
and that this expression is zero if $X\in\mathfrak{g}_{\bar{1}}$.
\begin{lemma}\label{lem: Z0}
 It holds that
\begin{align*}
 \sum_{i=1}^dU(\widebar{X}_i,X_i)=\sum_{i=1}^d(-1)^{[i]}U(X_i,\widebar{X}_i)=0.
\end{align*}
\end{lemma}
\begin{proof}
The first equality follows from \lemref{lem: fXX}. For the second equality, let
\begin{align*}
 Z:=\sum_{i=1}^dU(\widebar{X}_i,X_i),
\end{align*}
so $Z\in\mathfrak{m}_{\bar{0}}$. We want to show that $\langle Z,X\rangle=0$ for all $X\in\mathfrak{g}$.
Since $\langle\mathfrak{k},\mathfrak{m}\rangle=0$, it suffices to consider $X\in\mathfrak{m}$, and since 
$Z\in\mathfrak{m}_{\bar{0}}$, the relation trivially holds if $X\in\mathfrak{m}_{\bar{1}}$, so we will assume that 
$X\in\mathfrak{m}_{\bar{0}}$. It then follows from \eqref{eq: defU} that
\begin{align*}
 \langle Z,X\rangle=\sum_{i=1}^d\,\langle\widebar{X}_i, [X,X_i]_{\mathfrak{m}}\rangle 
      +\sum_{i=1}^d(-1)^{[i]}\langle X_i, [X,\widebar{X}_i]_{\mathfrak{m}}\rangle
 =2\Str_{\mathfrak{m}}(\ad_{\mathfrak{m}}(X))
 =2\Str_{\mathfrak{g}}(\ad_{\mathfrak{g}}(X)),
\end{align*}
where the last equality uses that $X\in\mathfrak{m}$, $\mathfrak{g}=\mathfrak{k}\oplus \mathfrak{m}$ and 
$[\mathfrak{k},\mathfrak{m}]\subseteq \mathfrak{m}$. Since $\mathfrak{g}$ is simple, 
we have $\mathfrak{g}=[\mathfrak{g},\mathfrak{g}]$, and since $\Str([f_1,f_2])=0$ for all 
$f_1,f_2\in\mathrm{End}(\mathfrak{g})$, it follows that $\Str_{\mathfrak{g}}(\ad_{\mathfrak{g}}(X))=0$.
\end{proof}
\begin{corollary}\label{lem: UU0}
For each $i=1,\ldots,d$, we have
\begin{align*}
 \sum_{j=1}^d\,\langle U(X_i, \widebar{X}_i),U(\widebar{X}_j,X_j)\rangle=0.
\end{align*}
\end{corollary}
\begin{remark}
The lemma holds as long as $\Str_{\mathfrak{g}}(\ad_{\mathfrak{g}}(X))=0$ for all $X$, 
even if $\mathfrak{g}$ is not simple. For example, the lemma holds for $\mathfrak{g}=\mathfrak{sl}(n|n)$.
\end{remark}
\begin{lemma}\label{lem: sumUU}
For each $i=1,\ldots,d$, we have
\begin{align*}
  &\sum_{j=1}^d (-1)^{[i]+[j]} \langle U(X_j,\widebar{X}_i), U(X_i, \widebar{X}_j)\rangle
    =-\sum_{j,k=1}^d\,\langle X_i, [\widebar{X}_k, \widebar{X}_j]_{\mathfrak{m}}\rangle\,
     \langle [X_j,X_k]_{\mathfrak{m}},\widebar{X}_i\rangle
 \\ 
 &+\sum_{k=1}^d\Big(\langle X_i,[[\widebar{X}_i,\widebar{X}_k]_{\mathfrak{m}},X_k]_{\mathfrak{m}}\rangle
    +\langle [[X_i,\widebar{X}_k]_{\mathfrak{m}},X_k]_{\mathfrak{m}},\widebar{X}_i\rangle
    -\langle [X_i,\widebar{X}_k]_{\mathfrak{m}}, [X_k, \widebar{X}_i]_{\mathfrak{m}}\rangle\Big). 
\end{align*}
\end{lemma}
\begin{proof}
Use \eqref{UX} to rewrite the lefthand side as a sum of four double sums. Then, apply \eqref{XXX} to 
$X=[\widebar{X}_i,\widebar{X}_k]_{\mathfrak{m}}$, respectively $X=[X_i,X_k]_{\mathfrak{m}}$ (twice), to reduce three of 
the double sums to single sums, and apply \lemref{lem: fXX}.
\end{proof}
\begin{lemma}\label{lem: Killing}
Let $X,Y\in\mathfrak{m}$ be homogeneous. Then,
\begin{align*}
    B(X,Y)= \sum_{i=1}^d\Big(\langle \widebar{X}_i, [X, [Y,X_i]_{\mathfrak{g}}]_{\mathfrak{m}}\rangle 
      +(-1)^{[X]} \langle \widebar{X}_i,[Y,[X,X_i]_{\mathfrak{k}}]_{\mathfrak{m}}\rangle\Big).
\end{align*}
\end{lemma}
\begin{proof}
Let $\{Y_1, \dots, Y_h\}$ be a homogeneous basis for $\mathfrak{k}$, with right dual basis 
$\{\widebar{Y}_1, \dots, \widebar{Y}_h \}$, and recall that $\langle\mathfrak{k},\mathfrak{m}\rangle=0$.
Expressing $\ad(X)\ad(Y)(\widebar{X}_i)$ and $\ad(X)\ad(Y)(\widebar{Y}_k)$ in terms of the dual basis, and
using a formula similar to \eqref{StrXX} but for a supertrace over all of $\mathfrak{g}$, yields
\begin{align*}
  B(X,Y)=\sum_{i=1}^d (-1)^{[X_i]} \langle X_i, [X, [Y, \widebar{X}_i]_{\mathfrak{g}}]_{\mathfrak{m}}\rangle 
       + \sum_{k=1}^h(-1)^{[Y_k]} \langle Y_k, [X,[Y, \widebar{Y_k}]_{\mathfrak{g}}]_{\mathfrak{k}}\rangle. 
\end{align*}
Use \eqref{gkm} and $[X, [Y, \widebar{Y_k}]_{\mathfrak{k}}]_{\mathfrak{g}}\in \mathfrak{m}$ to rewrite
the sum over $k$ as
\begin{align*}
 \sum_{k=1}^h(-1)^{[Y_k]} \langle Y_k, [X,[Y, \widebar{Y_k}]_{\mathfrak{m}}]_{\mathfrak{k}}\rangle
  &=-\sum_{i=1}^d\sum_{k=1}^h (-1)^{[Y_k]+[Y][Y_k]}\langle X_i,[\widebar{Y}_k,Y]_{\mathfrak{m}}\rangle\, 
    \langle Y_k, [X, \widebar{X}_i]_{\mathfrak{k}} \rangle
  \\
  &=\sum_{i=1}^d (-1)^{[X_i]+[X]} \langle X_i, [Y,[X,\widebar{X}_i]_{\mathfrak{k}}]_{\mathfrak{m}}\rangle,
\end{align*}
where the first equality follows from
$[Y, \widebar{Y_k}]_{\mathfrak{m}}= \sum_{i=1}^d\,\langle X_i, [Y, \widebar{Y_k}]_{\mathfrak{m}}\rangle \widebar{X}_i$, 
and the second from
$[X,\widebar{X}_i]_{\mathfrak{k}}= \sum_{k=1}^h\,\langle Y_k, [X, \widebar{X}_i]_{\mathfrak{k}} \rangle\,\widebar{Y}_k$.
By \lemref{lem: fXX}, this establishes the desired expression.
\end{proof}
Using \lemref{lem: fXX}, it follows from \lemref{lem: Killing} that, for each $i=1,\ldots,d$,
\begin{align}\label{BXX}
     B(X_i, \widebar{X}_i)
     =\sum_{j=1}^d\Big(
     \langle [X_i,[\widebar{X}_i,\widebar{X}_j]_{\mathfrak{g}}]_{\mathfrak{m}},X_j\rangle
     +\langle [[\widebar{X}_j,X_i]_{\mathfrak{k}},\widebar{X}_i]_{\mathfrak{m}},X_j\rangle\Big).
\end{align}
\begin{proof}[Proof of \thmref{thm: Ric}]
Since the $g$-normalised basis and its right-dual companion are related as in \eqref{eq: dualX}, we can use \eqref{RicR} and 
\propref{prop: Rie} to obtain an expression for $\Ric(X_i,\widebar{X}_i)$. Applying \lemref{lem: UU0} and \eqref{XX0} to the expression 
for $\mathrm{Ric}(X_i, \widebar{X}_i)$, the corresponding terms vanish, yielding
\begin{align*}
 4 \mathrm{Ric}(X_i, \widebar{X_i})
 &= -2\sum_{j=1}^d(-1)^{[j]+[i][j]} \langle [X_i,X_j]_{\mathfrak{m}}, [\widebar{X}_i, \widebar{X}_j]_{\mathfrak{m}}\rangle 
 + \sum_{j=1}^d (-1)^{[i]+[j]} \langle [X_j, \widebar{X}_i]_{\mathfrak{m}}, [X_i, \widebar{X}_j]_{\mathfrak{m}}\rangle 
 \\ 
 &-\sum_{j=1}^d (-1)^{[j]+[i][j]} \langle X_i, [[X_j, \widebar{X}_i]_{\mathfrak{m}}, \widebar{X}_j]_{\mathfrak{m}}\rangle 
 + \sum_{j=1}^d  (-1)^{[j]+[i][j]} \langle [X_i, [X_j, \widebar{X}_i]_{\mathfrak{m}}]_{\mathfrak{m}}, \widebar{X}_j \rangle 
 \\ 
 &+ \sum_{j=1}^d(-1)^{[j]} \langle X_j, [[X_i, \widebar{X}_i]_{\mathfrak{m}}, \widebar{X_j}]_{\mathfrak{m}}\rangle 
 + \sum_{j=1}^d(-1)^{[j]} \langle [[X_i, \widebar{X}_i]_{\mathfrak{m}}, X_j]_{\mathfrak{m}}, \widebar{X}_j \rangle 
 \\ 
 & + \sum_{j=1}^d (-1)^{[i]+[j]} \langle X_j, [\widebar{X}_i, [X_i, \widebar{X}_j]_{\mathfrak{m}} ]_{\mathfrak{m}} \rangle 
 + \sum_{j=1}^d (-1)^{[j]+[i][j]} \langle [X_j, [X_i, \widebar{X}_j]_{\mathfrak{m}}]_{\mathfrak{m}}, \widebar{X}_i \rangle 
 \\ 
 &+2 \sum_{j=1}^d (-1)^{[j]+[i][j]} \langle X_i, [X_j,[\widebar{X}_i, \widebar{X}_j]_{\mathfrak{g}} ]_{\mathfrak{m}} \rangle 
 -2 \sum_{j=1}^d (-1)^{[j]} \langle X_j, [X_i, [\widebar{X}_i, \widebar{X}_j]_{\mathfrak{g}}]_{\mathfrak{m}}\rangle 
 \\ 
 &+ \sum_{j=1}^d (-1)^{[i]+[j]} \langle U(X_j, \widebar{X}_i), U(X_i, \widebar{X}_j)\rangle. 
\end{align*} 
Signs may be eliminated by invoking \lemref{lem: fXX} and by applying the skew-supersymmetry of the Lie bracket together with 
the supersymmetry of the scalar product. Applying \lemref{lem: sumUU} and \eqref{gkm} 
to the ensuing expression for $\Ric(X_i,\widebar{X}_i)$, and using \eqref{BXX}, yields
\begin{align}\label{4RBB}
   &4\Ric(X_i, \widebar{X_i})+B(X_i,\widebar{X}_i)+(-1)^{[i]}B(\widebar{X}_i,X_i)
  \nonumber\\
  &=2\sum_{j=1}^d\,\langle [X_i,\widebar{X}_j]_{\mathfrak{m}}, [X_j,\widebar{X}_i]_{\mathfrak{m}}\rangle
  -\sum_{j,k=1}^d\,\langle X_i, [\widebar{X}_j,\widebar{X}_k]_{\mathfrak{m}}\rangle\,
       \langle [X_k,X_j]_{\mathfrak{m}}, \widebar{X}_i\rangle
  \nonumber\\
  &+2\sum_{j=1}^d\Big(
   \langle [ [\widebar{X}_j, X_i]_{\mathfrak{g}}, \widebar{X}_i]_{\mathfrak{m}}, X_j \rangle
   +\langle [[X_i, \widebar{X}_i]_{\mathfrak{m}}, \widebar{X}_j]_{\mathfrak{m}}, X_j \rangle
  \nonumber\\
  &\qquad\qquad
    - \langle[X_i, [\widebar{X}_i, \widebar{X}_j]_{\mathfrak{m}}]_{\mathfrak{m}}, X_j \rangle
    -\langle X_i, [ [\widebar{X}_i, \widebar{X}_j]_{\mathfrak{k}}, X_j ]_{\mathfrak{m}} \rangle
   \Big).
\end{align}
We can use the super-Jacobi identity
\begin{align*}
 [[\widebar{X}_j,X_i]_{\mathfrak{g}},\widebar{X}_i]_{\mathfrak{m}}
 =-[[X_i,\widebar{X}_i]_{\mathfrak{g}},\widebar{X}_j]_{\mathfrak{m}}
 +[X_i,[\widebar{X}_i,\widebar{X}_j]_{\mathfrak{g}}]_{\mathfrak{m}}
\end{align*}
to rewrite the first term in the last $j$-sum, and use \eqref{YXX} to conclude that
\begin{align*}
 \sum_{j=1}^d\,\langle [[X_i,\widebar{X}_i]_{\mathfrak{g}},\widebar{X}_j]_{\mathfrak{m}},X_j\rangle
 =\sum_{j=1}^d\,\langle [[X_i,\widebar{X}_i]_{\mathfrak{m}},\widebar{X}_j]_{\mathfrak{m}},X_j\rangle,
\end{align*}
while $\ad_{\mathfrak{k}}$-invariance implies that
\begin{align*}
   \langle X_i,[[\widebar{X}_i,\widebar{X}_j]_{\mathfrak{k}},X_j]_{\mathfrak{m}}\rangle
   =\langle [X_i,[\widebar{X}_i,\widebar{X}_j]_{\mathfrak{k}}]_{\mathfrak{m}},X_j\rangle.
\end{align*}
It follows that the last $j$-sum vanishes, and as $B(X_i,\widebar{X}_i)=(-1)^{[i]}B(\widebar{X}_i,X_i)$, 
the formula in \thmref{thm: Ric} follows.
\end{proof}

\subsection{Naturally reductive metrics}
\label{sec: nat red}

We say that the $G$-invariant metric $g$ on the homogeneous superspace $M=G/K$ is 
\emph{naturally reductive} if
\begin{align}\label{red}
 \langle [X,Y]_{\mathfrak{m}},Z\rangle=\langle X,[Y,Z]_{\mathfrak{m}}\rangle
\end{align}
for all $X,Y,Z\in\mathfrak{m}$. This mimics the definition of a naturally reductive metric in classical homogeneous 
geometry~\cite{DZ79}. Such metrics have been used extensively in the study of the Einstein equation, spectral geometry, 
the prescribed Ricci curvature problem, and other topics; see, e.g.,~\cite{DZ79,GS10,Lau20,AGP20}. 
When $g$ is naturally reductive, the curvature formulas obtained above simplify substantially.

\begin{lemma}\label{lem: U0}
Let $g$ be naturally reductive. Then, $U=0$.
\end{lemma}
\begin{proof}
Applying \eqref{red} to \eqref{eq: defU} yields
\begin{align*}
 \langle U(X,Y), Z\rangle=0,\qquad X,Y,Z\in\mathfrak{m}.
\end{align*}
Since this holds for all $Z$, the result follows from the non-degeneracy of $\langle-,-\rangle$.
\end{proof}

\begin{proposition}\label{prop: CurvInv}
Let $M=G/K$ be a homogeneous superspace with $G$-invariant graded naturally reductive Riemannian metric~$g$, 
where $G=(G_0, \mathfrak{g})$ is a connected Lie supergroup, $\mathfrak{g}$ a compact real form 
of a basic classical Lie superalgebra, and $K$ a connected closed Lie subsupergroup of $G$.
Let $\{X_1,\ldots,X_d\}$ be a $g$-normalised homogeneous basis for $\mathfrak{m}=T_K(M)$,
with right dual basis $\{\widebar{X}_1,\ldots,\widebar{X}_d\}$.
Then, for every $i,j,k,\ell\in\{1,\ldots,d\}$,
\begin{align*}
   4\,\langle R(X_i, X_j)X_k, X_{\ell}\rangle
   &=2\,\langle X_i, [X_j,[X_k,X_\ell]_{\mathfrak{m}}]_{\mathfrak{m}}\rangle 
     +4\,\langle X_i, [X_j,[X_k,X_\ell]_{\mathfrak{k}}]_{\mathfrak{m}}\rangle
    \\[.15cm] 
    &\quad+(-1)^{[X_j][X_k]}\langle X_i, [X_k,[X_j,X_\ell]_{\mathfrak{m}}]_{\mathfrak{m}}\rangle
      -(-1)^{([X_j]+[X_k])[X_\ell]}\langle X_i, [X_\ell,[X_j,X_k]_{\mathfrak{m}}]_{\mathfrak{m}}\rangle
\end{align*}
and
\begin{align*}
  \Ric(X_i, \widebar{X}_i)
  &=-\frac{1}{2}B(X_i, \widebar{X}_i)
   +\frac{1}{4}\sum_{j=1}^d\,\langle [X_i,\widebar{X}_j]_{\mathfrak{m}}, [X_j,\widebar{X}_i]_{\mathfrak{m}}\rangle.
\end{align*}
\end{proposition}
\begin{proof}
By applying \eqref{red} and \lemref{lem: U0} to the curvature expressions in \propref{prop: Rie} 
and \thmref{thm: Ric}, the result follows.
\end{proof}

\subsection{Diagonal metrics}
\label{sec: diag_curv}

Here, we discuss particular specialisations of the formula for the Ricci curvature 
that we will use to obtain our main results on Einstein metrics in \secref{sec: gen_flag}.

Recall that $\mathfrak{g}$ admits the $Q$-orthogonal decomposition $\mathfrak{g}=\mathfrak{k}\oplus \mathfrak{m}$. 
As $K$ is connected, this means that $\mathfrak{m}$ is an $\ad_{\mathfrak{k}}$-representation, 
that is, $[\mathfrak{k}, \mathfrak{m}]\subseteq \mathfrak{m}$. 
In the following, we will assume $\mathfrak{m}$ decomposes 
into $Q$-orthogonal $\ad_{\mathfrak{k}}$-irreducible representations,
\begin{align}\label{eq: mdec}
  \mathfrak{m}=\mathfrak{m}_1\oplus \cdots \oplus \mathfrak{m}_s,
\end{align}
where $\mathfrak{m}_i\not \cong \mathfrak{m}_j$ as $\ad_{\mathfrak{k}}$-representations, for all $i\neq j$. Such a 
decomposition is unique up to permutation of summands. As $\mathfrak{m}_i$ is $\ad_{\mathfrak{k}}$-irreducible, and both 
$Q$ and the Killing form $B$ are even bilinear forms, Schur's lemma for Lie superalgebras~\cite{Kac77} tells us that there 
exist $b_1,\ldots,b_s\in \mathbb{R}$ such that 
\begin{align}\label{BbQ}
 B|_{\mathfrak{m}_i}=-b_iQ|_{\mathfrak{m}_i}, \qquad i=1,\ldots,s,
\end{align}
where the sign convention is adopted from the non-super setting.

As we will discuss in the following, the curvature formulas obtained in \secref{subsec: curvs} simplify quite dramatically for 
diagonal metrics. We will frequently use an ordered basis for $\mathfrak{m}$ that is adapted to the decomposition 
\eqref{eq: mdec}. For each $i=1,\ldots,s$, we thus let 
$\mathfrak{M}_i=\{e^i_{\alpha}\,|\, \alpha=1,\ldots,\dim(\mathfrak{m}_i)\}$ be an ordered $Q$-normalised homogeneous 
basis for $\mathfrak{m}_i$ such that the corresponding matrix representation of $Q$ is of the form \eqref{eq: normbasis}, 
and let $\{\bar{e}_{\alpha}^i\,|\, \alpha=1,\ldots,\dim(\mathfrak{m}_i)\}$ be the right $Q$-dual basis, defined by
\begin{align*}
 Q(e_{\alpha}^i, \bar{e}_{\alpha'}^i)= \delta_{\alpha, \alpha'}, \qquad \alpha, \alpha'\in\{1,\ldots,\dim(\mathfrak{m}_i)\}.
\end{align*}
Ordered as indicated, $\mathfrak{M}_1\cup\cdots\cup \mathfrak{M}_s$ then constitutes a suitably ordered basis for 
$\mathfrak{m}$. 
For each $i=1,\ldots,s$, we let $I^i$ denote the index set of the basis $\mathfrak{M}_i$ for $\mathfrak{m}_i$.

Let $g$ be a $G$-invariant graded Riemannian metric on $M$. Then, by \thmref{thm: invmetric}, $g$ is identified with 
an $\Ad_K$-invariant scalar superproduct $\langle-,-\rangle$ on $\mathfrak{m}$. Since $K$ is connected, 
the $\Ad_K$-invariance is equivalent to the $\ad_{\mathfrak{k}}$-invariance of $\langle-,-\rangle$ at the level of Lie 
superalgebras. By Schur's lemma, the metric $g$ is therefore of the form
\begin{align}\label{eq: metric}
\langle-,-\rangle= \sum_{i=1}^s x_i Q|_{\mathfrak{m}_i}, \qquad x_1,\ldots,x_s\in\mathbb{R}^\times,
 \end{align} 
and we say that it is \emph{diagonal} with respect to the decomposition~\eqref{eq: mdec}.
Unlike in classical Riemannian geometry, the parameters $x_i$ need not be positive.
If they all are, then we say that the metric is \textit{positive}.
Moreover, the Ricci curvature is an even $\ad_{\mathfrak{k}}$-invariant 
supersymmetric bilinear form on $\mathfrak{m}$, so it is of the form
\begin{align}\label{RicQm}
 \Ric= \sum_{i=1}^s r_i Q|_{\mathfrak{m}_i}, \qquad r_1,\ldots,r_s\in\mathbb{R}.
\end{align}
With this, the Einstein equation \eqref{intro: Einstein} becomes
\begin{align}\label{eq: Ein}
 r_i=cx_i,\qquad i=1,\ldots,s.
\end{align}

\subsubsection{Structure constants}
\label{subsec: ijk}

We now introduce super-analogues of the structure constants; see, e.g.,~\cite{WZ86}.
\begin{definition}
For each triple $i,j ,k\in\{1,\ldots,s\}$, let
\begin{align*}
 [ijk]:=-\sum_{\alpha\in I^i,\,\beta\in I^j,\,\gamma\in I^k}
   Q\big(e_{\gamma}^k, [e_{\alpha}^i, e_{\beta}^j]_{\mathfrak{m}_k}\big)\,Q\big([\bar{e}_{\beta}^j,\bar{e}_{\alpha}^i]_{\mathfrak{m}_k},\bar{e}_{\gamma}^k\big).
\end{align*}
\end{definition}
The $\ad_{\mathfrak{g}}$-invariance of $Q$ implies that the structure constant $[ijk]$ is symmetric in all three indices.
\begin{remark}
If $\dim(\mathfrak{m}_{\bar{1}})=0$ and $Q|_{\mathfrak{m}_{\bar{0}}}$ is positive definite, i.e., $M$ is a classical 
homogeneous Riemannian manifold, then $\bar{e}_{\alpha}^i=e_{\alpha}^i$ and the structure constant $[ijk]$ coincides with 
its classical counterpart.
\end{remark}
Let $\Str_{\mathfrak{m}_i}(f)$ denote the supertrace of the restriction $f|_{\mathfrak{m}_i}$ for any 
$f\in\mathrm{End}(\mathfrak{g})$, and let $\ad_{\mathfrak{m}_i}(X):=[X,-]_{\mathfrak{m}_i}$ for any $X\in \mathfrak{g}$. 
We then have the following result. 
\begin{lemma}\label{lem: sumstr}
For each triple $i,j ,k\in\{1,\ldots,s\}$, we have
\begin{align*}
 [ijk]= -\sum_{\alpha\in I^i}(-1)^{[e_{\alpha}^i]} 
   \Str_{\mathfrak{m}_j}\!\big(\!\ad_{\mathfrak{m}_j}(e_{\alpha}^i) \ad_{\mathfrak{m}_k}(\bar{e}_{\alpha}^i)\big).
\end{align*} 
Moreover, $[ijk]$ is independent of the $Q$-normalised homogeneous bases of $\mathfrak{m}_1,\ldots,\mathfrak{m}_s$.
\end{lemma}
\begin{proof}
Using $[\bar{e}^j_{\beta}, \bar{e}_{\alpha}^i]_{\mathfrak{m}_k}= \sum_{\gamma\in I^k} Q([\bar{e}^j_{\beta}, \bar{e}_{\alpha}^i]_{\mathfrak{m}}, \bar{e}_{\gamma}^k) e_{\gamma}^k$,
we see that
\begin{align*}
 [ijk]=-\sum_{\alpha\in I^i,\,\beta\in I^j}Q\big([\bar{e}^j_{\beta}, \bar{e}_{\alpha}^i]_{\mathfrak{m}_k}, [e_{\alpha}^i, e_{\beta}^j]_{\mathfrak{m}_k}\big),
\end{align*}
so the desired expression for $[ijk]$ follows from
\begin{align}\label{QStr}
 \sum_{\beta\in I^j}Q\big([\bar{e}^j_{\beta}, \bar{e}_{\alpha}^i]_{\mathfrak{m}_k}, [e_{\alpha}^i, e_{\beta}^j]_{\mathfrak{m}_k}\big)
   &=\sum_{\beta\in I^j}Q\big( [[\bar{e}^j_{\beta}, \bar{e}_{\alpha}^i]_{\mathfrak{m}_k}, e_{\alpha}^i]_{\mathfrak{m}_j}, e_{\beta}^j\big) 
    =\sum_{\beta\in I^j}(-1)^{[e_{\alpha}^i] +[e_{\beta}^j]} Q\big(e_{\beta}^j, [e_{\alpha}^i, [\bar{e}_{\alpha}^i, \bar{e}_{\beta}^j]_{\mathfrak{m}_k}]_{\mathfrak{m}_j}\big)
  \nonumber\\[.1cm]
  &=(-1)^{[e_{\alpha}^i]}\Str_{\mathfrak{m}_j}\!\big(\!\ad_{\mathfrak{m}_j}(e_{\alpha}^i)\ad_{\mathfrak{m}_k}(\bar{e}_{\alpha}^i)\big).
\end{align}
It remains to show that $[ijk]$ is independent of the choice of 
basis for $\mathfrak{m}_i$. Suppose there is another $Q$-normalised homogeneous basis $\{f_{\beta}^i\,|\,\beta\in I^i\}$ 
for $\mathfrak{m}_i$, with right dual basis $\{\bar{f}_{\beta}^i\,|\,\beta\in I^i\}$, 
and let $A^i$ and $B^i$ be the (even) transition matrices such that
\begin{align*}
  f_{\beta}^{i}=\sum_{\alpha\in I^i} A^i_{\beta\alpha}e_{\alpha}^i, \qquad
  \bar{f}_{\beta}^{i}=\sum_{\alpha\in I^i} B^i_{\beta\alpha}\bar{e}_{\alpha}^i. 
\end{align*}
Then, $B^i=((A^i)^{-1})^T$ and
\begin{align*}
  \sum_{\beta \in I^i}(-1)^{[f_{\beta}^i]}\Str_{\mathfrak{m}_j}\!\big(\!\ad_{\mathfrak{m}_j}(f_{\beta}^i)\ad_{\mathfrak{m}_k}(\bar{f}_{\beta}^i)\big)
  &=\sum_{\alpha,\beta,\gamma \in I^i}(-1)^{[e_{\alpha}^i]} A^{i}_{\beta\alpha} B_{\beta\gamma}^i 
    \Str_{\mathfrak{m}_j}\!\big(\!\ad_{\mathfrak{m}_j}(e_{\alpha}^i)\ad_{\mathfrak{m}_k}(\bar{e}_{\gamma}^i)\big)
 \\[.1cm] 
 &=\sum_{\alpha\in I^i}(-1)^{[e_{\alpha}^i]}\Str_{\mathfrak{m}_j}\!\big(\!\ad_{\mathfrak{m}_j}(e_{\alpha}^i)\ad_{\mathfrak{m}_k}(\bar{e}_{\alpha}^i)\big).
\end{align*}
This proves the basis independence of $[ijk]$.
\end{proof}

Let $\{z_1,\ldots,z_{\dim(\mathfrak{k})}\}$ be a homogeneous basis 
for $\mathfrak{k}$, with right $Q$-dual basis $\{\bar{z}_1,\ldots,\bar{z}_{\dim(\mathfrak{k})}\}$. 
For each $i=1,\ldots,s$, define the Casimir operator on $\mathfrak{m}_i$ by
\begin{align*}
  C_{\mathfrak{m}_i, Q|_{\mathfrak{k}}}:= 
   -\sum_{\ell=1}^{\dim(\mathfrak{k})}(-1)^{[z_\ell]}\ad(z_\ell) \ad(\bar{z}_\ell), 
\end{align*}
where the negative sign is incorporated as a convention (cf.~\cite{WZ86}).
Since $\mathfrak{m}_i$ is $\ad_{\mathfrak{k}}$-irreducible and $C_{\mathfrak{m}_i, Q|_{\mathfrak{k}}}$ is an even operator, 
Schur's lemma implies that there exists  $c_i\in \mathbb{R}$ such that
\begin{align}\label{Cc}
 C_{\mathfrak{m}_i, Q|_{\mathfrak{k}}}= c_i\mathrm{Id},
\end{align}
where $\mathrm{Id}$ denotes the identity operator.
Using this, the next result provides a super-analogue of~\cite[Lemma 1.5]{WZ86}. 
For ease of notation, we set
\begin{align}\label{di} 
 d_i:=\mathrm{sdim}(\mathfrak{m}_i)=\dim((\mathfrak{m}_i)_{\bar{0}})-\dim((\mathfrak{m}_i)_{\bar{1}}), \qquad i=1,\ldots,s.
\end{align}
Note that $d_i$, as well as $b_i$ and $[ijk]$, can be negative in the super-setting. 
\begin{proposition}\label{prop: StrCas}
Let $M=G/K$ be a homogeneous superspace, with notation as in \thmref{thm: Ric}, and suppose $\mathfrak{m}$ admits a 
$Q$-orthogonal multiplicity-free decomposition of the form \eqref{eq: mdec}. Then,
\begin{align*} 
 \sum_{j,k=1}^s [ijk]=d_i(b_i-2c_i),\qquad i=1,\dots,s,
\end{align*}
where $d_i$, $b_i$ and $c_i$ are given in \eqref{di}, \eqref{BbQ} and \eqref{Cc}.
\end{proposition}
\begin{proof}
Using \lemref{lem: sumstr}, we have
\begin{align*}
    \sum_{j,k=1}^s[ijk]=-\sum_{j,k=1}^s \sum_{\alpha\in I^i}(-1)^{[e_{\alpha}^i]} \Str_{\mathfrak{m}_j}\!\big(\!\ad_{\mathfrak{m}_j}(e_{\alpha}^i) \ad_{\mathfrak{m}_k}(\bar{e}_{\alpha}^i)\big)
       =- \sum_{\alpha\in I^i}(-1)^{[e_{\alpha}^i]} \Str_{\mathfrak{m}}\!\big(\!\ad_{\mathfrak{m}}(e_{\alpha}^i)\ad_{\mathfrak{m}}(\bar{e}_{\alpha}^i)\big).
\end{align*}    
Since $\mathfrak{g}$ admits the $Q$-orthogonal decomposition $\mathfrak{g}=\mathfrak{k} \oplus \mathfrak{m}$ 
and $[\mathfrak{k}, \mathfrak{m}]\subseteq \mathfrak{m}$,  we have
\begin{align*}
   \Str_{\mathfrak{m}}\!\big(\!\ad_{\mathfrak{m}}(e_{\alpha}^i)\ad_{\mathfrak{m}}(\bar{e}_{\alpha}^i)\big)
  =\Str_{\mathfrak{g}}\!\big(\!\ad_{\mathfrak{g}}(e_{\alpha}^i)\ad_{\mathfrak{g}}(\bar{e}_{\alpha}^i)\big) 
    -\Str_{\mathfrak{m}}\!\big(\!\ad_{\mathfrak{m}}(e_{\alpha}^i)\ad_{\mathfrak{k}}(\bar{e}_{\alpha}^i)\big)  
    -\Str_{\mathfrak{k}}\!\big(\!\ad_{\mathfrak{k}}(e_{\alpha}^i)\ad_{\mathfrak{m}}(\bar{e}_{\alpha}^i)\big).
\end{align*}
Since $\Str_V(fg)= (-1)^{[f][g]}\Str_{W}(gf)$ for any two homogeneous linear maps $f: W\to V$ and $g: V\to W$ between 
vector superspaces $V$ and $W$, we see that 
$\Str_{\mathfrak{m}}(\ad_{\mathfrak{m}}(e_{\alpha}^i)\ad_{\mathfrak{k}}(\bar{e}_{\alpha}^i))=(-1)^{[e_{\alpha}^i]}\Str_{\mathfrak{k}}(\ad_{\mathfrak{k}}(\bar{e}_{\alpha}^i)\ad_{\mathfrak{m}}(e_{\alpha}^i))$.
It follows that
\begin{align*}
    \sum_{j,k=1}^s[ijk]
    &=-\sum_{\alpha\in I^i} (-1)^{[e_{\alpha}^i]} B(e_{\alpha}^i, \bar{e}_{\alpha}^i) 
       +\sum_{\alpha\in I^i} \Str_{\mathfrak{k}}\!\big(\!\ad_{\mathfrak{k}}(\bar{e}_{\alpha}^i)\ad_{\mathfrak{m}}(e_{\alpha}^i)\big)
       +\sum_{\alpha\in I^i} (-1)^{[e_{\alpha}^i]} \Str_{\mathfrak{k}}\!\big(\!\ad_{\mathfrak{k}}(e_{\alpha}^i)\ad_{\mathfrak{m}}(\bar{e}_{\alpha}^i)\big) 
   \\[.1cm]
   &=\sum_{\alpha\in I^i} (-1)^{[e_{\alpha}^i]}b_i Q(e_{\alpha}^i, \bar{e}_{\alpha}^i)
       +2\sum_{\alpha\in I^i} (-1)^{[e_{\alpha}^i]} \Str_{\mathfrak{k}}\!\big(\!\ad_{\mathfrak{k}}(e_{\alpha}^i)\ad_{\mathfrak{m}}(\bar{e}_{\alpha}^i)\big) 
   \\ 
   &=b_id_i +2 \sum_{\alpha\in I^i} (-1)^{[e_{\alpha}^i]} 
     \sum_{\ell=1}^{\dim(\mathfrak{k})} (-1)^{[z_\ell]} Q\big(z_\ell, [e_{\alpha}^i, [\bar{e}_{\alpha}^i,\bar{z}_\ell]_{\mathfrak{m}}]_{\mathfrak{k}}\big)
   \\
   &=b_id_i +2\sum_{\alpha\in I^i} (-1)^{[e_{\alpha}^i]} 
     \sum_{\ell=1}^{\dim(\mathfrak{k})} (-1)^{[z_\ell]} Q\big(e_{\alpha}^i, [z_\ell,[\bar{z}_\ell, \bar{e}_{\alpha}^i]_{\mathfrak{m}}]_{\mathfrak{m}}\big)
   \\ 
   &=b_id_i +2\sum_{\alpha\in I^i} (-1)^{[e_{\alpha}^i]} Q\big(e_{\alpha}^i, -C_{\mathfrak{m}_i, Q|_{\mathfrak{k}}}(\bar{e}_{\alpha}^i)\big)
   \\
   &=b_id_i-2d_ic_i,
\end{align*}
where the fourth equality uses the $\ad_{\mathfrak{g}}$-invariance and supersymmetry of $Q$.
\end{proof}

\subsubsection{Ricci curvature}
\label{subsec: RicciCoef}

For diagonal metrics, we now derive expressions for the \emph{Ricci coefficients} $r_i$, $i=1,\ldots,s$, in \eqref{RicQm}.
Although the superdimensions $d_i$ may be $0$ (unlike the corresponding ordinary dimensions), the formula for the Ricci 
coefficients in \thmref{thm: SimRic} below is analogous to the corresponding classical result; cf.~\cite[Lemma 1.1]{PS97} 
and~\cite[Lemma 3.3]{PR19}. If $\dim(\mathfrak{m}_{\bar{1}})=0$, then \thmref{thm: SimRic} gives a Ricci curvature formula 
for a pseudo-Riemannian manifold.
\begin{theorem}\label{thm: SimRic}
Let $M=G/K$ be a homogeneous superspace, with notation as in \thmref{thm: Ric}, and suppose $\mathfrak{m}$ admits a 
$Q$-orthogonal multiplicity-free decomposition of the form \eqref{eq: mdec}. With the notation for $g$ as in 
\eqnref{eq: metric}, each of the corresponding Ricci coefficients $r_i$, $i=1,\ldots,s$, satisfies
\begin{align}\label{4dr}
  d_ir_i=\frac{b_id_i}{2}+\sum_{j,k=1}^s\frac{[ijk]}{4}\Big( \frac{x_i^2}{x_jx_k} - \frac{2x_j }{x_k}\Big).
\end{align}
\end{theorem}
\begin{proof}
For each $i=1,\ldots,s$, we introduce $\epsilon_i=\pm1$ such that $\tilde{x}_i=\epsilon_ix_i>0$, 
and write $\tilde{r}_i=\epsilon_ir_i$. Establishing the stated formula is then equivalent to showing that
\begin{align*}
 4d_i\tilde{r}_i=2b_id_i\epsilon_i 
  +\sum_{j,k=1}^s [ijk]\epsilon_i\epsilon_j\epsilon_k\Big( \frac{\tilde{x}_i^2}{\tilde{x}_j\tilde{x}_k} 
    - \frac{2\tilde{x}_j }{\tilde{x}_k}\Big), \qquad i=1,\ldots,s.
\end{align*}

For each $i=1,\ldots,s$, let $\{X^i_{\alpha}\,|\,\alpha\in I^i\}$ be a $g$-normalised
basis for $\mathfrak{m}_i$, with right $g$-dual basis $\{\widebar{X}_{\alpha}^i\,|\,\alpha\in I^i\}$, 
and let $\{X_1,\dots,X_d\}$ be a corresponding basis for $\mathfrak{m}$ adapted to the decomposition \eqref{eq: mdec}.
Then, $Q$-orthogonality implies that $\langle X^i_{\alpha}, X^j_{\beta}\rangle=0$ for $i\neq j$ and any $\alpha,\beta$.
For each $i=1,\ldots,s$, we also introduce the renormalised basis 
$\{e^i_{\alpha}=\sqrt{\tilde{x}_i}X_{\alpha}^i\,|\,\alpha\in I^i\}$ for $\mathfrak{m}_i$, whose right $Q$-dual basis is
seen to be given by $\{\bar{e}_{\alpha}^i=\epsilon_i\sqrt{\tilde{x}_i}\widebar{X}^i_{\alpha}\,|\, \alpha\in I^i\}$.

We now fix $i\in\{1,\ldots,s\}$ and $\alpha\in I^i$, and set 
$X=e^i_{\alpha}$ and $\widebar{X}=\epsilon_i\bar{e}_{\alpha}^i$. Then,
\begin{align*}
 Q(X, \widebar{X})=\epsilon_i,\qquad
 B(X,\widebar{X})=-b_i\epsilon_i,\qquad
 \Ric(X,\widebar{X})=\tilde{r}_i.
\end{align*}
By \thmref{thm: Ric}, we have 
\begin{align}\label{ri}
     4\tilde{r}_i&=2b_i\epsilon_i
     -2\sum_{j=1}^d (-1)^{[X_{\alpha}^i][X_j]+[X_j]}\,\tilde{x}_i\,\langle [X_{\alpha}^i, X_j]_{\mathfrak{m}}, [\widebar{X}_{\alpha}^i, \widebar{X}_j]_{\mathfrak{m}}\rangle
 \nonumber\\ 
   &\quad+\sum_{j,k=1}^d (-1)^{[X_k][X_j]+[X_{\alpha}^i]}\,\tilde{x}_i\,\langle X_{\alpha}^i, [X_k, X_j]_{\mathfrak{m}}\rangle \,
       \langle [\widebar{X}_k, \widebar{X}_j]_{\mathfrak{m}}, \widebar{X}_{\alpha}^i\rangle
 \nonumber\\ 
   &=2b_i\epsilon_i
   -2\sum_{j,k=1}^d (-1)^{[X_{\alpha}^i][X_j]+[X_{\alpha}^i]}\,\tilde{x}_i\,\langle X_k, [X_{\alpha}^i,X_j]_{\mathfrak{m}}\rangle\,
      \langle  [\widebar{X}_{\alpha}^i, \widebar{X}_j]_{\mathfrak{m}}, \widebar{X}_k\rangle
 \nonumber\\ 
   &\quad+\sum_{j,k=1}^d (-1)^{[X_k][X_j]+[X_{\alpha}^i]}\,\tilde{x}_i\,\langle X_{\alpha}^i, [X_k, X_j]_{\mathfrak{m}}\rangle\,
      \langle [\widebar{X}_k, \widebar{X}_j]_{\mathfrak{m}}, \widebar{X}_{\alpha}^i\rangle,
 \nonumber\\ 
   &=2b_i\epsilon_i
          -2\sum_{j,k=1}^s \sum_{ \beta \in I^j,\,\gamma\in I^k } (-1)^{[X_{\alpha}^i][X_{\beta}^j]+[X_{\alpha}^i]}\,\tilde{x}_i\,\langle X_{\gamma}^k, [X_{\alpha}^i,X_{\beta}^j]_{\mathfrak{m}}\rangle\,
          \langle  [\widebar{X}_{\alpha}^i, \widebar{X}_{\beta}^j]_{\mathfrak{m}}, \widebar{X}_{\gamma}^k\rangle
 \nonumber\\ 
   &\quad+\sum_{j,k=1}^s \sum_{ \beta \in I^j,\,\gamma\in I^k } (-1)^{[X_{\gamma}^k][X_{\beta}^j]+[X_{\alpha}^i]}\,\tilde{x}_i\,
     \langle X_{\alpha}^i, [X_{\gamma}^k, X_{\beta}^j]_{\mathfrak{m}}\rangle\, \langle [\widebar{X}_{\gamma}^k, \widebar{X}_{\beta}^j]_{\mathfrak{m}}, \widebar{X}_{\alpha}^i\rangle
 \nonumber\\ 
  &=2b_i\epsilon_i
   +2\sum_{j,k=1}^s \sum_{ \beta \in I^j,\,\gamma\in I^k }(-1)^{[e_{\alpha}^i]}\,\frac{\tilde{x}_k}{\tilde{x}_j}\epsilon_i\epsilon_j\epsilon_k 
  Q\big(e_{\gamma}^k, [e_{\alpha}^i,e_{\beta}^j]_{\mathfrak{m}_k}\big)\, 
  Q\big([\bar{e}_{\beta}^j, \bar{e}^i_{\alpha}]_{\mathfrak{m}_k}, \bar{e}_{\gamma}^k\big)
 \nonumber\\ 
 &\quad-\sum_{j,k=1}^s \sum_{ \beta \in I^j,\,\gamma\in I^k } (-1)^{[e_{\alpha}^i]}\,\frac{\tilde{x}_i^2}{\tilde{x}_j\tilde{x}_k}\epsilon_i\epsilon_j\epsilon_k
   Q\big(e_{\alpha}^i, [e_{\gamma}^k, e_{\beta}^j]_{\mathfrak{m}_i}\big)\, 
   Q\big([\bar{e}_{\beta}^j, \bar{e}_{\gamma}^k]_{\mathfrak{m}_i}, \bar{e}^i_{\alpha}\big)
 \nonumber\\ 
 &=2b_i\epsilon_i-\sum_{j,k=1}^s \sum_{ \beta \in I^j,\,\gamma\in I^k }(-1)^{[e_{\alpha}^i]}\epsilon_i\epsilon_j\epsilon_k 
   Q\big(e_{\gamma}^k, [e_{\alpha}^i,e_{\beta}^j]_{\mathfrak{m}_k}\big)\, 
   Q\big([\bar{e}_{\beta}^j, \bar{e}^i_{\alpha}]_{\mathfrak{m}_k}, \bar{e}_{\gamma}^k\big) 
   \Big(\frac{\tilde{x}_i^2}{\tilde{x}_j\tilde{x}_k}  - \frac{2\tilde{x}_k}{\tilde{x}_j} \Big) 
 \nonumber\\ 
 &=2b_i\epsilon_i-\sum_{j,k=1}^s \sum_{\beta\in I^j}(-1)^{[e_{\alpha}^i]}\epsilon_i 
 Q\big([\bar{e}_{\beta}^j, \bar{e}^i_{\alpha}]_{\mathfrak{m}_k}, [e_{\alpha}^i,e_{\beta}^j]_{\mathfrak{m}_k}\big)
 \Big( \frac{x_i^2}{x_jx_k}  - \frac{2x_k}{x_j} \Big)
 \nonumber\\ 
 &=2b_i \epsilon_i-\sum_{j,k=1}^s \Str_{\mathfrak{m}_j}\!\big(\!\ad_{\mathfrak{m}_j}(e_{\alpha}^i) \ad_{\mathfrak{m}_k}(\bar{e}_{\alpha}^i)\big)\epsilon_i\Big( \frac{x_i^2}{x_jx_k} - \frac{2x_k}{x_j}\Big),
\end{align}
where the second, second-to-last and last equalities follow from 
$[X_{\alpha}^i,X_j]_{\mathfrak{m}}= \sum_{k=1}^d\,\langle X_k, [X_{\alpha}^i,X_j]_{\mathfrak{m}}\rangle\,\widebar{X}_k$,
$[\bar{e}_{\beta}^j, \bar{e}^i_{\alpha}]_{\mathfrak{m}_k}= \sum_{\gamma\in I^k} Q\big([\bar{e}_{\beta}^j, \bar{e}^i_{\alpha}]_{\mathfrak{m}_k},\bar{e}_{\gamma}^k\big) e_{\gamma}^k$ and \eqref{QStr}, respectively. 
Taking the graded sum over the $Q$-normalised homogeneous basis for $\mathfrak{m}_i$, we obtain
\begin{align*}
 4d_ir_i
 =4\sum_{\alpha\in I^i}(-1)^{[e_{\alpha}^i]}\,r_i
   =2b_id_i-\sum_{\alpha\in I^i} 
      \sum_{j,k=1}^s(-1)^{[e_{\alpha}^i]}\Str_{\mathfrak{m}_j}\!\big(\!\ad_{\mathfrak{m}_j}(e_{\alpha}^i) \ad_{\mathfrak{m}_k}(\bar{e}_{\alpha}^i)\big)
      \Big( \frac{x_i^2}{x_jx_k} - \frac{2x_k}{x_j}\Big).
 \end{align*}
The desired formula now follows by application of \lemref{lem: sumstr} and the symmetry of $[ijk]$.
\end{proof}
\begin{proposition}\label{prop: scalar}
With the notation as in \thmref{thm: SimRic}, the scalar curvature $S$ is given by 
\begin{align*}
  S= \frac{1}{2}\sum_{i=1}^s \frac{b_id_i}{x_i} -\frac{1}{4} \sum_{i,j,k=1}^s [ijk]\frac{x_k}{x_ix_j}. 
\end{align*}
\end{proposition}
\begin{proof}
It follows from \eqref{RicQm} that
\begin{align*}
  S= \sum_{i=1}^d (-1)^{[X_i]}\Ric(X_i, \widebar{X}_i)
   = \sum_{i=1}^s \sum_{\alpha\in I^i}(-1)^{[X_\alpha^i]} r_i Q(X_\alpha^i, \widebar{X}_\alpha^i) 
   = \sum_{i=1}^s \frac{d_ir_i}{x_i}.
\end{align*}
Using \thmref{thm: SimRic} and the symmetry of $[ijk]$, we obtain the formula for $S$. 
\end{proof}

If $d_i\neq0$, then \thmref{thm: SimRic} says that
\begin{align}\label{rdnot0}
  r_i= \frac{b_i}{2} +\sum_{j,k=1}^s\frac{[ijk]}{4d_i} \Big( \frac{x_i^2}{x_jx_k} - \frac{2x_j }{x_k}\Big).
\end{align}
If $d_i=0$, on the other hand, we cannot isolate $r_i$ from \eqref{4dr} but will now derive a formula for $r_i$ in case there 
exists an $i$-selected pair $(j,k)$. Here, $(j,k)$ is \emph{$i$-selected} if $j,k\in\{1,\ldots,s\}$ are such that $j\leq k$ and
\begin{align}\label{eq: trcon}
 [\mathfrak{m}_i, \mathfrak{m}_u]_{\mathfrak{m}_{v}}={0}
\end{align}
for all $1\leq u\leq v\leq s$ such that $(u,v)\neq (j,k)$. In \secref{sec: flagA} and \secref{sec: flagC}, 
we give concrete examples where \eqref{eq: trcon} is satisfied; see \propref{prop: mdec} and \propref{prop: rmdecC}.

\begin{proposition}\label{prop: ricexc}
Let the notation be as in \thmref{thm: SimRic}, and suppose $d_i=0$ for some $i\in\{1,\ldots,s\}$, 
and that there exists $i$-selected $(j,k)$. Then,
\begin{align*} 
           r_i=\frac{b_i}{2}+\frac{b_i-2c_i}{4}\Big( \frac{x_i^2}{x_jx_k}-\frac{x_j}{x_k}-\frac{x_k}{x_j} \Big). 
\end{align*}
\end{proposition}
\begin{proof}
In the homogeneous $\mathfrak{m}_i$-basis $\{e^i_{\alpha}\,|\,\alpha\in I^i\}$ introduced in the proof 
of \thmref{thm: SimRic}, there exists at least one even vector, because $d_i=0$. 
Let us select one such vector and denote it by $e_{\alpha}^i$. Since \eqnref{eq: trcon} holds, we have
\begin{align*}
  \Str_{\mathfrak{m}}\!\big(\!\ad_{\mathfrak{m}}(e_{\alpha}^i)\ad_{\mathfrak{m}}(\bar{e}_{\alpha}^i)\big)
  =\sum_{u,v=1}^s\Str_{\mathfrak{m}_u}\!\big(\!\ad_{\mathfrak{m}_u}(e_{\alpha}^i)\ad_{\mathfrak{m}_v}(\bar{e}_{\alpha}^i)\big)
   =(2-\delta_{jk})\Str_{\mathfrak{m}_j}\!\big(\!\ad_{\mathfrak{m}_j}(e_{\alpha}^i)\ad_{\mathfrak{m}_k}(\bar{e}_{\alpha}^i)\big).
\end{align*}
We also have
\begin{align*}
   \Str_{\mathfrak{m}}\!\big(\!\ad_{\mathfrak{m}}(e_{\alpha}^i)\ad_{\mathfrak{m}}(\bar{e}_{\alpha}^i)\big)
   &=\Str_{\mathfrak{g}}\!\big(\!\ad_{\mathfrak{g}}(e_{\alpha}^i)\ad_{\mathfrak{g}}(\bar{e}_{\alpha}^i)\big) 
   -2\Str_{\mathfrak{k}}\!\big(\!\ad_{\mathfrak{k}}(e_{\alpha}^i)\ad_{\mathfrak{m}}(\bar{e}_{\alpha}^i)\big)
   \\ 
   &= B(e_{\alpha}^i, \bar{e}_{\alpha}^i)
     -2 \sum_{\ell=1}^{\dim(\mathfrak{k})} (-1)^{[z_\ell]} 
     Q\big(z_\ell, [e_{\alpha}^i, [\bar{e}_{\alpha}^i, \bar{z}_\ell]_{\mathfrak{m}}]_{\mathfrak{k}}\big)
   \\ 
 &=-b_iQ(e_{\alpha}^i, \bar{e}_{\alpha}^i)-2 Q\big(e_{\alpha}^i, -C_{\mathfrak{m}_i, Q|_{\mathfrak{k}}}(\bar{e}_{\alpha}^i)\big)
   \\[.1cm] 
   &= -b_i+2c_i.
\end{align*}
We now use \eqref{ri} and obtain
\begin{align*}
  r_i
 &=\frac{b_i}{2} - \frac{1}{4} \sum_{u,v=1}^s 
   \Str_{\mathfrak{m}_u}\!\big(\!\ad_{\mathfrak{m}_u}(e_{\alpha}^i) \ad_{\mathfrak{m}_v}(\bar{e}_{\alpha}^i)\big) 
   \Big( \frac{x_i^2}{x_ux_v} - \frac{2x_v}{x_u}\Big)
   \\
 &=\frac{b_i}{2}- \frac{2-\delta_{jk}}{4}\Str_{\mathfrak{m}_j}\!\big(\!\ad_{\mathfrak{m}_j}(e_{\alpha}^i)\ad_{\mathfrak{m}_k}(\bar{e}_{\alpha}^i)\big) 
    \Big(\frac{x_i^2}{x_jx_k} - \frac{x_j}{x_k}-\frac{x_k}{x_j}\Big),
 \end{align*}
from which the expression for $r_i$ follows.
\end{proof}

\section{Flag supermanifolds}
\label{sec: gen_flag}

In this section, we introduce a class of homogeneous superspaces called flag supermanifolds.
In analogy with the construction of generalised flag manifolds in terms of so-called painted Dynkin diagrams in the 
non-super case~\cite{Arv03}, we introduce what we call circled Dynkin diagrams to describe the flag supermanifolds 
in \secref{sec: flag}. In \secref{sec: examflag}, we identify the compact real forms of the Lie superalgebras of types $A$ 
and $C$. In \secref{sec: flagA1}, \secref{sec: flagA} and \secref{sec: flagC}, 
we then classify Einstein metrics on flag supermanifolds constructed from circled Dynkin diagrams of types $A$ and $C$. 
We refer to \cite{Kac77} for Dynkin diagrams of complex Lie superalgebras, to \cite{Ser83,Par90,Chu13} for classifications 
of real forms of Lie superalgebras,  and to \cite{Kac78} for the  representation theory of complex Lie superalgebras.

\subsection{Circled Dynkin diagrams}
\label{sec: flag}

Let $\mathfrak{g}^{\mathbb{C}}=\mathfrak{g}_{\bar{0}}^{\mathbb{C}}\oplus\mathfrak{g}_{\bar{1}}^{\mathbb{C}}$ be 
a basic classical Lie superalgebra,
with $Q$ a non-degenerate supersymmetric even bilinear form. 
Let $\mathfrak{h}^{\mathbb{C}}\subset \mathfrak{g}_{\bar{0}}^{\mathbb{C}}$ be a Cartan subalgebra of 
$\mathfrak{g}^{\mathbb{C}}$, 
and set $(\mathfrak{h}^{\mathbb{C}})^\vee:=\mathrm{Hom}_{\mathbb{C}}(\mathfrak{h}^{\mathbb{C}}, \mathbb{C})$. 
For each $\alpha\in (\mathfrak{h}^{\mathbb{C}})^\vee$, define
\begin{align*}
 \mathfrak{g}_{\alpha}^{\mathbb{C}}
  :=\big\{X\in \mathfrak{g}^{\mathbb{C}}\,|\, [H, X]=\alpha(H)X,\ \text{for all}\ H\in\mathfrak{h}^{\mathbb{C}}\big\}.
\end{align*}
If $\mathfrak{g}_{\alpha}^{\mathbb{C}}\neq \{0\}$ for $\alpha\neq 0$, we call $\alpha$ a \emph{root} and 
$\mathfrak{g}^{\mathbb{C}}_{\alpha}$ the corresponding \emph{root space}.
In a basic Lie superalgebra, every root $\alpha$ is either 
\emph{even} (if $\mathfrak{g}^{\mathbb{C}}_{\alpha}\subseteq \mathfrak{g}_{\bar{0}}^{\mathbb{C}}$) 
or \emph{odd} (if $\mathfrak{g}^{\mathbb{C}}_{\alpha}\subseteq \mathfrak{g}_{\bar{1}}^{\mathbb{C}}$),
and $\dim(\mathfrak{g}^{\mathbb{C}}_\alpha)=1$.
The set of all roots, $\Delta$, called the \emph{root system}, thus partitions 
as $\Delta=\Delta_{\bar{0}} \cup \Delta_{\bar{1}}$, where 
$\Delta_{\bar{0}}$ is the set of even roots (in fact, the root system of $\mathfrak{g}_{\bar{0}}^{\mathbb{C}}$)
and $\Delta_{\bar{1}}$ the set of odd roots.
As in the non-super case, we have a root space decomposition:
\begin{align*}
 \mathfrak{g}^{\mathbb{C}}
  = \mathfrak{h}^{\mathbb{C}}\oplus \bigoplus_{\alpha\in \Delta}\mathfrak{g}_{\alpha}^{\mathbb{C}}.
\end{align*} 
We choose a distinguished Borel subsuperalgebra 
$\mathfrak{b}^{\mathbb{C}}$ such that the corresponding set $\Pi=\{\alpha_{1}, \dots , \alpha_{\ell}\}$ of simple roots 
contains a \emph{single} odd simple root. 
The Weyl vector $\rho=\rho_{\bar{0}}-\rho_{\bar{1}}$ of $\mathfrak{g}^{\mathbb{C}}$ is defined as half the 
graded sum of the positive roots, that is,
\begin{align}\label{Weyl}
 \rho_{\bar{0}}=\tfrac{1}{2}\sum_{\alpha\in\Delta^+_{\bar{0}}}\alpha,\qquad
 \rho_{\bar{1}}=\tfrac{1}{2}\sum_{\alpha\in\Delta^+_{\bar{1}}}\alpha,
\end{align}
where $\Delta^+=\Delta^+_{\bar{0}}\cup\Delta^+_{\bar{1}}$ denotes the set of positive roots.
The bilinear form $Q|_{\mathfrak{h}^{\mathbb{C}}}$ is non-degenerate, 
so we may identify $\mathfrak{h}^{\mathbb{C}}$ with $(\mathfrak{h}^{\mathbb{C}})^\vee$. 
A root $\alpha$ is called \emph{isotropic} if $Q(\alpha, \alpha)=0$. Note that isotropic roots are odd.

We now choose a subset $\Pi_{K}$ of $\Pi$ and denote by $\Delta_K\subseteq\Delta$ the system generated from $\Pi_K$. 
Then,
\begin{align*}
 \mathfrak{k}^{\mathbb{C}}
  := \mathfrak{h}^{\mathbb{C}} \oplus \bigoplus_{\alpha\in \Delta_K}\mathfrak{g}_{\alpha}^{\mathbb{C}}
\end{align*}
is a Lie subsuperalgebra of $\mathfrak{g}^{\mathbb{C}}$.
Note that $\mathfrak{k}^{\mathbb{C}}$ decomposes into its centre $Z(\mathfrak{k}^{\mathbb{C}})$ and the semisimple 
part $\mathfrak{k}_{\mathrm{ss}}^{\mathbb{C}}:=[\mathfrak{k}^{\mathbb{C}},\mathfrak{k}^{\mathbb{C}}]$. 
Let $\mathfrak{g}$ be the compact real form of $\mathfrak{g}^{\mathbb{C}}$ induced from a star operation $\ast$ 
on $\mathfrak{g}^{\mathbb{C}}$. Then, we can form the Lie subsupergroup $K=(K_0, \mathfrak{k})$ of $G$, where 
$\mathfrak{k}= \mathfrak{k}^{\mathbb{C}}\cap \mathfrak{g}$ is a Lie subsuperalgebra of $\mathfrak{g}$, and $K_0$ is 
the Lie subgroup of $G_0$ generated by $\mathfrak{k}_{\bar{0}}$. 
In this way, we obtain a homogeneous superspace $M=G/K$ called a \emph{flag supermanifold}. 

The tangent space $T_K(M)$ of $M$ at $K$ can be described as follows. Let $\Pi_{M}=\Pi- \Pi_{K}$
and $\Delta_{M}=\Delta-\Delta_{K}$. We have the $Q$-orthogonal vector-space decomposition 
$\mathfrak{g}^{\mathbb{C}}=\mathfrak{k}^{\mathbb{C}}\oplus \mathfrak{m}^{\mathbb{C}}$, and
we identify the complexified tangent space $T_K(M)^{\mathbb{C}}$ with 
\begin{align}\label{eq: mComp}
 \mathfrak{m}^{\mathbb{C}}=\bigoplus_{\alpha\in \Delta_M}\mathfrak{g}_{\alpha}^{\mathbb{C}}.
\end{align}
Similarly to \eqnref{eq: real}, we then have
\begin{align}\label{eq: mreal}
 \mathfrak{m}=\sspan_{\mathbb{R}}\big\{ X\in \mathfrak{m}_{\bar{0}}^{\mathbb{C}}\,|\, X^*=-X\big\}
 \oplus\sspan_{\mathbb{R}}\big\{ \sqrt{\imath}X\in \mathfrak{m}_{\bar{1}}^{\mathbb{C}}\,|\, X^*=-X\big\}, 
\end{align}
where $\ast$ is the star operation on $\mathfrak{g}^{\mathbb{C}}$.
For each $\alpha\in \Delta^+_M$, let $E_{\pm\alpha}$ span $\mathfrak{g}^{\mathbb{C}}_{\pm\alpha}$ and define
\begin{align}\label{AB}
 A_{\alpha}:= E_{\alpha} -E_{-\alpha}, \qquad B_{\alpha}:= \imath(E_{\alpha} + E_{-\alpha}),
\end{align}
noting that $(A_{\alpha})^{*}= -A_{\alpha}$ and $(B_{\alpha})^{*}= -B_{\alpha}$ under the star operation \eqref{eq: starA}.
It follows that $\mathfrak{m}$ has $\mathbb{R}$-basis  
\begin{align*}
 \{A_{\alpha}, B_{\alpha}\,|\, \alpha\in \Delta_{M}\cap \Delta_{\bar{0}}^+\}
 \cup\{\sqrt{\imath}A_{\alpha},\sqrt{\imath}B_{\alpha}\,|\, \alpha\in \Delta_{M}\cap \Delta_{\bar{1}}^+\}.
\end{align*}

In the non-super case, the above construction is encoded in a so-called painted Dynkin diagram, where the nodes of 
$\Pi_{M}$ are painted black while the other nodes of $\Pi$ remain white. In the super-setting, however, black nodes 
typically represent non-isotropic odd simple roots, so to avoid confusion, we shall \emph{circle} the nodes of $\Pi_{M}$, thereby 
obtaining a \emph{circled Dynkin diagram} of $G$. In addition, a node marked with a cross, $\otimes$, 
denotes an isotropic odd simple root, in accordance with the standard Dynkin diagram conventions for basic classical Lie superalgebras. 
The uncircled nodes of $\Pi$ form a simple root system of $\mathfrak{k}_{\mathrm{ss}}^{\mathbb{C}}$. 

Conversely, given a Dynkin diagram of $G$ with some nodes circled, we can construct the corresponding homogeneous 
superspace $M=G/K$ as follows. Let $\mathfrak{g}^{\mathbb{C}}$ be the complex simple Lie superalgebra characterised 
by the (uncircled) Dynkin diagram. The Lie superalgebra $\mathfrak{g}$ in $G=(G_0, \mathfrak{g})$ is then the compact 
real form of $\mathfrak{g}^{\mathbb{C}}$ induced from a specified star operation, whereas the Lie group 
$G_0$ is the unique simply connected Lie group generated by $\mathfrak{g}_{\bar{0}}$.
The semisimple part $K_{\mathrm{ss}}$ of $K$ is obtained from the Dynkin diagram by removing the nodes that are circled, 
while $K=K_{ss}\times \UU(1)^{n}$, where $n$ is the number of circled nodes of the Dynkin diagram. 
This determines the homogeneous superspace $M=G/K$.

\subsection{Compact real forms}
\label{sec: examflag}

To construct a flag supermanifold $M=(G_0, \mathfrak{g})/(K_0, \mathfrak{k})$ from a given circled Dynkin diagram, 
we first specify the compact real form $\mathfrak{g}$.  In the following, we recall Dynkin diagrams, star operations and 
compact real forms for the complex Lie superalgebras of types $A$ and $C$. It would be interesting to study 
(e.g.~using Vogan diagrams \cite{Chu13}) how the construction changes for \textit{non}-compact real forms.

\subsubsection{$\SU(m|n)$} \label{subsubsec: typeA}

The Lie superalgebra $\mathfrak{gl}(m|n)^{\mathbb{C}}$ has a basis consisting of matrix units $E_{ij}$, 
$i,j\in\{1,\ldots,m+n\}$, satisfying
\begin{align*}
 [E_{ij}, E_{k\ell}]= \delta_{jk}E_{i\ell}-(-1)^{[E_{ij}][E_{k\ell}]}\delta_{i\ell} E_{kj}, 
 \end{align*}
where $[E_{ij}]=0$ if $i,j\leq m$ or $i,j>m$, and $[E_{ij}]=1$ otherwise. Let $\mathfrak{h}^{\mathbb{C}}$ be the Cartan 
subalgebra spanned by all diagonal matrices, and denote by $\{\varepsilon_1,\ldots,\varepsilon_{m+n}\}$ the 
$(\mathfrak{h}^{\mathbb{C}})^\vee$-basis where $\varepsilon_i(E_{jj})=\delta_{ij}$ for all $i,j$. 
For convenience, we introduce $\delta_{\mu}:=\varepsilon_{m+\mu}$ for each $\mu=1,\ldots,n$.

Let $\mathfrak{sl}(m|n)^{\mathbb{C}}:=\{ X\in \mathfrak{gl}(m|n)^{\mathbb{C}}\,|\, \Str(X)=0\}$, and let 
$\mathfrak{b}^{\mathbb{C}}$ be 
the Borel subsuperalgebra spanned by all upper-triangular matrices. Then, the set of positive roots partitions as 
$\Delta^+=\Delta_{\bar{0}}^+\cup \Delta_{\bar{1}}^+$, where
\begin{align*}
 \Delta_{\bar{0}}^+=\{\varepsilon_i-\varepsilon_j, \delta_{\mu}-\delta_{\nu}\,|\, 1\leq i<j\leq m; 1\leq \mu<\nu\leq n \},\qquad
 \Delta_{\bar{1}}^+=\{\varepsilon_i-\delta_{\mu}\,|\, 1\leq i\leq m; 1\leq \mu\leq n\},
\end{align*}
with simple root system
\begin{align*} 
 \Pi=\{\alpha_i=\varepsilon_i-\varepsilon_{i+1},
   \alpha_m=\varepsilon_{m}-\delta_1,
   \alpha_{m+\mu}=\delta_{\mu}-\delta_{\mu+1}
   \,|\, 1\leq i\leq m-1; 1\leq \mu \leq n-1 \}
\end{align*}
and corresponding Dynkin diagram
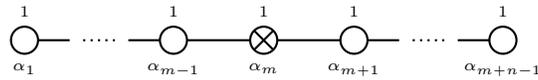
\begin{center}
\begin{tikzpicture}[scale=.6]
    \draw[thick] (0 cm,0) circle (.3cm);
    \draw[thick] (3.3 cm,0) circle (.3cm);
    \draw[thick] (5.3 cm,0) circle (.3cm);
    \draw[thick] (7.3 cm,0) circle (.3cm);
    \draw[thick] (10.6 cm,0) circle (.3cm);
     
    \draw[thick] (5.1 cm,-0.2 cm) -- (5.5cm, 0.2 cm); 
    \draw[thick] (5.1 cm, 0.2 cm) -- (5.5cm, -0.2 cm); 

    \draw[thick] (0.3 cm,0) -- (1 cm,0);
    \draw[thick] (2.3 cm,0) -- (3 cm,0);
    \draw[thick] (3.6 cm,0) -- (5 cm,0);
    \draw[thick] (5.6 cm,0) -- (7 cm,0);
    \draw[thick] (7.6 cm,0) -- (8.3  cm,0);
    \draw[thick] (9.6 cm,0) -- (10.3 cm,0);

    \draw[dotted,thick] (1.3 cm,0) -- (2.0 cm,0);
    \draw[dotted,thick] (8.6 cm,0) -- (9.3 cm,0);

    \draw (12 cm,-0.2) node[anchor=east]  {};
    \draw (0,-0.3) node[anchor=north]  {\tiny $\alpha_1$};
    \draw (3.3,-0.3) node[anchor=north]  {\tiny $\alpha_{m-1}$};
    \draw (5.3,-0.3) node[anchor=north]  {\tiny $\alpha_m$};
    \draw (7.3,-0.3) node[anchor=north]  {\tiny $\alpha_{m+1}$};
    \draw (10.6,-0.3) node[anchor=north]  {\tiny $\alpha_{m+n-1}$};  

    \draw (0,0.3) node[anchor=south]  {\tiny $1$};
    \draw (3.3,0.3) node[anchor=south]  {\tiny $1$};
    \draw (5.3,0.3) node[anchor=south]  {\tiny $1$};
    \draw (7.3,0.3) node[anchor=south]  {\tiny $1$};
    \draw (10.6,0.3) node[anchor=south]  {\tiny $1$};
\end{tikzpicture}
\end{center}
The numbers in the Dynkin diagram indicate the multiplicities in the decomposition of the highest root:
$\varepsilon_1-\delta_n=\alpha_1+\cdots+\alpha_{m+n-1}$.
The Weyl vector follows from
\begin{align*}
 2\rho=\sum_{i=1}^m(m-n-2i+1)\varepsilon_i+\sum_{\mu=1}^n(m+n-2\mu+1)\delta_\mu.
\end{align*}

Recall from~\cite{SNR77,GZ90b} that the so-called type (1) star operation on $\mathfrak{gl}(m|n)^{\mathbb{C}}$ 
is given by
\begin{align}\label{eq: starA}
 (E_{ij})^*=E_{ji}, \qquad i,j\in\{1,\ldots,m+n\}.
\end{align}
Using \eqnref{eq: real}, we then obtain a compact real form of $\mathfrak{gl}(m|n)^{\mathbb{C}}$, denoted by 
$\mathfrak{u}(m|n)$, with $\mathbb{R}$-basis
\begin{gather*}
 \{\imath E_{ii}\,|\, i=1,\ldots,m+n\}\cup
 \{A_{ij}, B_{ij}\,|\, 1\leq i<j\leq m\}\cup\{A_{ij}, B_{ij}\,|\, m+1\leq i<j\leq m+n\}
 \\[.1cm] 
 \cup\,\{\sqrt{\imath}A_{ij}, \sqrt{\imath} B_{ij}\,|\, 1\leq i\leq m;\,m+1\leq j\leq m+n\},
\end{gather*} 
where $A_{ij}= E_{ij}- E_{ji} $ and $B_{ij}= \imath (E_{ij}+E_{ji})$; cf.~\eqref{AB}.

The compact real form $\mathfrak{su}(m|n)$ of $\mathfrak{sl}(m|n)^{\mathbb{C}}$ consists of all matrices of 
$\mathfrak{u}(m|n)$ with supertrace zero, and we assume that $m\geq 3$ if $m=n$ (since for $m=n=2$, each odd root 
space has dimension $2$). The corresponding Lie supergroup is given by the Harish-Chandra pair 
\begin{align}\label{eq: HCSU}
 \SU(m|n)=(\SU(m)\times \SU(n)\times \UU(1), \mathfrak{su}(m|n)).
\end{align}

The Killing form on $\mathfrak{sl}(m|n)^{\mathbb{C}}$ is given by 
\begin{align}\label{eq: KillingA}
  B(X,Y)= 2(m-n) \Str(XY), \qquad X,Y\in \mathfrak{sl}(m|n)^{\mathbb{C}}.
\end{align}
In what follows, the non-degenerate supersymmetric even bilinear form $Q$ on 
$\mathfrak{sl}(m|n)^{\mathbb{C}}$ is defined by
\begin{align}\label{eq: slQ1}
 Q(X,Y)=-\Str(XY), \qquad \text{if $m\neq n$}.
\end{align}
For $m=n$, $B$ is identically zero, and the Lie superalgebra $\mathfrak{sl}(n|n)^{\mathbb{C}}$ has a one-dimensional 
centre spanned by the identity matrix $I_{2n}$. The corresponding quotient Lie superalgebra 
$\mathfrak{sl}(n|n)^{\mathbb{C}}/\mathbb{C}I_{2n}$ is simple, and we define $Q$ on
$\mathfrak{sl}(n|n)^{\mathbb{C}}=\mathbb{C}I_{2n}\oplus\mathfrak{sl}(n|n)^{\mathbb{C}}/\mathbb{C}I_{2n}$
by
\begin{align}\label{eq: slQ2}
 Q(I_{2n},I_{2n})=1,\qquad 
 Q(I_{2n}, X)=0,\qquad
 Q(X,Y)=-\Str(XY), \qquad 
 X,Y\in \mathfrak{sl}(n|n)/\mathbb{C}I_{2n}.
\end{align}

\begin{example}
\label{exam: flag}
Consider $G=\SU(3|2)$ and the circled Dynkin diagram
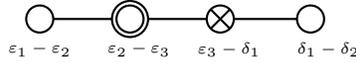
\begin{center}
\begin{tikzpicture}[scale=.6]
    \draw[thick] (0 cm,0) circle (.3cm);
     \draw[thick] (0.3 cm,0) -- +(1.3 cm,0);
    \draw[thick] (2 cm,0) circle (.3cm);
    \draw[thick] (2 cm,0) circle (.4cm);

    \draw[thick] (2.4 cm,0) -- +(1.3 cm,0);
     
     \draw[thick] (4 cm,0) circle (.3cm);
     
     \draw[thick] (3.8 cm,-0.2 cm) -- (4.2cm, 0.2 cm); 
     \draw[thick] (3.8 cm, 0.2 cm) -- (4.2cm, -0.2 cm); 
      \draw[thick] (4.3 cm,0) -- +(1.4 cm,0);    
     \draw[thick] (6 cm,0) circle (.3cm);

    \draw (0,-0.3) node[anchor=north]  {\tiny $\varepsilon_1-\varepsilon_2$};
    \draw (2.2,-0.3) node[anchor=north]  {\tiny $\varepsilon_2-\varepsilon_3$};
    \draw (4.2,-0.3) node[anchor=north]  {\tiny $\varepsilon_3-\delta_1$};
    \draw (6.4,-0.3) node[anchor=north]  {\tiny $\delta_1-\delta_2$}; 
\end{tikzpicture}
\end{center}
We now construct the corresponding flag supermanifold $M=G/K$. 

The Dynkin diagram corresponds to the complex Lie superalgebra 
$\mathfrak{g}^{\mathbb{C}}= \mathfrak{sl}(3|2)^{\mathbb{C}}$, while the Dynkin diagram obtained by removing the circled 
node corresponds to the Lie superalgebra 
$\mathfrak{k}^{\mathbb{C}}= \mathfrak{gl}(2)^{\mathbb{C}}\oplus \mathfrak{sl}(1|2)^{\mathbb{C}}$. 
The star operation gives rise to the real forms $\mathfrak{g}=\mathfrak{su}(3|2)$ and 
$\mathfrak{k}=\mathfrak{su}(2)\oplus\mathfrak{su}(1|2)\oplus\mathfrak{u}(1)\cong\mathfrak{u}(2)\oplus\mathfrak{su}(1|2)$. 
Thus, $K=\SU(2)\times \SU(1|2)\times \UU(1)\cong \UU(2)\times \SU(1|2)$, and the root systems of $G$ and $K$ 
are given by 
\begin{align*}
 \Delta&=\big\{\!\pm\!(\varepsilon_1-\varepsilon_2), \pm (\varepsilon_2-\varepsilon_3), 
 \pm(\varepsilon_1-\varepsilon_3), \pm (\delta_1-\delta_2) \big\}
 \cup\big\{\!\pm\!(\varepsilon_i-\delta_1), \pm (\varepsilon_i-\delta_2) \,|\, i=1,2,3 \big\},
 \\[.15cm] 
 \Delta_{K}&= \big\{\!\pm\!(\varepsilon_1-\varepsilon_2), \pm(\delta_1-\delta_2)\big\}
 \cup\big\{\!\pm\!(\varepsilon_3-\delta_1), \pm (\varepsilon_3-\delta_2)\big\},
\end{align*} 
respectively. Setting
\begin{align*}
 \Delta_{M}=\Delta-\Delta_{K}
 =\big\{\!\pm\!(\varepsilon_i-\varepsilon_3),\pm(\varepsilon_i-\delta_1), \pm (\varepsilon_i-\delta_2) \,|\, i=1,2\big\},
\end{align*}
we let $\Delta_{M}^+= \Delta^+\cap \Delta_{M}$ (respectively $\Delta_{M}^-=\Delta^-\cap \Delta_{M}$) denote the set of 
positive (respectively negative) roots of $\Delta_{M}$. It follows from \eqnref{eq: mComp} that 
$\mathfrak{m}^{\mathbb{C}}=\mathfrak{m}_{1}^{\mathbb{C}}\oplus \mathfrak{m}_{-1}^{\mathbb{C}}$, where 
\begin{align*}
 \mathfrak{m}_{\pm1}^{\mathbb{C}}=\bigoplus_{\alpha\in \Delta_M^\pm}\mathfrak{g}_{\alpha}^{\mathbb{C}}.
\end{align*}
Note that $\mathfrak{m}_{1}^{\mathbb{C}}$ and $\mathfrak{m}_{-1}^{\mathbb{C}}$ are irreducible 
$\ad_{\mathfrak{k}^{\mathbb{C}}}$-modules, and that they are dual to each other. The highest-weight vector of 
$\mathfrak{m}_{1}^{\mathbb{C}}$ is $E_{15}$, with highest weight $\varepsilon_1-\delta_2$. 
It follows from \eqref{eq: mreal} and \eqref{AB} that 
\begin{align*}
 \mathfrak{m}
  = \sum_{\alpha\in (\Delta_M^+)_{\bar{0}}} \mathbb{R}A_{\alpha}
    \ \oplus\!\sum_{\alpha\in (\Delta_M^+)_{\bar{0}}} \mathbb{R} B_{\alpha}
    \ \oplus\! \sum_{\alpha\in (\Delta_M^+)_{\bar{1}}} \mathbb{R} \sqrt{\imath}A_{\alpha}
    \  \oplus\! \sum_{\alpha\in (\Delta_M^+)_{\bar{1}}}\mathbb{R}\sqrt{\imath}B_{\alpha}.
\end{align*}
Since $\mathfrak{m}_{1}^{\mathbb{C}}$ is an irreducible $\ad_{\mathfrak{k}^{\mathbb{C}}}$-module, the real tangent 
space $\mathfrak{m}$ is irreducible as an $\ad_{\mathfrak{k}}$-module.

We extend the analysis above to general $\SU(m|n)$ with a single circled node in the Dynkin diagram 
in \secref{sec: flagA1}, and with two circled nodes in \secref{sec: flagA}. 
A similar analysis is carried out for $\SOSp(2|2n)$ in \secref{sec: flagC}, although for a single circled node only.
\end{example}

\subsubsection{$\mathrm{SOSp}(2|2n)$}
\label{subsubsec: typeC}

Let $V$ be a complex vector superspace with $\dim(V_{\bar{0}})=2$ and $\dim(V_{\bar{1}})=2n$, where $n\geq2$,
and let $\omega: V\times V\to \mathbb{C}$ be a non-degenerate 
supersymmetric even bilinear form whose matrix representation relative to a suitable basis for $V$ is given by
\begin{align*}
 \begin{pmatrix}
0 & 1 & 0& 0\\
1 & 0 & 0 & 0\\ 
0 & 0 & 0 & I_n\\
0 & 0 & -I_n & 0
\end{pmatrix}.
\end{align*}
The corresponding orthosymplectic superalgebra is defined by
\begin{align*}
 \mathfrak{osp}(2|2n)^{\mathbb{C}}
 :=\big\{ X\in \mathfrak{gl}(2|2n)^{\mathbb{C}} \,|\, \omega(Xu,v)+ (-1)^{[X][u]}\omega(u,Xv)=0 \text{ for all } u,v\in V \big\}.
\end{align*}
In the remainder of this subsection, we set $\mathfrak{g}^{\mathbb{C}}=\mathfrak{osp}(2|2n)^{\mathbb{C}}$.

Let $\mathfrak{h}^{\mathbb{C}}$ be the Cartan subalgebra spanned by the 
diagonal matrices of $\mathfrak{g}^{\mathbb{C}}$:
\begin{align}\label{eq: cartanC}
 \mathfrak{h}^{\mathbb{C}}=\big\{\mathrm{diag}(a,-a, f_1, \dots, f_{n}, -f_1, \dots, -f_n)
   \,|\, a,f_1,\ldots,f_n\in\mathbb{C}\big\}.
\end{align}
Relative to an arbitrary but fixed nonzero $H\in \mathfrak{h}^{\mathbb{C}}$, 
define the dual basis $\{\varepsilon, \delta_1, \dots, \delta_n\}\subset (\mathfrak{h}^{\mathbb{C}})^\vee$ by
\begin{align*} 
 \varepsilon(H)=a, \qquad \delta_j(H)=f_j,\qquad j=1,\ldots,n. 
\end{align*}
The root system partitions as $\Delta=\Delta^{+}\cup \Delta^-$, with the positive root system partitioning as 
$\Delta^+= \Delta^+_{\bar{0}}\cup \Delta^+_{\bar{1}}$, where
\begin{align*}
  \Delta_{\bar{0}}^+=\{ \delta_i\pm\delta_j\,|\, 1\leq i<j \leq n \}\cup \{2\delta_j\,|\, 1\leq j\leq n\}, 
    \qquad
  \Delta_{\bar{1}}^+=\{ \varepsilon\pm \delta_j\,|\, 1\leq j\leq n\}.
\end{align*}
The corresponding simple root system is given by
\begin{align*}
\Pi=\{ \alpha_1=\varepsilon-\delta_1, 
  \alpha_2=\delta_1-\delta_2, \dots, 
  \alpha_n=\delta_{n-1}-\delta_{n}, 
  \alpha_{n+1}=2\delta_n\},
\end{align*}
with Dynkin diagram
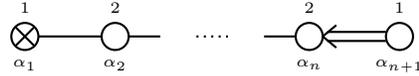
\begin{center}
\begin{tikzpicture}[scale=.6]
    \draw[thick] (0 cm,0) circle (.3cm);
    
    \draw[thick] (-0.2 cm,-0.2 cm) -- (0.2cm, 0.2 cm); 
     \draw[thick] (-0.2 cm,0.2 cm) -- (0.2cm, -0.2 cm); 

    \draw[thick] (0.3 cm,0) -- +(1.4 cm,0);
    \draw[thick] (2 cm,0) circle (.3cm);

     \draw[thick] (2.3 cm,0) -- (3  cm,0);
     \draw[dotted,thick] (3.8 cm,0) -- (4.5 cm,0);
    \draw[thick] (5.3 cm,0) -- (6 cm,0);

     \draw[thick] (6.3 cm,0) circle (.3cm);

     \draw[thick] (6.7 cm,.1cm) -- (8 cm,.1cm);
     \draw[thick] (6.7 cm,-.1cm) -- (8 cm,-.1cm);

     \draw[thick] (6.6 cm, 0cm) -- (6.9 cm,.25cm);
     \draw[thick] (6.6 cm, 0cm) -- (6.9 cm,-.25cm);
     \draw[thick] (8.3 cm,0) circle (.3cm);

   \draw (9.5,-0.3) node[anchor=east]  {};
    \draw (0,-0.3) node[anchor=north]  {\tiny $\alpha_1$};
     \draw (0,0.3) node[anchor=south]  {\tiny $1$};
    \draw (2,-0.3) node[anchor=north]  {\tiny $\alpha_2$};
     \draw (2,0.3) node[anchor=south]  {\tiny $2$};
    \draw (6.3,-0.3) node[anchor=north]  {\tiny $\alpha_n$};
     \draw (6.3,0.3) node[anchor=south]  {\tiny $2$};
    \draw (8.3,-0.3) node[anchor=north]  {\tiny $\alpha_{n+1}$};
     \draw (8.3,0.3) node[anchor=south]  {\tiny $1$};
\end{tikzpicture}
\end{center}
The numbers in the Dynkin diagram indicate the multiplicities in the decomposition of the highest root:
$\varepsilon+\delta_1=\alpha_1+2\alpha_2+\cdots+2\alpha_n+\alpha_{n+1}$. 
The Weyl vector is given by
\begin{align*}
 \rho=-n\varepsilon+\sum_{j=1}^n(n-j+1)\delta_j.
\end{align*}

For any root $\alpha\in \Delta$, let $E_{\alpha}\in\mathfrak{gl}(2|2n)^{\mathbb{C}}$ span the (one-dimensional) root space 
$\mathfrak{g}_{\alpha}^{\mathbb{C}}$, and let
\begin{align*}
 (E_{\alpha})^{*}=E_{-\alpha}\quad (\alpha\in \Delta^+), \qquad 
 (H)^*=H\quad (H\in \mathfrak{h}^{\mathbb{C}}),
\end{align*}
define a type (1) star operation~\cite{SNR77,GZ90a}.
In light of \eqnref{eq: real}, the compact real form $\mathfrak{osp}(2|2n)$ of $\mathfrak{g}^{\mathbb{C}}$ satisfies
\begin{align*}
 \mathfrak{osp}(2|2n)_{\bar{0}}
  =\mathfrak{so}(2,\mathbb{R})\oplus \mathfrak{sp}(2n)
  \cong\mathfrak{so}(2,\mathbb{R}) \oplus \mathfrak{u}(n, \mathbb{H}), \qquad
 \mathfrak{osp}(2|2n)_{\bar{1}}\cong \mathbb{H}^{n}.
\end{align*}
The connected orthosymplectic supergroup $\SOSp(2|2n)$ is given by the Harish-Chandra pair
\begin{align}\label{SOSp}
 \SOSp(2|2n)= (\SO(2,\mathbb{R})\times\Sp(2n), \mathfrak{osp}(2|2n)),
\end{align}
where $\Sp(2n)=\Sp(2n, \mathbb{C})\cap\UU(2n)\cong\UU(n,\mathbb{H})$ is our notation for the compact symplectic group.

As in the case of $\SU(m|n)$, we choose the non-degenerate $\ad_{\mathfrak{g}^{\mathbb{C}}}$-invariant supersymmetric 
even bilinear form to be
\begin{align}\label{eq: Qosp}
 Q(X,Y)=-\Str(XY), \qquad  X,Y\in \mathfrak{osp}(2|2n)^{\mathbb{C}},
\end{align}
and note that it is related to the Killing form $B$ as
\begin{align}\label{eq: BQrel}
 B(X,Y)=2n\,Q(X,Y).
\end{align}

\subsection{$\SU(m|n)$: one isotropy summand}
\label{sec: flagA1}

Let $G=\SU(m|n)$ be as in \eqref{eq: HCSU}. We will construct a flag supermanifold $M=G/K$ whose isotropy 
representation is an irreducible $\ad_{\mathfrak{k}}$-module, here denoted by $\mathfrak{m}_1$ 
(the index is for ease of comparison with results presented in \secref{sec: flagA} and \secref{sec: flagC}).
Retaining the notation from \secref{sec: diag_curv} and \secref{subsubsec: typeA}, 
a $G$-invariant graded Riemannian metric on $M$ is then of the form
\begin{align*}
 g= x_1 Q|_{\mathfrak{m}_1},\qquad
  x_1\in \mathbb{R}^\times.
\end{align*}
Using \eqref{eq: KillingA}, \eqref{eq: slQ1} and \eqref{eq: slQ2}, the coefficient in the expression \eqref{BbQ} is 
$b_1=2(m-n)$.

Generalising \exref{exam: flag}, we construct $M$ using a Dynkin diagram of $G$ with the $p^{\text{th}}\!$ node circled, 
where
\begin{align*}
 1\leq p\leq m+n-1,
\end{align*} 
and we divide our considerations into two cases, with $K$ given as follows:
\begin{align}\label{SUK1}
 K=\begin{cases}
    \UU(p|0)\times \SU(m-p|n), \qquad &p\leq m, \\[.1cm] 
    \SU(m|p-m)\times \UU(0|m+n-p) , \qquad & m<p.
  \end{cases} 
\end{align}
Note that $\SU(k|0)\cong\SU(0|k)\cong\SU(k)$ is an ordinary Lie group for $k\geq1$. However, writing $\UU(p|0)$, for 
example, keeps track of the concrete embedding of $\UU(p)$ inside $G$. 

Recall that $\mathfrak{m}^{\mathbb{C}}$ is the $Q$-orthogonal complement of $\mathfrak{k}^{\mathbb{C}}$ in 
$\mathfrak{g}^{\mathbb{C}}=\mathfrak{sl}(m|n)^{\mathbb{C}}$, where, with $K$ given in \eqref{SUK1},
\begin{align*}
 \mathfrak{k}^{\mathbb{C}}=\begin{cases}
   \mathfrak{gl}(p|0)^{\mathbb{C}}\oplus\mathfrak{sl}(m-p|n)^{\mathbb{C}}, 
   \qquad &p\leq m, 
	\\[.15cm] 
   \mathfrak{sl}(m|p-m)^{\mathbb{C}}\oplus\mathfrak{gl}(0|m+n-p)^{\mathbb{C}}, 
   \qquad & m<p.
  \end{cases} 
\end{align*}
First, let us focus on $p\leq m$ and introduce $\hat{\mathfrak{g}}^{\mathbb{C}}= \mathfrak{gl}(m|n)^{\mathbb{C}}$ and
$\hat{\mathfrak{k}}^{\mathbb{C}}=\mathfrak{gl}(p|0)^{\mathbb{C}}\oplus\mathfrak{gl}(m-p|n)^{\mathbb{C}}$.
Then, $\hat{\mathfrak{g}}^{\mathbb{C}}=\hat{\mathfrak{k}}^{\mathbb{C}} \oplus\mathfrak{m}^{\mathbb{C}}$ is a
$Q$-orthogonal decomposition, and the decomposition of $\mathfrak{m}^{\mathbb{C}}$ as an 
$\ad_{\hat{\mathfrak{k}}^{\mathbb{C}}}$-module is the same as the decomposition of $\mathfrak{m}^{\mathbb{C}}$ as 
an $\ad_{\mathfrak{k}^{\mathbb{C}}}$-module. So, we may (and will) continue our analysis using 
$\hat{\mathfrak{g}}^{\mathbb{C}}$ and $\hat{\mathfrak{k}}^{\mathbb{C}}$ instead of
$\mathfrak{g}^{\mathbb{C}}$ and $\mathfrak{k}^{\mathbb{C}}$.
For all applicable $r,s$, we use $\mu_{r|s}$ to denote the natural representation of $\mathfrak{gl}(r|s)^{\mathbb{C}}$, 
with dual representation $\bar{\mu}_{r|s}$.
In accord with the regular embedding $\hat{\mathfrak{k}}^{\mathbb{C}}\subseteq\hat{\mathfrak{g}}^{\mathbb{C}}$,
we denote by $\varepsilon_1, \dots, \varepsilon_p$ the weights of $\mu_{p|0}$, 
and by $\varepsilon_{p+1}, \dots, \varepsilon_m, \delta_1, \dots, \delta_n$ the weights of $\mu_{m-p|n}$. 
It then follows that $\mathfrak{m}^{\mathbb{C}}_1$ is indeed an irreducible $\ad_{\mathfrak{k}^{\mathbb{C}}}$-module,
with
\begin{align*}
 \mathfrak{m}^{\mathbb{C}}=\mathfrak{m}^{\mathbb{C}}_{-1}\oplus\mathfrak{m}^{\mathbb{C}}_1,\qquad
 \mathfrak{m}^{\mathbb{C}}_1=\begin{cases}
  \mu_{p|0}\otimes\bar{\mu}_{m-p|n},
   \qquad &p\leq m, 
	\\[.15cm]
    \mu_{m|p-m}\otimes\bar{\mu}_{0|m+n-p},
   \qquad &p>m, 
 \end{cases}
\end{align*}
where $\mathfrak{m}^{\mathbb{C}}_{-1}$ is the dual $\ad_{\mathfrak{k}^{\mathbb{C}}}$-module of $\mathfrak{m}^{\mathbb{C}}_1$,
and where we have included the expression for $\mathfrak{m}^{\mathbb{C}}_1$ obtained similarly in the case $p>m$.
For all $p$, the corresponding highest weight of $\mathfrak{m}^{\mathbb{C}}_1$ is $\Lambda_1=\varepsilon_1-\delta_n$.
We also have
\begin{align*}
 d_1:=\sdim(\mathfrak{m}_1)=2\sdim_{\mathbb{C}}(\mathfrak{m}_1^{\mathbb{C}})
 =\begin{cases}
  2p(m-p-n),\quad &p\leq m,\\[.15cm]
  -2(2m-p)(m+n-p),\quad &p>m.
  \end{cases}
\end{align*}

The eigenvalue of $C_{\mathfrak{m}_1, Q|_{\mathfrak{k}}}$ is given by
$c_1=-(\Lambda_1+2\rho_{\mathfrak{k}^{\mathbb{C}}}, \Lambda_1)$,
where the Weyl vector $\rho_{\mathfrak{k}^{\mathbb{C}}}$ of $\mathfrak{k}^{\mathbb{C}}$ (see \eqref{Weyl}) follows from
\begin{align*}
 2\rho_{\mathfrak{k}^{\mathbb{C}}}
   =\sum_{i=1}^p(p-2i+1)\varepsilon_i
   +\sum_{i=p+1}^m(m+p-n-2i+1)\varepsilon_i
   +\sum_{\mu=1}^n(m-p+n-2\mu+1)\delta_\mu
\end{align*}
for $p\leq m$, while for $p>m$, we have
\begin{align*}
 2\rho_{\mathfrak{k}^{\mathbb{C}}}
  &=\sum_{i=1}^m(2m-p-2i+1)\varepsilon_i
   +\sum_{\mu=1}^{p-m}(p-2\mu+1)\delta_\mu
   +\sum_{\mu=p-m+1}^n(-m+p+n-2\mu+1)\delta_\mu.
\end{align*}
To compute $c_1$, we use that
\begin{align*}
 (\lambda, \mu)=\sum_{i=1}^{m+n}\lambda(E_{ii})\mu(\bar{E}_{ii}),\qquad \lambda,\mu\in(\mathfrak{h}^{\mathbb{C}})^\vee,
\end{align*}
where the right $Q$-dual matrices, $\bar{E}_{11},\ldots,\bar{E}_{m+n,m+n}$, to the diagonal matrix units, 
$E_{11},\ldots,E_{m+n,m+n}$, are given by $\bar{E}_{ii}=-E_{ii}$ if $i\leq m$ and $\bar{E}_{ii}=E_{ii}$ if $i>m$.
We thus have $(\varepsilon_1,\varepsilon_1)=-1$ and $(\delta_n,\delta_n)=1$, in particular.
It follows that, for $p\leq m$,
\begin{align*}
 c_1
  =-(\Lambda_1+2\rho_{\mathfrak{k}^{\mathbb{C}}},\Lambda_1)
  =-(2\rho_{\mathfrak{k}^{\mathbb{C}}},\varepsilon_1-\delta_n)
  =-(p-1)(\varepsilon_1,\varepsilon_1)+(m-p-n+1)(\delta_n,\delta_n)=m-n,
\end{align*}
and similarly for $p>m$,
\begin{align*}
 c_1=-(2m-p-1)(\varepsilon_1,\varepsilon_1)-(m-p+n-1)(\delta_n,\delta_n)=m-n.
\end{align*}
In both cases, we thus have
\begin{align*}
 b_1=2c_1,
\end{align*}
so by \propref{prop: StrCas}, the single structure constant vanishes: $[111]=0$.

The Ricci coefficient $r_1$ follows from \eqref{rdnot0} if $d_1\neq0$, and from \propref{prop: ricexc} if $d_1=0$.
In both situations, we have $r_1=c_1$, so the Einstein equation \eqref{eq: Ein} simply reads
\begin{align*}
 c_1=cx_1.
\end{align*}
For $m\neq n$, there is a unique solution up to scaling, namely $(g,c)=((m-n)Q|_{\mathfrak{m}_1},1)$, and the metric is 
positive (see text following \eqref{eq: metric}) with positive $c$ if $m>n$, and positive with negative $c$ if $m<n$.
For $m=n$, there is a continuous family of solutions, namely $(g,c)=(x_1Q|_{\mathfrak{m}_1},0)$, so the metric is Ricci-flat.

Circling a single node in the Dynkin diagram of $\SU(m|n)$ results in a flag supermanifold with one isotropy summand. 
In the next section, we will show that circling two nodes leads to three isotropy summands.

\subsection{$\SU(m|n)$: three isotropy summands}
\label{sec: flagA}

Let $G=\SU(m|n)$ be as in \eqref{eq: HCSU}. We will construct a flag supermanifold $M=G/K$ whose isotropy 
representation decomposes into three inequivalent nonzero irreducible $\ad_{\mathfrak{k}}$-modules:
\begin{align*}
 \mathfrak{m}= \mathfrak{m}_1\oplus \mathfrak{m}_2\oplus \mathfrak{m}_3.
\end{align*}
Retaining the notation from \secref{sec: diag_curv} and \secref{subsubsec: typeA}, 
a $G$-invariant graded Riemannian metric on $M$ is of the form
\begin{align}\label{gxxx}
 g= x_1 Q|_{\mathfrak{m}_1}
  + x_2 Q|_{\mathfrak{m}_2}
  +x_3 Q|_{\mathfrak{m}_3},\qquad 
  x_1,x_2,x_3\in \mathbb{R}^\times.
\end{align}
We will use $(x_1,x_2,x_3)$ as a shorthand notation for the sum in \eqref{gxxx}. Using \eqref{eq: KillingA}, \eqref{eq: slQ1} 
and \eqref{eq: slQ2}, the coefficients in the expression \eqref{BbQ} are $b_1=b_2=b_3=2(m-n)$. In the following, we will 
denote this common value simply by $b$:
\begin{align*}
 b=2(m-n).
\end{align*}

We construct $M$ using a Dynkin diagram of $G$ with two circled nodes, with $p$ and $q$,
\begin{align*}
 1\leq p<q\leq m+n-1,
\end{align*} 
denoting the indices of the circled nodes, and we divide our considerations into three cases, as depicted 
in \figref{fig: CDynA}, with $K$ given as follows:
\begin{align}\label{SUK}
 K=\begin{cases}
  	 \UU(p|0)\times \UU(q-p|0)\times \SU(m-q|n), \qquad &p<q\leq m,\qquad (\text{Case 1}), \\[.1cm] 
  	 \UU(p|0)\times \SU(m-p|q-m)\times  \UU(0|m+n-q), \qquad&p<m< q,\qquad (\text{Case 2}),\\[.1cm] 
   \SU(m|p-m)\times  \UU(0|q-p) \times \UU(0|m+n-q) , \qquad & m\leq p<q,\qquad (\text{Case 3}).
  \end{cases} 
\end{align}
Note that the odd isotropic node is circled if either $q=m$ (Case 1) or $p=m$ (Case 3).

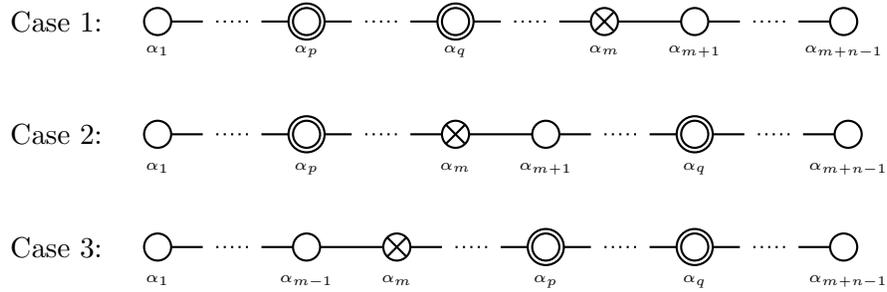
\begin{figure}[ht]
\begin{center}
  \begin{tikzpicture}[scale=.6]
   \draw (-1,0) node[anchor=east]  {Case 1:};
    \draw[thick] (0 cm,0) circle (.3cm);
    \draw[thick] (0.3 cm,0) -- (1  cm,0);
    \draw[dotted,thick] (1.3 cm,0) -- (2 cm,0);
    \draw[thick] (2.3 cm,0) -- (2.9 cm,0);

    \draw[thick] (3.3 cm,0) circle (.3cm);
    \draw[thick] (3.3 cm,0) circle (.4cm);
     
    \draw[thick] (3.7 cm,0) -- (4.3 cm,0);
    \draw[dotted,thick] (4.6 cm,0) -- (5.3 cm,0);

    \draw[thick] (5.6 cm,0) -- (6.2 cm,0);
    \draw[thick] (6.6 cm,0) circle (.3cm);
    \draw[thick] (6.6 cm,0) circle (.4cm);
    \draw[thick] (7 cm,0) -- (7.6 cm,0);

    \draw[dotted,thick] (7.9 cm,0) -- (8.6 cm,0);
    \draw[thick] (8.9 cm,0) -- (9.6 cm,0);
     \draw[thick] (9.9 cm,0) circle (.3cm);
     \draw[thick] (9.7 cm,-0.2 cm) -- (10.1cm, 0.2 cm); 
     \draw[thick] (9.7 cm, 0.2 cm) -- (10.1cm, -0.2 cm); 
     \draw[thick] (10.2 cm,0) -- (11.6 cm,0);
     \draw[thick] (11.9 cm,0) circle (.3cm);
     \draw[thick] (12.2 cm,0) -- (12.9 cm,0);     

    \draw[dotted,thick] (13.2 cm,0) -- (13.9 cm,0);
    \draw[thick] (14.2 cm,0) -- (14.9 cm,0); 
     \draw[thick] (15.2 cm,0) circle (.3cm);
    \draw (0,-0.3) node[anchor=north]  {\tiny $\alpha_1$};
    \draw (3.3,-0.3) node[anchor=north]  {\tiny $\alpha_p$};
    \draw (6.6,-0.3) node[anchor=north]  {\tiny $\alpha_q$};
    \draw (9.9,-0.3) node[anchor=north]  {\tiny $\alpha_m$};
    \draw (11.9,-0.3) node[anchor=north]  {\tiny $\alpha_{m+1}$};  
    \draw (15.2,-0.3) node[anchor=north]  {\tiny $\alpha_{m+n-1}$};  

  \begin{scope}[shift={(0,-2.5cm)}]
     \draw (-1,0) node[anchor=east]  {Case 2:};
    \draw[thick] (0 cm,0) circle (.3cm);
    \draw[thick] (0.3 cm,0) -- (1  cm,0);
    \draw[dotted,thick] (1.3 cm,0) -- (2 cm,0);
    \draw[thick] (2.3 cm,0) -- (2.9 cm,0);

    \draw[thick] (3.3 cm,0) circle (.3cm);
    \draw[thick] (3.3 cm,0) circle (.4cm);

    \draw[thick] (3.7 cm,0) -- (4.3 cm,0);
    \draw[dotted,thick] (4.6 cm,0) -- (5.3 cm,0);

    \draw[thick] (5.6 cm,0) -- (6.3 cm,0);
    \draw[thick] (6.6 cm,0) circle (.3cm);
     \draw[thick] (6.4 cm,-0.2 cm) -- (6.8cm, 0.2 cm); 
     \draw[thick] (6.4 cm, 0.2 cm) -- (6.8cm, -0.2 cm); 
    \draw[thick] (6.9 cm,0) -- +(1.4 cm,0);

    \draw[thick] (8.6 cm,0) circle (.3cm);
    \draw[thick] (8.9 cm,0) -- (9.6 cm,0);

    \draw[dotted,thick] (9.9 cm,0) -- (10.6 cm,0);
    \draw[thick] (10.9 cm,0) -- (11.5 cm,0);

    \draw[thick] (11.9 cm,0) circle (.3cm);
    \draw[thick] (11.9 cm,0) circle (.4cm);
    \draw[thick] (12.3 cm,0) -- (13 cm,0);
   \draw[dotted,thick] (13.3 cm,0) -- (14 cm,0);
    \draw[thick] (14.3 cm,0) -- (15 cm,0);
    \draw[thick] (15.3 cm,0) circle (.3cm);

    \draw (0,-0.4) node[anchor=north]  {\tiny $\alpha_1$};
    \draw (3.3,-0.4) node[anchor=north]  {\tiny $\alpha_p$};
    \draw (6.6,-0.4) node[anchor=north]  {\tiny $\alpha_m$};
    \draw (8.6,-0.4) node[anchor=north]  {\tiny $\alpha_{m+1}$};
    \draw (11.9,-0.4) node[anchor=north]  {\tiny $\alpha_q$};  
    \draw (15.3,-0.4 ) node[anchor=north]  {\tiny $\alpha_{m+n-1}$};  
  \end{scope}

  \begin{scope}[shift={(0,-5cm)}]
     \draw (-1,0) node[anchor=east]  {Case 3:};
    \draw[thick] (0 cm,0) circle (.3cm);
    \draw[thick] (0.3 cm,0) -- (1  cm,0);
    \draw[dotted,thick] (1.3 cm,0) -- (2 cm,0);
    \draw[thick] (2.3 cm,0) -- (3 cm,0);

    \draw[thick] (3.3 cm,0) circle (.3cm);

    \draw[thick] (3.6 cm,0) -- (5 cm,0);
    \draw[thick] (5.3 cm,0) circle (.3cm);
    \draw[thick] (5.1 cm,-0.2 cm) -- (5.5cm, 0.2 cm); 
    \draw[thick] (5.1 cm, 0.2 cm) -- (5.5cm, -0.2 cm); 
    
    \begin{scope}[shift={(5.3,0)}]
    \draw[thick] (0.3 cm,0) -- (1  cm,0);
    \draw[dotted,thick] (1.3 cm,0) -- (2 cm,0);
    \draw[thick] (2.3 cm,0) -- (2.9 cm,0);

    \draw[thick] (3.3 cm,0) circle (.3cm);
    \draw[thick] (3.3 cm,0) circle (.4cm);
     
    \draw[thick] (3.7 cm,0) -- (4.3 cm,0);
    \draw[dotted,thick] (4.6 cm,0) -- (5.3 cm,0);

    \draw[thick] (5.6 cm,0) -- (6.2 cm,0);
    \draw[thick] (6.6 cm,0) circle (.3cm);
    \draw[thick] (6.6 cm,0) circle (.4cm);
    \draw[thick] (7 cm,0) -- (7.6 cm,0);

    \draw[dotted,thick] (7.9 cm,0) -- (8.6 cm,0);
    \draw[thick] (8.9 cm,0) -- (9.6 cm,0);
     \draw[thick] (9.9 cm,0) circle (.3cm);
    \end{scope}

    \draw (0,-0.4) node[anchor=north]  {\tiny $\alpha_1$};
    \draw (3.3,-0.4) node[anchor=north]  {\tiny $\alpha_{m-1}$};
    \draw (5.3,-0.4) node[anchor=north]  {\tiny $\alpha_m$};
    \draw (8.6,-0.4) node[anchor=north]  {\tiny $\alpha_p$};
    \draw (11.9,-0.4) node[anchor=north]  {\tiny $\alpha_q$};  
    \draw (15.3,-0.4 ) node[anchor=north]  {\tiny $\alpha_{m+n-1}$};  
  \end{scope}

  \end{tikzpicture}
  \caption{Dynkin diagrams of $\SU(m|n)$ with two circled nodes.}
  \label{fig: CDynA}
\end{center}
\end{figure}

In \secref{sec: Aclass}, we use results obtained in \secref{sec: diag_curv} to classify the $G$-invariant Einstein metrics 
\eqref{gxxx} on $M=G/K$ constructed from a circled Dynkin diagram of $G=\SU(m|n)$, with $K$ given in \eqref{SUK}.
In preparation for these classification results, we examine the structure of the isotropy representation of $M$ 
in \secref{sec: isosum}, and confirm that there are precisely three inequivalent isotropy summands in $T_KM$ for $M$ 
constructed as above. We subsequently determine the structure constants of $M$ in \secref{sec: strconsA}.

\subsubsection{Isotropy summands}
\label{sec: isosum}

To describe the $\ad_{\mathfrak{k}}$-module structure of $\mathfrak{m}$, we first characterise the 
$\ad_{\mathfrak{k}^{\mathbb{C}}}$-module structure of $\mathfrak{m}^{\mathbb{C}}$. Recall that 
$\mathfrak{m}^{\mathbb{C}}$ is the $Q$-orthogonal complement of $\mathfrak{k}^{\mathbb{C}}$ in 
$\mathfrak{g}^{\mathbb{C}}=\mathfrak{sl}(m|n)^{\mathbb{C}}$, where, with $K$ given in \eqref{SUK},
\begin{align}\label{eq: compk}
 \mathfrak{k}^{\mathbb{C}}=\begin{cases}
   \mathfrak{gl}(p|0)^{\mathbb{C}}\oplus \mathfrak{gl}(q-p|0)^{\mathbb{C}}\oplus\mathfrak{sl}(m-q|n)^{\mathbb{C}}, 
   \qquad &p<q\leq m, 
 \\[.15cm] 
  \mathfrak{gl}(p|0)^{\mathbb{C}}\oplus  \mathfrak{sl}(m-p|q-m)^{\mathbb{C}} \oplus  \mathfrak{gl}(0|m+n-q)^{\mathbb{C}}, 
	\qquad& p< m< q,
	\\[.15cm] 
   \mathfrak{sl}(m|p-m)^{\mathbb{C}} \oplus \mathfrak{gl}(0|q-p)^{\mathbb{C}} \oplus \mathfrak{gl}(0|m+n-q)^{\mathbb{C}}, 
   \qquad & m\leq p<q.
  \end{cases} 
\end{align}
For the root-space decomposition \eqref{eq: mComp} of $\mathfrak{m}^{\mathbb{C}}$, we introduce
\begin{align*}
  \mathfrak{m}_{k_1,k_2}^{\mathbb{C}}:= \bigoplus_{\alpha\in \Delta_{k_1,k_2}} \mathfrak{g}_{\alpha}^{\mathbb{C}},\qquad
  \Delta_{k_1,k_2}:= \Big\{ \sum_{i=1}^{m+n-1}c_i  \alpha_i \in \Delta_M \,|\, c_p=k_1, c_{q}=k_2\Big\},\qquad
  (k_1,k_2)\in I_{\mathfrak{m}^{\mathbb{C}}},
\end{align*}
where $I_{\mathfrak{m}^{\mathbb{C}}}=\{(-1,-1),(0,-1),(-1,0),(1,0),(0,1),(1,1)\}$. If $\Delta_{k_1,k_2}$ is non-empty, 
then $\mathfrak{m}_{k_1,k_2}^{\mathbb{C}}$ is an $\ad_{\mathfrak{k}^{\mathbb{C}}}$-invariant subspace of 
$\mathfrak{m}^{\mathbb{C}}$.
\begin{proposition}\label{prop: mdecA}
Let $\mathfrak{k}^{\mathbb{C}}$ be one of  the Lie superalgebras in \eqref{eq: compk}. As an 
$\ad_{\mathfrak{k}^{\mathbb{C}}}$-module, $\mathfrak{m}^{\mathbb{C}}$ admits the $Q$-orthogonal decomposition
\begin{align*} 
  \mathfrak{m}^{\mathbb{C}}=\mathfrak{m}_{-1,-1}^{\mathbb{C}} 
  \oplus \mathfrak{m}_{0,-1}^{\mathbb{C}}
  \oplus \mathfrak{m}_{-1,0}^{\mathbb{C}} 
  \oplus \mathfrak{m}_{1,0}^{\mathbb{C}}
  \oplus \mathfrak{m}_{0,1}^{\mathbb{C}} 
  \oplus \mathfrak{m}_{1,1}^{\mathbb{C}}.
\end{align*}
Moreover, $\mathfrak{m}_{k_1,k_2}^{\mathbb{C}}$ and $\mathfrak{m}_{-k_1,-k_2}^{\mathbb{C}}$ are irreducible and dual 
for each $(k_1,k_2)\in\{(1,0),(0,1),(1,1)\}$.
\end{proposition}
\begin{proof}
We present a proof in Case 1 ($p< q\leq m$); the other two cases can be treated similarly.
Let $\hat{\mathfrak{g}}^{\mathbb{C}}= \mathfrak{gl}(m|n)^{\mathbb{C}}$ and
$\hat{\mathfrak{k}}^{\mathbb{C}}=\mathfrak{gl}(p|0)^{\mathbb{C}}\oplus \mathfrak{gl}(q-p|0)^{\mathbb{C}}\oplus\mathfrak{gl}(m-q|n)^{\mathbb{C}}$.
Then, $\hat{\mathfrak{g}}^{\mathbb{C}}=\hat{\mathfrak{k}}^{\mathbb{C}} \oplus\mathfrak{m}^{\mathbb{C}}$ is a
$Q$-orthogonal decomposition, and the decomposition of $\mathfrak{m}^{\mathbb{C}}$ as an 
$\ad_{\hat{\mathfrak{k}}^{\mathbb{C}}}$-module is the same as the decomposition of $\mathfrak{m}^{\mathbb{C}}$ as 
an $\ad_{\mathfrak{k}^{\mathbb{C}}}$-module. So, we may (and will) continue our analysis using 
$\hat{\mathfrak{g}}^{\mathbb{C}}$ and $\hat{\mathfrak{k}}^{\mathbb{C}}$ instead of
$\mathfrak{g}^{\mathbb{C}}$ and $\mathfrak{k}^{\mathbb{C}}$.

In accord with the regular embedding $\hat{\mathfrak{k}}^{\mathbb{C}}\subseteq\hat{\mathfrak{g}}^{\mathbb{C}}$,
we denote by $\varepsilon_1, \dots, \varepsilon_p$ the weights of the natural representation~$\mu_{p|0}$, 
by $\varepsilon_{p+1}, \dots, \varepsilon_q$ the weights of $\mu_{q-p|0}$, 
and by $\varepsilon_{q+1}, \dots, \varepsilon_m, \delta_1, \dots, \delta_n$ the weights of $\mu_{m-q|n}$. We then have
\begin{align*}
 \hat{\mathfrak{k}}^{\mathbb{C}}\cong(\mu_{p|0} \otimes \bar{\mu}_{p|0}) 
 \oplus (\mu_{q-p|0} \otimes \bar{\mu}_{q-p|0}) 
 \oplus (\mu_{m-q|n} \otimes \bar{\mu}_{m-q|n}).
\end{align*}
Restricting to $\hat{\mathfrak{k}}^{\mathbb{C}}$, we have 
$\mu_{m|n}|_{\hat{\mathfrak{k}}^{\mathbb{C}}}\cong \mu_{p|0}\oplus \mu_{q-p|0}\oplus \mu_{m-q|n}$
and hence the $\hat{\mathfrak{k}}^{\mathbb{C}}$-module isomorphisms
\begin{align*}
 \hat{\mathfrak{g}}^{\mathbb{C}}|_{\ad_{\hat{\mathfrak{k}}^{\mathbb{C}}}}
 \cong\mu_{m|n}|_{\hat{\mathfrak{k}}^{\mathbb{C}}} \otimes \bar{\mu}_{m|n}|_{\hat{\mathfrak{k}}^{\mathbb{C}}}
 \cong(\mu_{p|0}\oplus\mu_{q-p|0}\oplus\mu_{m-q|n})\otimes(\bar{\mu}_{p|0}\oplus\bar{\mu}_{q-p|0}\oplus\bar{\mu}_{m-q|n}).
\end{align*}
Expanding this out, $\hat{\mathfrak{k}}^{\mathbb{C}}$ arises as direct summand, with remainder
\begin{align*}
 \mathfrak{m}^{\mathbb{C}}=\mathfrak{m}_{-1,-1}^{'\mathbb{C}} \oplus \mathfrak{m}_{0,-1}^{'\mathbb{C}}
   \oplus\mathfrak{m}_{-1,0}^{'\mathbb{C}} \oplus \mathfrak{m}_{1,0}^{'\mathbb{C}}\oplus \mathfrak{m}_{0,1}^{'\mathbb{C}} 
   \oplus\mathfrak{m}_{1,1}^{'\mathbb{C}},
\end{align*}
where
\begin{align*}
 &\mathfrak{m}^{'\mathbb{C}}_{-1,-1}\cong\bar{\mu}_{p|0}\otimes \mu_{m-q|n}, 
 &&\mathfrak{m}^{'\mathbb{C}}_{0,-1}\cong\bar{\mu}_{q-p|0}\otimes \mu_{m-q|n}, 
 && \mathfrak{m}^{'\mathbb{C}}_{-1,0}\cong\bar{\mu}_{p|0}\otimes \mu_{q-p|0},
 \\[.1cm] 
 &\mathfrak{m}^{'\mathbb{C}}_{1,0}\cong\mu_{p|0}\otimes \bar{\mu}_{q-p|0}, 
 &&\mathfrak{m}^{'\mathbb{C}}_{0,1}\cong\mu_{q-p|0}\otimes\bar{\mu}_{m-q|n}, 
 && \mathfrak{m}^{'\mathbb{C}}_{1,1}\cong\mu_{p|0}\otimes \bar{\mu}_{m-q|n}. 
\end{align*}
It follows that $\mathfrak{m}_{k_1,k_2}^{'\mathbb{C}}$ and $\mathfrak{m}_{-k_1,-k_2}^{'\mathbb{C}}$ are irreducible and 
dual to each other.

We now show that $\mathfrak{m}^{\mathbb{C}}_{k_1, k_2}=\mathfrak{m}^{'\mathbb{C}}_{k_1,k_2}$ for all relevant 
$k_1,k_2$, which then concludes the proof. Note that the weights of $\mathfrak{m}^{'\mathbb{C}}_{1,0}$ are 
\begin{align*} 
 \varepsilon_i-\varepsilon_j =\alpha_i+\cdots +\alpha_{j-1}, \qquad 1\leq i\leq p< j\leq q,
\end{align*}
so the coefficient of $\alpha_p$ is $1$, while the coefficient of $\alpha_q$ is $0$. 
This implies that $\mathfrak{m}^{'\mathbb{C}}_{1,0}\subseteq \mathfrak{m}^{\mathbb{C}}_{1, 0}$.
Similarly, it follows that $\mathfrak{m}^{'\mathbb{C}}_{k_1,k_2}\subseteq \mathfrak{m}^{\mathbb{C}}_{k_1, k_2}$ 
for all $(k_1,k_2)\in I_{\mathfrak{m}^{\mathbb{C}}}$, so
\begin{align*}
    \bigoplus_{k_1,k_2} \mathfrak{m}_{k_1,k_2}^{\mathbb{C}} \subseteq \mathfrak{m}^{\mathbb{C}} 
    =\bigoplus_{k_1,k_2} \mathfrak{m}_{k_1,k_2}^{'\mathbb{C}} \subseteq  
       \bigoplus_{k_1,k_2} \mathfrak{m}_{k_1,k_2}^{\mathbb{C}},
\end{align*}
where the direct sums are over $(k_1,k_2)\in I_{\mathfrak{m}^{\mathbb{C}}}$. 
It follows that $\mathfrak{m}^{\mathbb{C}}_{k_1, k_2}=\mathfrak{m}^{'\mathbb{C}}_{k_1,k_2}$ 
for all $(k_1,k_2)\in I_{\mathfrak{m}^{\mathbb{C}}}$. 
\end{proof}
\begin{corollary}\label{coro: mrep}
Let the notation be as in \propref{prop: mdecA}. Then, the following holds.
\begin{itemize}
   \item  In Case 1 $(p<q\leq m)$,
    \begin{align*} 
    \mathfrak{m}_{1,0}^{\mathbb{C}} \cong \mu_{p|0}\otimes \bar{\mu}_{q-p|0}, \qquad 
    \mathfrak{m}_{0,1}^{\mathbb{C}}\cong \mu_{q-p|0}\otimes \bar{\mu}_{m-q|n}, \qquad 
    \mathfrak{m}_{1,1}^{\mathbb{C}}\cong \mu_{p|0}\otimes \bar{\mu}_{m-q|n},  
    \end{align*}
    as $\ad_{\mathfrak{k}^{\mathbb{C}}}$-modules, with corresponding highest weights
\begin{align*}
 \Lambda_{1,0}=\varepsilon_1-\varepsilon_q, \qquad 
 \Lambda_{0,1}= \varepsilon_{p+1}-\delta_n, \qquad 
 \Lambda_{1,1}= \varepsilon_1-\delta_n.
\end{align*}
 \item  In Case 2 $(p< m <q)$,
    \begin{align*}
     \mathfrak{m}_{1,0}^{\mathbb{C}} \cong \mu_{p|0}\otimes \bar{\mu}_{m-p|q-m}, \qquad 
     \mathfrak{m}_{0,1}^{\mathbb{C}}\cong \mu_{m-p|q-m}\otimes \bar{\mu}_{0|m+n-q}, \qquad 
     \mathfrak{m}_{1,1}^{\mathbb{C}}\cong \mu_{p|0}\otimes \bar{\mu}_{0|m+n-q},  
     \end{align*}
    as $\ad_{\mathfrak{k}^{\mathbb{C}}}$-modules, with corresponding highest weights
\begin{align*}
 \Lambda_{1,0}=\varepsilon_1-\delta_{q-m},\qquad  
 \Lambda_{0,1}= \varepsilon_{p+1}-\delta_n,\qquad 
 \Lambda_{1,1}=\varepsilon_1-\delta_n.
\end{align*}
  \item In Case 3 $(m\leq p<q)$,
    \begin{align*} 
    \mathfrak{m}_{1,0}^{\mathbb{C}} \cong \mu_{m|p-m}\otimes \bar{\mu}_{0|q-p}, \qquad 
    \mathfrak{m}_{0,1}^{\mathbb{C}}\cong \mu_{0|q-p}\otimes \bar{\mu}_{0|m+n-q}, \qquad 
    \mathfrak{m}_{1,1}^{\mathbb{C}}\cong \mu_{m|p-m}\otimes \bar{\mu}_{0|m+n-q},  
    \end{align*}
    as $\ad_{\mathfrak{k}^{\mathbb{C}}}$-modules, with corresponding highest weights
\begin{align*}
 \Lambda_{1,0}=\varepsilon_1-\delta_{q-m},\qquad  
 \Lambda_{0,1}= \delta_{p-m+1}-\delta_n,\qquad 
 \Lambda_{1,1}=\varepsilon_1-\delta_n.
\end{align*}
 \end{itemize}   
\end{corollary}
\begin{proposition}\label{prop: mdec}
In each case in \eqref{SUK}, the $\ad_{\mathfrak{k}}$-representation $\mathfrak{m}$ admits a $Q$-orthogonal 
decomposition of the form 
\begin{align*}
 \mathfrak{m}\cong \mathfrak{m}_{1,0}\oplus \mathfrak{m}_{0,1}\oplus \mathfrak{m}_{1,1},
\end{align*}
where the summands are irreducible $\ad_{\mathfrak{k}}$-representations satisfying
\begin{gather*}
   [\mathfrak{m}_{1,0}, \mathfrak{m}_{1,0}]\subseteq \mathfrak{k},\qquad 
   [\mathfrak{m}_{0,1}, \mathfrak{m}_{0,1}]\subseteq \mathfrak{k},\qquad
   [\mathfrak{m}_{1,1}, \mathfrak{m}_{1,1}]\subseteq \mathfrak{k},
   \\[.1cm] 
   [\mathfrak{m}_{1,0}, \mathfrak{m}_{1,1}]\subseteq \mathfrak{k} \oplus \mathfrak{m}_{0,1},\qquad 
   [\mathfrak{m}_{0,1}, \mathfrak{m}_{1,1}]\subseteq \mathfrak{k}\oplus \mathfrak{m}_{1,0},\qquad 
   [\mathfrak{m}_{1,0}, \mathfrak{m}_{0,1}]\subseteq \mathfrak{k}\oplus \mathfrak{m}_{1,1}.
\end{gather*}
\end{proposition}
\begin{proof}
As in \eqref{eq: real}, $\mathfrak{m}_{\bar{0}}$ is spanned by  $(-1)$-eigenvectors of the star operation restricted to 
$\mathfrak{m}_{\bar{0}}^{\mathbb{C}}$. Using \propref{prop: mdecA}, the star operation on 
$\mathfrak{m}_{\bar{0}}^{\mathbb{C}}$ splits into restrictions on
$(\mathfrak{m}_{-1,0}^{\mathbb{C}}\oplus \mathfrak{m}_{1,0}^{\mathbb{C}})_{\bar{0}}$, 
$(\mathfrak{m}_{0,-1}^{\mathbb{C}}\oplus \mathfrak{m}_{0,1}^{\mathbb{C}})_{\bar{0}}$ 
and $(\mathfrak{m}_{-1,-1}^{\mathbb{C}}\oplus \mathfrak{m}_{1,1}^{\mathbb{C}})_{\bar{0}}$, respectively. 
For each pair $(k_1,k_2)\in\{(1,0),(0,1),(1,1)\}$, the real $(-1)$-eigenspace of the star operation on 
$(\mathfrak{m}_{-k_1,-k_2}^{\mathbb{C}}\oplus \mathfrak{m}_{k_1,k_2}^{\mathbb{C}})_{\bar{0}}$ 
is spanned by $A_{\alpha}$ and $B_{\alpha}$, $\alpha\in (\Delta_{k_1,k_2}^{+})_{\bar{0}}$, so
\begin{align*}
 \mathfrak{m}_{\bar{0}}
  =(\mathfrak{m}_{1,0})_{\bar{0}} \oplus (\mathfrak{m}_{0,1})_{\bar{0}}\oplus (\mathfrak{m}_{1,1})_{\bar{0}}.
\end{align*}
Similarly, 
$\mathfrak{m}_{\bar{1}}=(\mathfrak{m}_{1,0})_{\bar{1}}\oplus(\mathfrak{m}_{0,1})_{\bar{1}}\oplus(\mathfrak{m}_{1,1})_{\bar{1}}$,
so the desired decomposition follows. Moreover, $\mathfrak{m}_{k_1,k_2}$ is an irreducible 
$\ad_{\mathfrak{k}}$-representation since $\mathfrak{m}_{k_1,k_2}^{\mathbb{C}}$ is an irreducible 
$\ad_{\mathfrak{k}^{\mathbb{C}}}$-representation.
Finally, the commutation relations follow from those among the root vectors $E_{\alpha}$.
\end{proof}
By \corref{coro: mrep}, the $\ad_{\mathfrak{k}^{\mathbb{C}}}$-modules 
$\mathfrak{m}_{1,0}^{\mathbb{C}}, \mathfrak{m}_{0,1}^{\mathbb{C}}$ and $\mathfrak{m}_{1,1}^{\mathbb{C}}$ have distinct 
highest weights and are therefore mutually non-isomorphic. It follows from \propref{prop: mdec} that the flag supermanifold 
$G/K$ has three inequivalent irreducible isotropy summands. 

It follows from the proof of \propref{prop: mdec} that, for each $(k_1,k_2)\in\{(1,0),(0,1),(1,1)\}$,
\begin{align*}
 \{A_{\alpha}, B_{\alpha}\,|\, \alpha\in \Delta_{k_1,k_2}\cap \Delta_{\bar{0}}^+\}
 \cup\{\sqrt{\imath}A_{\alpha},\sqrt{\imath}B_{\alpha}\,|\, \alpha\in \Delta_{k_1,k_2}\cap \Delta_{\bar{1}}^+\}
\end{align*}
is an $\mathbb{R}$-basis for the $\ad_{\mathfrak{k}}$-module $\mathfrak{m}_{k_1, k_2}$, with
\begin{align}\label{dmkk}
 \sdim(\mathfrak{m}_{k_1,k_2})=2\sdim_{\mathbb{C}}(\mathfrak{m}_{k_1,k_2}^{\mathbb{C}}).
\end{align}
For convenience, we introduce the notation
\begin{align}\label{m123}
 \mathfrak{m}_1=\mathfrak{m}_{0,1}, \qquad 
 \mathfrak{m}_2=\mathfrak{m}_{1,1}, \qquad 
 \mathfrak{m}_3=\mathfrak{m}_{1,0},
\end{align}
and write
\begin{align}\label{ddd}
 d_i=\sdim(\mathfrak{m}_{i}),\qquad i=1,2,3,
\end{align}
as in \eqref{di}.
Explicit expressions for $d_1,d_2,d_3$ are readily obtained and are given in \lemref{lem: msdim} below.

\subsubsection{Structure constants and Ricci curvature}
\label{sec: strconsA}

\setlength{\tabcolsep}{10pt} 
\renewcommand{\arraystretch}{1.5}
\begin{table}[t]
\begin{center}
\begin{tabular}{l|c|c|c}\hline
 &  $c_1$& $c_2$ & $c_3$   \\\hline
Case 1: $ p<q\leq m$ & $m-n-p$  &  $m-n+p-q$  & $q$  \\
Case 2: $p<m<q$ & $m-n-p$  & $-m-n+p+q$  &  $2m-q$  \\
Case 3: $m\leq p<q$ & $-m-n+p$ &  $m-n-p+q$ &   $2m-q$\\ \hline
\end{tabular}
\vspace*{2mm}
\caption{Casimir eigenvalues.}
\label{tab: c}
\end{center}
\end{table}
\begin{proposition}\label{prop: c}
Let $K$ be as in \eqref{SUK} and $\mathfrak{m}$ as in \propref{prop: mdec} and \eqref{m123}.
Then, the Casimir eigenvalues \eqref{Cc} are given as in Table~\ref{tab: c}.
\end{proposition}
\begin{proof}
In each of the three cases, the eigenvalues of $C_{\mathfrak{m}_i, Q|_{\mathfrak{k}}}$ are given by
$c_i=-(\Lambda_i+2\rho_{\mathfrak{k}^{\mathbb{C}}}, \Lambda_i)$, $i=1,2,3$,
where $\rho_{\mathfrak{k}^{\mathbb{C}}}$ is the Weyl vector of $\mathfrak{k}^{\mathbb{C}}$ (see \eqref{Weyl}),
and where $\Lambda_3,\Lambda_1,\Lambda_2$ are the highest weights given in \corref{coro: mrep}.
To compute these eigenvalues, we use that
\begin{align*}
 (\lambda, \mu)=\sum_{i=1}^{m+n}\lambda(E_{ii})\mu(\bar{E}_{ii}),\qquad \lambda,\mu\in(\mathfrak{h}^{\mathbb{C}})^\vee,
\end{align*}
where the right $Q$-dual matrices, $\bar{E}_{11},\ldots,\bar{E}_{m+n,m+n}$, to the diagonal matrix units, 
$E_{11},\ldots,E_{m+n,m+n}$, are given by $\bar{E}_{ii}=-E_{ii}$ if $i\leq m$ and $\bar{E}_{ii}=E_{ii}$ if $i>m$. 
In Case 1, for example, we have
\begin{align*}
 2\rho_{\mathfrak{k}^{\mathbb{C}}}
  &=\sum_{i=1}^p(p-2i+1)\varepsilon_i
  +\sum_{i=p+1}^q\!\big(q-p-2(i-p)+1\big)\varepsilon_i
  \\
  &\quad+\sum_{i=q+1}^m\!\big(m-q-n-2(i-q)+1\big)\varepsilon_i
  +\sum_{\mu=1}^n(m-q+n-2\mu+1)\delta_\mu,
\end{align*}
while the highest weight of $\mathfrak{m}_2^{\mathbb{C}}$ is $\Lambda_2=\varepsilon_1-\delta_n$, so
\begin{align*} 
 c_2=-(\Lambda_2+2\rho_{\mathfrak{k}^{\mathbb{C}}},\Lambda_2)
  =-(2\rho_{\mathfrak{k}^{\mathbb{C}}},\varepsilon_1-\delta_n)
  =-(p-1)(\varepsilon_1,\varepsilon_1)+(m-q-n+1)(\delta_n,\delta_n)
  =m-n+p-q.
\end{align*}
The eigenvalues $c_1$ and $c_3$ are computed similarly, as are the eigenvalues in Cases 2 and 3.
\end{proof}
\begin{corollary}\label{cor: cccb}
Let the notation be as in \propref{prop: c}. Then,
\begin{align*}
 b=c_1+c_2+c_3.
\end{align*}
\end{corollary}
\begin{lemma}\label{lem: msdim}
Let the notation be as in \propref{prop: mdec}, \eqref{m123} and \eqref{ddd}. 
Then,
\begin{align*}
 d_1=\tfrac{1}{2}(b-2c_2)(b-2c_3),\qquad
 d_2=\tfrac{1}{2}(b-2c_1)(b-2c_3),\qquad
 d_3=\tfrac{1}{2}(b-2c_1)(b-2c_2).
\end{align*} 
\end{lemma}
\begin{proof}
The result follows from \corref{coro: mrep}, \eqref{VW} and \eqref{dmkk}.
\end{proof}
\begin{proposition}\label{prop: StrCons}
Let $K$ be as in \eqref{SUK} and $\mathfrak{m}$ as in \propref{prop: mdec} and \eqref{m123}. 
Then, up to permutation, $[123]$ is the only possibly 
nonzero structure constant, and it is given as
\begin{align*}
 [123]=\tfrac{1}{4}(b-2c_1)(b-2c_2)(b-2c_3).
\end{align*}
\end{proposition} 
\begin{proof}
It follows from \propref{prop: mdec} that $[123]$ is the only possibly nonzero structure constant up to permutation. 
Using \propref{prop: StrCas} and the symmetry of $[ijk]$, we have $2[123]= [312]+[321]=d_3(b_3-2c_3)$.
The result now follows using $b_3=b$, \propref{prop: c} and \lemref{lem: msdim}.
\end{proof}
\begin{remark}
In the super-setting, structure constants can be negative, zero or positive, in contrast to the classical setting 
where they are always non-negative.
\end{remark}
\begin{proposition}\label{prop: RciA}
Let $g=(x_1,x_2,x_3)$ be a $G$-invariant graded Riemannian metric on $M=G/K$, where $G=\SU(m|n)$ and $K$ 
is given in \eqref{SUK}. Set $x^2=x_1^2+x_2^2+x_3^2$. Then, the following holds.
\begin{itemize}
   \item The Ricci coefficients in \eqref{RicQm} are given by
\begin{align*}
 r_i=\frac{b}{2}+\frac{(b-2c_i)x_i(2x_i^2-x^2)}{4x_1x_2x_3},\qquad i=1,2,3.
\end{align*}
   \item The scalar curvature is given by
\begin{align*}
 S=\frac{b}{2}\Big(\frac{d_1}{x_1}+\frac{d_2}{x_2}+\frac{d_3}{x_3}\Big)
  -\frac{[123]\,x^2}{2x_1x_2x_3},
\end{align*}
     with $d_1,d_2,d_3$ given in \lemref{lem: msdim} and $[123]$ in \propref{prop: StrCons}.
   \end{itemize}
\end{proposition}
\begin{proof}
If $[123]\neq0$, then $d_1,d_2,d_3$ are nonzero and the desired 
expression for $r_i$ follows by applying $b_1=b_2=b_3=b$, \lemref{lem: msdim}, \propref{prop: StrCons} and the 
symmetry of $[ijk]$ to the expression in \eqref{rdnot0}.
In Case 1, $b-2c_i=0$ is only possible if $i=3$, in which case $d_1=d_2=0$ while $d_3\neq0$; 
in Case 2, only $b-2c_2$ can be zero, in which case $d_1=d_3=0$ while $d_2\neq0$; 
and in Case 3, only $b-2c_1$ can be zero, in which case $d_2=d_3=0$ while $d_1\neq0$.
By \propref{prop: mdec}, $(2,3)$ is $1$-selected, $(1,3)$ is $2$-selected, and $(1,2)$ is $3$-selected; see \eqref{eq: trcon}.
We can thus use \propref{prop: ricexc} to compute $r_1$ and $r_2$ in Case 1 if $2c_3=b$; 
$r_1$ and $r_3$ in Case 2 if $2c_2=b$;
and $r_2$ and $r_3$ in Case 3 if $2c_1=b$. In every case, the desired expression for $r_i$ follows.
Finally, the expression for $S$ follows by applying \lemref{lem: msdim}, \propref{prop: StrCons} 
and the symmetry of $[ijk]$ to the expression in \propref{prop: scalar}.
\end{proof}

\subsubsection{Einstein metrics}
\label{sec: Einstein}

Using \propref{prop: RciA}, we can write the Einstein equations \eqref{eq: Ein} as
\begin{align}\label{cs}
 c=\frac{b}{2x_i}+\frac{(b-2c_i)(2x_i^2-x^2)}{4x_1x_2x_3},\qquad i=1,2,3.
\end{align}
As always, if $(g,c)$ provides a solution to \eqref{cs}, then so does $(\lambda g,\frac{c}{\lambda})$ for any
$\lambda\in\mathbb{R}^\times$. If $b=0$, then simultaneously changing the sign on $c$ and on any one of 
$x_1,x_2,x_3$ will also yield a solution. As in the non-super setting, an Einstein metric is \emph{Ricci-flat} if $c=0$.

Eliminating $c$ by combining the equations \eqref{cs} for $i=2,3$ and $i=1,3$, respectively, 
and using \corref{cor: cccb}, yields
\begin{align}
 (c_1x_{23}^+-b)x_{23}^-&=c_3-c_2,
 \label{e1}
 \\[.15cm]
 (b-c_1)(x_{23}^-)^2+\big(b-(b-2c_3)x_{23}^+\big)x_{23}^-
  &=c_1(x_{23}^+)^2-bx_{23}^++2c_2,
 \label{e2}
\end{align}
where $x_{23}^\pm:=\frac{x_2\pm x_3}{x_1}$.
We now multiply \eqref{e2} by $(c_1x_{23}^+-b)^2$ and 
then apply \eqref{e1} to replace factors of $(c_1x_{23}^+-b)x_{23}^-$
by $c_3-c_2$.
This turns \eqref{e2} into the quartic relation
\begin{align*}
 \big[x_{23}^+-1\big]
 \big[c_1x_{23}^+-(c_1+2c_2)\big]
 \big[c_1x_{23}^+-(c_1+2c_3)\big]
 \big[c_1x_{23}^+-(c_2+c_3)\big]=0,
\end{align*}
which $x_{23}^+$ must satisfy to provide an Einstein metric.
For $c_1,c_2,c_3$ generic, this quartic relation combined with \eqref{e1}
yields the following four solutions to the Einstein equations:
\begin{align*}
 &(\mathrm{S1}):\quad (x_1,x_2,x_3)
  =(c_2+c_3,c_2,c_3),
   \qquad
   c=1,
 \\[.15cm]
 &(\mathrm{S2}):\quad (x_1,x_2,x_3)
  =(c_1,c_1+c_3,c_3),
   \qquad
   c=1,
 \\[.15cm]
 &(\mathrm{S3}):\quad (x_1,x_2,x_3)
  =(c_1,c_2,c_1+c_2),
   \qquad
   c=1,
 \\[.15cm]
 &(\mathrm{S4}):\quad (x_1,x_2,x_3)
  =(c_1,c_2,c_3),
  \qquad\qquad
   c=1+\frac{(b-2c_1)(b-2c_2)(b-2c_3)}{4c_1c_2c_3}.
\end{align*}
In these solutions, we respectively have
\begin{align}\label{xxx}
 x_1=x_2+x_3,\qquad
 x_2=x_1+x_3,\qquad
 x_3=x_1+x_2,\qquad
 x_1+x_2+x_3=b,
\end{align}
and we note that the expression for $c$ in (S4) is well-defined as long as $x_i\neq0$, $i=1,2,3$.

Since we must have $x_1,x_2,x_3\in\mathbb{R}^\times$, the Casimir eigenvalues $c_1,c_2,c_3$ could be related 
in a way that invalidates one or more of the four solutions.
Indeed, if $c_1=0$, $c_2=0$ or $c_3=0$, then the only viable solution is (S1), (S2) or (S3), respectively. 
Similarly, if $c_2+c_3=0$, $c_1+c_3=0$ or $c_1+c_2=0$, then
(S1), (S2), respectively (S3), is inviable. However, in these situations, other solutions may be present.
To see this, let us analyse the cases $c_1=0$ and $c_2+c_3=0$, noting that some of the scenarios discussed below cannot 
arise for the Casimir eigenvalues $c_1,c_2,c_3$ given in Table~\ref{tab: c}.

First suppose $c_1=0$.
If $c_2+c_3=0$, then, in fact, $c_2=c_3=0$, so $c=0$ while $g$ is unconstrained. 
If $c_2+c_3\neq0$ and $c_2c_3=0$, then $c=0$ and $g\in\{(x_1,x_2,x_2\pm x_1)\,|\, x_2\neq\mp x_1;\,x_1x_2\neq0\}$.
If $c_2+c_3\neq0$ and $c_2c_3\neq0$, then (S1) holds, that is, $c=\frac{1}{\lambda}$ and $g=\lambda(c_2+c_3,c_2,c_3)$, 
$\lambda\in\mathbb{R}^\times$. 

Now suppose $c_2+c_3=0$. Since $c_1=0$ was covered above, we also suppose $c_1\neq0$.
If $c_3=0$, then $c=0$ and $g\in\{(x_1,x_2,x_1-x_2)\,|\, x_2\neq x_1;\,x_1x_2\neq0\}$.
If $c_3\neq0$, then (S4) holds, that is, $c=\frac{c_1^2}{4\lambda c_3^2}$ 
and $g=\lambda(c_1,c_2,c_3)$, $\lambda\in\mathbb{R}^\times$, provides a solution.
However, this is not the only solution if $c_3\notin\{0,c_1\}$, as then (S3) also holds, that is,
$c=\frac{1}{\lambda}$ and $g=\lambda(c_1,c_2,c_1+c_2)$, $\lambda\in\mathbb{R}^\times$, provides a solution. 
Further, if $c_3\notin\{0,c_1,-c_1\}$, then (S2) also holds, that is, $c=\frac{1}{\lambda}$ and 
$g=\lambda(c_1,c_1+c_3,c_3)$, $\lambda\in\mathbb{R}^\times$, provides a solution.

\subsubsection{Classification results}
\label{sec: Aclass}

Following the analysis in \secref{sec: Einstein}, we now summarise the solutions to \eqref{cs} in the specific cases of our 
interest, the flag supermanifolds $M=G/K$ built from $G=\SU(m|n)$ with $K$ given 
in \eqref{SUK}. By invoking permutation symmetry, we have 
covered all possible scenarios in \secref{sec: Einstein}. As established in \secref{sec: isosum}, the corresponding isotropy 
representation decomposes into a sum of three inequivalent nonzero irreducible $\ad_{\mathfrak{k}}$-modules,
$\mathfrak{m}=\mathfrak{m}_1\oplus\mathfrak{m}_2\oplus\mathfrak{m}_3$, with the labelling for the summands given 
in \eqref{m123} and explained in \secref{sec: isosum}.
For all $x_1,x_2,x_3,c\in\mathbb{R}$, we introduce
\begin{align*}
 [\,x_1:x_2:x_3\mid c\,]=\begin{cases} \big\{(\lambda g,\tfrac{c}{\lambda})\,|\, 
  g=(x_1,x_2,x_3);\,\lambda\in\mathbb{R}^\times\big\},\ &\text{if $x_1x_2x_3\neq0$,}\\[.15cm]
   \emptyset,\ &\text{if $x_1x_2x_3=0$,}\end{cases}
\end{align*}
and
\begin{align*}
 \Fc_1(u,v)&=
\begin{cases} \big\{\big((x_2+x_3,x_2,x_3),0\big)\,|\, x_2+x_3\neq0;\,x_2x_3\neq0\big\},\ &\text{if $u=p\ \&\ v=q$,}
\\[.15cm]
\emptyset,\ &\text{otherwise,}\end{cases}
\\[.15cm]
\Fc_3(u,v)&=\begin{cases}\big\{\big((x_1,x_2,x_1+x_2),0\big)\,|\, x_1+x_2\neq0;\,x_1x_2\neq0\big\},\ &\text{if $u=p\ \&\ v=q$,}
\\[.15cm]
\emptyset,\ &\text{otherwise.}\end{cases}
\end{align*}
\begin{theorem}\label{th: SU}
Let $M=G/K$ be a flag supermanifold, where $G=\SU(m|n)$ and $K$ is given in \eqref{SUK}, 
with $m,n\geq1$ such that $m+n\geq4$ and, in the case $m=n$, such that $n\geq3$.
Then, the $G$-invariant Einstein metrics $g$ on $M$, with corresponding Einstein constant $c$, 
are of the form \eqref{gxxx} and classified as follows:
\begin{align*}
 \text{Case 1 } (p<q\leq m):&\qquad (g,c)\in\Cc_{1,1}\cup\Cc_{1,2}\cup\Cc_{1,3}\cup\Cc_{1,4}\cup\Fc_3(m-n,2(m-n)),
 \\[.15cm]
 \text{Case 2 } (p<m<q):&\qquad (g,c)\in\Cc_{2,1}\cup\Cc_{2,2}\cup\Cc_{2,3}\cup\Cc_{2,4}\cup\Fc_1(n-m,2m)\cup\Fc_3(m-n,2n),
 \\[.15cm]
 \text{Case 3 } (m\leq p<q):&\qquad (g,c)\in\Cc_{3,1}\cup\Cc_{3,2}\cup\Cc_{3,3}\cup\Cc_{3,4}\cup\Fc_1(3m-n,2m),
\end{align*}
where
\begin{align*}
 \Cc_{1,1}&=[\,m-n+p: m-n+p-q: q\mid 1\,]\,,
 \\[.15cm]
 \Cc_{1,2}&=[\,m-n-p: m-n-p+q: q\mid 1\,]\,,
 \\[.15cm]
 \Cc_{1,3}&=[\,m-n-p: m-n+p-q: 2(m-n)-q\mid 1\,]\,,
 \\[.15cm]
 \Cc_{1,4}&=\Big[m-n-p: m-n+p-q: q\mid 1+\frac{2p(q-p)(m-n-q)}{q(m-n-p)(m-n+p-q)}\Big],
\end{align*}
\begin{align*}
 \Cc_{2,1}&=[\,m-n+p: -m-n+p+q: 2m-q\mid 1\,]\,,
 \\[.15cm]
 \Cc_{2,2}&=[\,m-n-p: 3m-n-p-q: 2m-q\mid 1\,]\,,
 \\[.15cm]
 \Cc_{2,3}&=[\,m-n-p: -m-n+p+q: -2n+q\mid 1\,]\,,
 \\[.15cm]
 \Cc_{2,4}&=\Big[m-n-p: -m-n+p+q: 2m-q\mid 1+\frac{2p(m+n-q)(2m-p-q)}{(m-n-p)(m+n-p-q)(2m-q)}\Big],
\end{align*}
\begin{align*}
 \Cc_{3,1}&=[\,3m-n-p: m-n-p+q: 2m-q\mid 1\,]\,,
 \\[.15cm]
 \Cc_{3,2}&=[\,-m-n+p: m-n+p-q: 2m-q\mid 1\,]\,,
 \\[.15cm]
 \Cc_{3,3}&=[\,-m-n+p: m-n-p+q: -2n+q\mid 1\,]\,,
 \\[.15cm]
 \Cc_{3,4}&=\Big[-m-n+p: m-n-p+q: 2m-q\mid 1+\frac{2(p-q)(m+n-q)(2m-p)}{(m+n-p)(m-n-p+q)(2m-q)}\Big].
\end{align*}
\end{theorem}

In accordance with symmetry relations between circled $\SU(m|n)$ and $\SU(n|m)$ Dynkin diagrams, 
the family of solution spaces $\Cc_{k,l}$ and $\Fc_k(u,v)$ associated with $\SU(m|n)$ is closed under 
simultaneously mapping the parameters and arguments as
\begin{align*}
 m\mapsto n,\quad n\mapsto m,\quad p\mapsto m+n-q,\quad q\mapsto m+n-p;\qquad 
 x_1\mapsto -x_3,\quad x_2\mapsto -x_2,\quad x_3\mapsto-x_1.
\end{align*}
Concretely, under these simultaneous maps,
\begin{align*}
 \Cc_{k,1}\mapsto\Cc_{4-k,3},\qquad
 \Cc_{k,2}\mapsto\Cc_{4-k,2},\qquad
 \Cc_{k,3}\mapsto\Cc_{4-k,1},\qquad
 \Cc_{k,4}\mapsto\Cc_{4-k,4},\qquad
 k=1,2,3,
\end{align*}
and
\begin{align*}
 \Fc_1\mapsto\Fc_3,\qquad
 \Fc_3\mapsto\Fc_1.
\end{align*}
 
Every $\Cc_{k,l}$ in \thmref{th: SU} is of the form $[\,x_1: x_2: x_3\mid c\,]$, so $\Cc_{k,l}=\emptyset$ 
if and only if $x_1x_2x_3=0$. For $q=2(m-n)$, for example, $\Cc_{1,3}=[\,m-n-p: -m+n+p: 0\mid 1\,]=\emptyset$. 
\corref{cor: mnpq} below summarises the result of analysing all possible linear relations between the parameters $m,n,p,q$.

\begin{corollary}\label{cor: mnpq}
Let the notation be as in \thmref{th: SU}. 
Then, the numbers and types of $G$-invariant solutions to the Einstein equation are as follows.
\\[.15cm]
\underline{Case 1 $(p<q\leq m)$ with $m>n$:}
\begin{itemize}
\item Four solutions up to scaling: 
 $\Cc_{1,1},\Cc_{1,2},\Cc_{1,3},\Cc_{1,4},\quad$ if $\ p\neq m-n\ \&\ q-p\neq m-n\ \&\ q\neq2(m-n)$.
\item Three solutions up to scaling: $\Cc_{1,1},\Cc_{1,2},\Cc_{1,4},\quad$ if $\ p\neq m-n\ \&\ q-p\neq m-n\ \&\ q=2(m-n)$.
\item One solution up to scaling: $\begin{cases} \Cc_{1,1},\quad \text{if } \ p=m-n\ \&\ q-p\neq m-n,\\[.15cm]
									\Cc_{1,2},\quad \text{if } \ p\neq m-n\ \&\ q-p=m-n.  \end{cases}$
\item Continuous family of Ricci-flat solutions: $\Fc_3(m-n,2(m-n))$.
\end{itemize}
\underline{Case 1 $(p<q\leq m)$ with $m=n$:}
\begin{itemize}
\item Four solutions up to scaling: $\Cc_{1,1},\Cc_{1,2},\Cc_{1,3},\Cc_{1,4}$.
\end{itemize}
\underline{Case 1 $(p<q\leq m)$ with $m<n$:}
\begin{itemize}
\item Four solutions up to scaling: $\Cc_{1,1},\Cc_{1,2},\Cc_{1,3},\Cc_{1,4},\quad$ if $\ p\neq n-m\ \&\ q-p\neq n-m$.
\item Three solutions up to scaling: 
 $\begin{cases} \Cc_{1,1},\Cc_{1,3},\Cc_{1,4},\quad \text{if } \ p\neq n-m\ \&\ q-p=n-m,\\[.15cm]
			\Cc_{1,2},\Cc_{1,3},\Cc_{1,4},\quad \text{if } \ p=n-m\ \&\ q-p\neq n-m.  \end{cases}$
\end{itemize}
\underline{Case 2 $(p<m<q)$ with $m>n$:}
\begin{itemize}
\item Four solutions up to scaling: 
 $\Cc_{2,1},\Cc_{2,2},\Cc_{2,3},\Cc_{2,4},\quad$ if $\ p\neq m-n\ \&\ p+q\neq 3m-n\ \&\ q\neq2n$.
\item Three solutions up to scaling: 
 $\begin{cases} \Cc_{2,1},\Cc_{2,2},\Cc_{2,4},\quad \text{if } \ p\neq m-n\ \&\ p\neq3(m-n)\ \&\ q=2n,\\[.15cm]
			\Cc_{2,1},\Cc_{2,3},\Cc_{2,4},\quad \text{if } \ p\neq m-n\ \&\ p+q=3m-n\ \&\ q\neq 2n.  \end{cases}$
\item Two solutions up to scaling: $\Cc_{2,1},\Cc_{2,4},\quad$ if $\ p=3(m-n)\ \&\ q=2n$.
\item One solution up to scaling: 
 $\begin{cases} \Cc_{2,1},\quad \text{if } \ p=m-n\ \&\ q\neq2n,\\[.15cm]
			\Cc_{2,2},\quad \text{if } \ p\neq m-n\ \&\ p+q=m+n\ \&\ q\neq 2n.  \end{cases}$
\item Continuous family of Ricci-flat solutions: $\Fc_3(m-n,2n)$.
\end{itemize}
\underline{Case 2 $(p<m<q)$ with $m=n$:}
\begin{itemize}
\item Four solutions up to scaling: $\Cc_{2,1},\Cc_{2,2},\Cc_{2,3},\Cc_{2,4},\quad$ if $\ p+q\neq 2m$.
\item No solutions if $p+q=2m$.
\end{itemize}
\underline{Case 2 $(p<m<q)$ with $m<n$:}
\begin{itemize}
\item Four solutions up to scaling: 
 $\Cc_{2,1},\Cc_{2,2},\Cc_{2,3},\Cc_{2,4},\quad$ if $\ p\neq n-m\ \&\ p+q\neq 3m-n\ \&\ q\neq2m$.
\item Three solutions up to scaling: 
 $\begin{cases} \Cc_{2,1},\Cc_{2,3},\Cc_{2,4},\quad \text{if } \ p\neq n-m\ \&\ p+q=3m-n\ \&\ q\neq 2m,\\[.15cm]
			\Cc_{2,2},\Cc_{2,3},\Cc_{2,4},\quad \text{if } \ p=n-m\ \&\ q\neq4m-2n\ \&\ q\neq2m.  \end{cases}$
\item Two solutions up to scaling: $\Cc_{2,3},\Cc_{2,4},\quad$ if $\ p=n-m\ \&\ q=4m-2n$.
\item One solution up to scaling: $\begin{cases} \Cc_{2,2},\quad \text{if } \ p\neq n-m\ \&\ p+q=m+n\ \&\ q\neq 2m,\\[.15cm]
									\Cc_{2,3},\quad \text{if } \ p\neq n-m\ \&\ q=2m.  \end{cases}$
\item Continuous family of Ricci-flat solutions: $\Fc_1(n-m,2m)$.
\end{itemize}
\underline{Case 3 $(m\leq p<q)$ with $m>n$:}
\begin{itemize}
\item Four solutions up to scaling: $\Cc_{3,1},\Cc_{3,2},\Cc_{3,3},\Cc_{3,4},\quad$ if $\ q-p\neq m-n\ \&\ q\neq 2n$.
\item Three solutions up to scaling:
 $\begin{cases} \Cc_{3,1},\Cc_{3,2},\Cc_{3,4},\quad \text{if } \ q-p\neq m-n\ \&\ q=2n,\\[.15cm]
			\Cc_{3,1},\Cc_{3,3},\Cc_{3,4},\quad \text{if } \ q-p=m-n\ \&\ q\neq 2n.  \end{cases}$
\end{itemize}
\underline{Case 3 $(m\leq p<q)$ with $m=n$:}
\begin{itemize}
\item Four solutions up to scaling: $\Cc_{3,1},\Cc_{3,2},\Cc_{3,3},\Cc_{3,4}$.
\end{itemize}
\underline{Case 3 $(m\leq p<q)$ with $m<n$:}
\begin{itemize}
\item Four solutions up to scaling: 
 $\Cc_{3,1},\Cc_{3,2},\Cc_{3,3},\Cc_{3,4},\quad$ if $\ p\neq 3m-n\ \&\ q-p\neq n-m\ \&\ q\neq2m$.
\item Three solutions up to scaling: $\Cc_{3,2},\Cc_{3,3},\Cc_{3,4},\quad$ if $\ p=3m-n\ \&\ q-p\neq n-m\ \&\ q\neq2m$.
\item One solution up to scaling: $\begin{cases} \Cc_{3,2},\quad \text{if } \ q-p=n-m\ \&\ q\neq2m,\\[.15cm]
									\Cc_{3,3},\quad \text{if } \ q-p\neq n-m\ \&\ q=2m.  \end{cases}$
\item Continuous family of Ricci-flat solutions: $\Fc_1(3m-n,2m)$.
\end{itemize}
\end{corollary}
The Ricci-flat solutions $\Fc_3(m-n,2(m-n))$ are only present if $n<m\leq2n$; $\Fc_3(m-n,2n)$ if $n<m<2n$;
$\Fc_1(n-m,2m)$ if $m<n<2m$; and $\Fc_1(3m-n,2m)$ if $m<n\leq2m$.
Indeed, there are no Ricci-flat solutions for $m=n$. In fact, the Ricci-flat solutions arise exactly when $p,q,m,n$ 
are related so that $x_1=x_2=x_3=0$ in $\Cc_{k,l}=[\,x_1:x_2:x_3\mid c\,]$, 
which can only happen in $\Cc_{1,3}$, $\Cc_{2,1}$, $\Cc_{2,3}$ and $\Cc_{3,1}$.

\begin{remark}
The isotropy representation of the homogeneous space $\SU(3)/\mathbb T^2$ splits into three inequivalent irreducible 
summands; see~\cite{Kim90}. One can identify an $\SU(3)$-invariant metric on this space with a triple $(x_1,x_2,x_3)$,
and it is well known~\cite{Kim90,PZ24} that such a metric is K\"ahler if and only if $x_i+x_j=x_k$ for some permutation 
$(i,j,k)$ of the set~$\{1,2,3\}$. The similarity between this condition and the three first conditions in \eqref{xxx}
suggests that the Ricci-flat metrics in \thmref{th: SU} may be related to some Calabi--Yau-type structure on~$M=G/K$.
\end{remark}

By considering $n=0$ or $m=0$, our classification results reduce to the similar results in the non-super case.
For $0=n<p<q<m$, we thus find that there are exactly four $G$-invariant Einstein metrics, 
namely $\Cc_{1,1},\Cc_{1,2},\Cc_{1,3},\Cc_{1,4}$, while for $0=m<p<q<n$, there are exactly four $G$-invariant Einstein 
metrics, namely $\Cc_{3,1},\Cc_{3,2},\Cc_{3,3},\Cc_{3,4}$. This is in agreement with the analysis in \cite{Arv93}.

\subsubsection{Positive metrics}

We recall that $g=(x_1,x_2,x_3)$ is said to be positive if $x_1,x_2,x_3>0$. Following from \corref{cor: mnpq}, 
\propref{prop: pos} below lists all the positive Einstein metrics found in \thmref{th: SU}. Corresponding to each $\Cc_{k,l}$, 
it is convenient to introduce
\begin{align*}
 \Cc_{k,l}^\pm=\big\{(\lambda g,\tfrac{c}{\lambda})\,|\, g=(x_1^{k,l},x_2^{k,l},x_3^{k,l});\,
  c=c^{k,l};\,\lambda\in\mathbb{R}^\pm\big\},
\end{align*}
where $x_1^{k,l},x_2^{k,l},x_3^{k,l},c^{k,l}$ are the parameters used in the characterisation 
of $\Cc_{k,l}=[\,x_1^{k,l}:x_2^{k,l}:x_3^{k,l}\mid c^{k,l}\,]$ in \thmref{th: SU}.
For example, for $(m-n-p)( m-n+p-q)(2m-2n-q)\neq0$, we have
\begin{align*}
 \Cc_{1,3}^-=\big\{\big(\lambda(m-n-p,m-n+p-q,2m-2n-q),\tfrac{1}{\lambda}\big)\,|\, \lambda<0\big\}.
\end{align*}
We also introduce
\begin{align*}
 \Fc_1^+=\big\{\big((x_2+x_3,x_2,x_3),0\big)\,|\, x_2,x_3>0\big\},\qquad
 \Fc_3^+=\big\{\big((x_1,x_2,x_1+x_2),0\big)\,|\, x_1,x_2>0\big\}.
\end{align*}
\begin{proposition}\label{prop: pos}
Let the notation be as in \thmref{th: SU}. 
Then, the numbers and types of positive $G$-invariant solutions to the Einstein equation are as follows.
\\[.15cm]
\underline{Case 1 $(p<q\leq m)$:}
\begin{itemize}
\item Four solutions up to scaling: $\Cc_{1,1}^+,\Cc_{1,2}^+,\Cc_{1,3}^+,\Cc_{1,4}^+,\quad$ if $\ \max(p,q-p)<m-n$.
\item One solution up to scaling: $\begin{cases} \Cc_{1,1}^+,\quad \text{if } \ q-p<m-n\leq p,\\[.15cm]
									\Cc_{1,2}^+,\quad \text{if } \ p<m-n\leq q-p,\\[.15cm]
									 \Cc_{1,3}^-,\quad \text{if } \ m-n<\min(p,q-p).\end{cases}$
\item Continuous family of Ricci-flat solutions: $\Fc_3^+,\quad$ if $\ p=m-n\ \&\ q=2(m-n)$.
\end{itemize}
\underline{Case 2 $(p<m<q)$:}
\begin{itemize}
\item Four solutions up to scaling: 
 $\begin{cases} \Cc_{2,1}^+,\Cc_{2,2}^+,\Cc_{2,3}^+,\Cc_{2,4}^+,\quad \text{if } \ n-m<m+n-q<p<m-n,\\[.15cm]
			\Cc_{2,1}^-,\Cc_{2,2}^-,\Cc_{2,3}^-,\Cc_{2,4}^-,\quad \text{if } \ m-n<p<m+n-q<n-m.\end{cases}$
\item One solution up to scaling: 
 $\begin{cases} \Cc_{2,1}^+,\quad \text{if } \ m-n\leq p\ \&\ n-m<m+n-q<p,\\[.15cm]
			\Cc_{2,2}^+,\quad \text{if } \ p\leq m+n-q\ \&\ p<m-n,\\[.15cm]
			\Cc_{2,2}^-,\quad \text{if } \ m+n-q\leq p\ \&\ m+n-q<n-m,\\[.15cm]
			\Cc_{2,3}^-,\quad \text{if } \ n-m\leq m+n-q\ \&\ m-n<p<m+n-q.
			\end{cases}$
\item Continuous family of Ricci-flat solutions: $\begin{cases} \Fc_1^+,\quad \text{if } \ p=n-m\ \&\ q=2m,\\[.15cm]
									 \Fc_3^+,\quad \text{if } \ p=m-n\ \&\ q=2n.\end{cases}$
\end{itemize}
\underline{Case 3 $(m\leq p<q)$:}
\begin{itemize}
\item Four solutions up to scaling: $\Cc_{3,1}^-,\Cc_{3,2}^-,\Cc_{3,3}^-,\Cc_{3,4}^-,\quad$ if $\ \max(m+n-q,q-p)<n-m$.
\item One solution up to scaling: $\begin{cases} \Cc_{3,1}^+,\quad \text{if } \ n-m<\min(m+n-q,q-p),\\[.15cm]
									\Cc_{3,2}^-,\quad \text{if } \ m+n-q<n-m\leq q-p,\\[.15cm]
									 \Cc_{3,3}^-,\quad \text{if } \ q-p<n-m\leq m+n-q.\end{cases}$
\item Continuous family of Ricci-flat solutions: $\Fc_1^+,\quad$ if $\ p=3m-n\ \&\ q=2m$.
\end{itemize}
\end{proposition}
For any $(g,c)\in\Cc_{k,l}^\epsilon$ where $l<4$ in \propref{prop: pos}, the Einstein constant is positive 
if $\epsilon=+$ and negative if $\epsilon=-$. 
If $l=4$, on the other hand, the sign of the Einstein constant depends intricately on the parameters $m,n,p,q$.

In the situation we are considering, when $M=G/K$ is constructed from $G=\SU(m|n)$ with $K$ given in~\eqref{SUK}, 
\corref{coro: mrep} and \propref{prop: mdec} imply that the isotropy summand $\mathfrak m_i$ is neither purely even 
nor purely odd for at least two $i\in\{1,2,3\}$. As a consequence, $Q$ has mixed signature on 
$\mathfrak m_{\bar{0}}=(\mathfrak m_1)_{\bar{0}}\oplus(\mathfrak m_2)_{\bar{0}}\oplus(\mathfrak m_3)_{\bar{0}}$, 
and $Q|_{\mathfrak m_{\bar{0}}}$ is not positive-definite. This means that no $G$-invariant metric on $M$ induces 
a Riemannian metric on the reduced manifold~$|M|$. For metrics on~$M=G/K$, the notion of positivity introduced 
in Section~\ref{sec: diag_curv} thus appears to be a natural super-geometric counterpart to positive definiteness.

The finiteness conjecture in classical geometry states that a compact homogeneous space $G_0/K_0$ 
with pairwise inequivalent isotropy summands can support only finitely many $G_0$-invariant Einstein metrics 
up to scaling; see~\cite{BWZ04}. The existence of continuous families of positive Ricci-flat metrics established in 
Proposition~\ref{prop: pos} demonstrates that this conjecture fails on supermanifolds. 
This is particularly intriguing since, as we showed in Section~\ref{subsec: RicciCoef}, 
the formulas for the Ricci curvature on $M$ are similar to those in the non-super setting. 
The main difference comes from the possible non-positivity of the dimension parameters 
and possible negativity of the structure constants.

In classical homogeneous geometry, Bochner's vanishing theorem (see \cite[Theorem 1.84]{Bes87}) 
implies that a compact homogeneous space $G_0/K_0$ can support a $G_0$-invariant Riemannian metric 
with non-positive Ricci curvature only if the identity component of $G_0$ is a torus. 
The existence of positive $\mathrm{SU}(m|n)$-invariant Einstein metrics on $M$ with vanishing or 
negative Ricci curvature shows that this conclusion does not extend to homogeneous superspaces.

\subsection{$\SOSp(2|2n)$: two isotropy summands}
\label{sec: flagC}

Let $G=\SOSp(2|2n)$ be as in \eqref{SOSp}. We will construct a flag supermanifold $M=G/K$ whose isotropy 
representation decomposes into two inequivalent nonzero irreducible $\ad_{\mathfrak{k}}$-modules:
\begin{align*}
 \mathfrak{m}=\mathfrak{m}_1\oplus \mathfrak{m}_2.
\end{align*}
Retaining the notation from \secref{sec: diag_curv} and \secref{subsubsec: typeC}, 
a $G$-invariant graded Riemannian metric on $M$ is of the form
\begin{align}\label{gxx}
 g= x_1 Q|_{\mathfrak{m}_1}
  + x_2 Q|_{\mathfrak{m}_2},\qquad
  x_1,x_2\in \mathbb{R}^\times.
\end{align}
We will use $(x_1,x_2)$ as a shorthand notation for the sum in \eqref{gxx}. The coefficients in the expression \eqref{BbQ} 
follow from \eqref{eq: BQrel} and are given by $b_1=b_2=-2n$. In the following, we will denote this common value simply by 
$b$:
\begin{align*}
 b=-2n.
\end{align*}

We construct $M$ using a Dynkin diagram of $G$ with the $p^{\text{th}}\!$ node circled, where
\begin{align*}
 2\leq p\leq n,
\end{align*} 
as in \figref{fig: DynC}, and with
\begin{align}\label{KUSp}
 K=\UU(1|p-1)\times \Sp(2(n+1-p)).
\end{align}
The compact real form associated with \eqref{KUSp} is $\mathfrak{k}=\mathfrak{u}(1|p-1)\oplus\mathfrak{sp}(2(n+1-p))$,
and the scalar superproduct $Q$ on $\mathfrak{m}$, where
$\mathfrak{g}=\mathfrak{k}\oplus\mathfrak{m}$, is defined by the corresponding restriction of \eqnref{eq: Qosp}. 
\begin{figure}[hbt!]
\begin{center}
\begin{tikzpicture}[scale=.6]
    \draw[thick] (0 cm,0) circle (.3cm);
    
    \draw[thick] (-0.2 cm,-0.2 cm) -- (0.2cm, 0.2 cm); 
     \draw[thick] (-0.2 cm,0.2 cm) -- (0.2cm, -0.2 cm); 

    \draw[thick] (0.3 cm,0) -- +(1.4 cm,0);
    \draw[thick] (2 cm,0) circle (.3cm);

     \draw[thick] (2.3 cm,0) -- (3  cm,0);
     \draw[dotted,thick] (3.3 cm,0) -- (4 cm,0);
     \draw[thick] (4 cm,0) -- (4.6  cm,0);

    \draw[thick] (5 cm,0) circle (.3cm);
     \draw[thick] (5 cm,0) circle (.4cm);  
    \draw[thick] (5.4 cm,0) -- (6  cm,0);   
    \draw[dotted,thick] (6.3 cm,0) -- (7 cm,0);

    \begin{scope}[shift={(2,0)}]
    \draw[thick] (5.3 cm,0) -- (6 cm,0);

     \draw[thick] (6.3 cm,0) circle (.3cm);

     \draw[thick] (6.7 cm,.1cm) -- (8 cm,.1cm);
     \draw[thick] (6.7 cm,-.1cm) -- (8 cm,-.1cm);

     \draw[thick] (6.6 cm, 0cm) -- (6.9 cm,.25cm);
     \draw[thick] (6.6 cm, 0cm) -- (6.9 cm,-.25cm);
     \draw[thick] (8.3 cm,0) circle (.3cm);
    \end{scope}

    \draw (0,-0.3) node[anchor=north]  {\tiny $\alpha_1$};
         \draw (0,0.3) node[anchor=south]  {\tiny $1$};
    \draw (2,-0.3) node[anchor=north]  {\tiny $\alpha_2$};
         \draw (2,0.3) node[anchor=south]  {\tiny $2$};
    \draw (5,-0.35) node[anchor=north]  {\tiny $\alpha_p$};
         \draw (5,0.3) node[anchor=south]  {\tiny $2$};
    \draw (8.3,-0.3) node[anchor=north]  {\tiny $\alpha_n$};
         \draw (8.3,0.3) node[anchor=south]  {\tiny $2$};
    \draw (10.3,-0.3) node[anchor=north]  {\tiny $\alpha_{n+1}$};
         \draw (10.3,0.3) node[anchor=south]  {\tiny $1$};
  \end{tikzpicture}
\end{center}
  \caption{Dynkin diagram of $\SOSp(2|n)$ with one circled node.}
   \label{fig: DynC}
\end{figure}
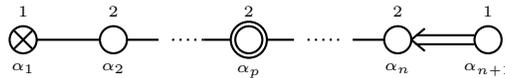

\subsubsection{Isotropy summands}

Let $\mathfrak{m}^{\mathbb{C}}$ denote the $Q$-orthogonal complement of 
$\mathfrak{k}^{\mathbb{C}}= \mathfrak{gl}(1|p-1)^{\mathbb{C}}\oplus \mathfrak{sp}(2(n+1-p))^{\mathbb{C}}$ 
in $\mathfrak{g}^{\mathbb{C}}=\mathfrak{osp}(2|2n)^{\mathbb{C}}$,
where $\mathfrak{k}^{\mathbb{C}}$ is the complexification of $\mathfrak{k}$.
For the root-space decomposition \eqref{eq: mComp} of $\mathfrak{m}^{\mathbb{C}}$, we introduce
\begin{align*}
 \mathfrak{m}_k^{\mathbb{C}}:= \bigoplus_{\alpha\in \Delta_{k}} \mathfrak{g}_{\alpha}^{\mathbb{C}},\qquad
 \Delta_{k}:= \big\{ \sum_{i=1}^{n+1} c_i  \alpha_i \in \Delta_M\,|\, c_p=k\big\},\qquad
 k\in I_{\mathfrak{m}^{\mathbb{C}}},
\end{align*}
where $I_{\mathfrak{m}^{\mathbb{C}}}=\{-2,-1,1,2\}$. If $\Delta_k$ is non-empty, then $\mathfrak{m}_k^{\mathbb{C}}$ 
is an $\ad_{\mathfrak{k}^{\mathbb{C}}}$-invariant subspace of $\mathfrak{m}^{\mathbb{C}}$.

Let $\nu_{2|2n}$ denote the natural representation of $\mathfrak{osp}(2|2n)^{\mathbb{C}}$, so
$\mathfrak{g}^{\mathbb{C}}\cong \bigw_s^2(\nu_{2|2n})$, where $\bigw_s^2$ denotes the second super exterior power.
We also recall that, for all applicable $r,s$, $\mu_{r|s}$ denotes the natural representation of 
$\mathfrak{gl}(r|s)^{\mathbb{C}}$, with dual representation $\bar{\mu}_{r|s}$.
\begin{proposition}\label{prop: mdecC}
As an $\ad_{\mathfrak{k}^{\mathbb{C}}}$-module, $\mathfrak{m}^{\mathbb{C}}$ 
admits the $Q$-orthogonal decomposition
\begin{align*}
 \mathfrak{m}^{\mathbb{C}}= \mathfrak{m}_{-2}^{\mathbb{C}} \oplus\mathfrak{m}_{-1}^{\mathbb{C}}
  \oplus\mathfrak{m}_{1}^{\mathbb{C}} \oplus  \mathfrak{m}_{2}^{\mathbb{C}}.
\end{align*}
Moreover, $\mathfrak{m}_i^{\mathbb{C}}$ and $\mathfrak{m}_{-i}^{\mathbb{C}}$ are irreducible and dual for each $i=1,2$.
\end{proposition}
\begin{proof}
Let $\pm\varepsilon,\pm\delta_1,\dots,\pm\delta_n$ denote the weights of $\nu_{2|2n}$. In accord with the regular 
embedding $\mathfrak{k}^{\mathbb{C}}\subseteq\mathfrak{g}^{\mathbb{C}}$, we denote by 
$\varepsilon,\delta_1,\dots,\delta_{p-1}$ the weights of $\mu_{1|p-1}$, by $-\varepsilon,-\delta_1,\dots,-\delta_{p-1}$ the 
weights of $\bar{\mu}_{1|p-1}$, and by $\pm \delta_p, \dots, \pm \delta_n$ the weights of $\nu_{0|2(n+1-p)}$. We then have
\begin{align*}
 \mathfrak{k}^{\mathbb{C}}\cong\mu_{1|p-1}\otimes\bar{\mu}_{1|p-1}\oplus\bigw_s^2(\nu_{0|2(n+1-p)}).
\end{align*} 
Restricting to $\mathfrak{k}^{\mathbb{C}}$, we have
$\nu_{2|2n}|_{\mathfrak{k}^{\mathbb{C}}}= \mu_{1|p-1} \oplus \bar{\mu}_{1|p-1} \oplus \nu_{0|2(n+1-p)}$
and hence the $\mathfrak{k}^{\mathbb{C}}$-module isomorphisms
\begin{align*}
   \mathfrak{g}^{\mathbb{C}}|_{\ad_{\mathfrak{k}^{\mathbb{C}}}}
   &\cong\bigw_s^2(\nu_{2|2n}|_{\mathfrak{k}^{\mathbb{C}}}) 
      \cong\bigw_s^2( \mu_{1|p-1} \oplus \bar{\mu}_{1|p-1} \oplus \nu_{0|2(n+1-p)})
   \\[.15cm]
    &\cong\bigw_s^2(\mu_{1|p-1})
       \oplus \bigw_s^2(\bar{\mu}_{1|p-1})
       \oplus \bigw_s^2(\nu_{0|2(n+1-p)})
    \\[.15cm] 
   &\quad \oplus (\mu_{1|p-1}\otimes \bar{\mu}_{1|p-1})
      \oplus(\mu_{1|p-1}\otimes \nu_{0|2(n+1-p)})
      \oplus (\bar{\mu}_{1|p-1}\otimes \nu_{0|2(n+1-p)})
     \\[.15cm]
    &\cong\mathfrak{k}^{\mathbb{C}} \oplus (\mu_{1|p-1}\otimes \nu_{0|2(n+1-p)}) 
       \oplus (\bar{\mu}_{1|p-1}\otimes \nu_{0|2(n+1-p)}) 
       \oplus \bigw_s^2(\mu_{1|p-1})
       \oplus \bigw_s^2(\bar{\mu}_{1|p-1}).
\end{align*}
As $\mathfrak{k}^{\mathbb{C}}$ arises as direct summand, it follows that
\begin{align*}
 \mathfrak{m}^{\mathbb{C}}
  =\mathfrak{m}_{-2}^{'\mathbb{C}} 
   \oplus\mathfrak{m}_{-1}^{'\mathbb{C}}
   \oplus\mathfrak{m}_{1}^{'\mathbb{C}} 
   \oplus\mathfrak{m}_{2}^{'\mathbb{C}},
\end{align*}
where
\begin{align*}
 \mathfrak{m}_{-2}^{'\mathbb{C}}\cong \bigw_s^2(\bar{\mu}_{1|p-1}),\quad\
 \mathfrak{m}_{-1}^{'\mathbb{C}}\cong \bar{\mu}_{1|p-1}\otimes \nu_{0|2(n+1-p)}, \quad\ 
 \mathfrak{m}_1^{'\mathbb{C}}\cong \mu_{1|p-1}\otimes \nu_{0|2(n+1-p)}, \quad\
 \mathfrak{m}_2^{'\mathbb{C}}\cong \bigw_s^2(\mu_{1|p-1}),
\end{align*}
noting that $\mathfrak{m}_{-i}^{'\mathbb{C}}$ and $\mathfrak{m}_i^{'\mathbb{C}}$ are irreducible and dual for 
each $i=1,2$.

We now show that $\mathfrak{m}_k^{\mathbb{C}}=\mathfrak{m}_k^{'\mathbb{C}}$ for all 
$k\in I_{\mathfrak{m}^{\mathbb{C}}}$, which then concludes the proof. The weights of $\mathfrak{m}_1^{'\mathbb{C}}$ are 
$\delta_i\pm\delta_j$ and $\varepsilon\pm \delta_j$ with $1\leq i<p\leq j\leq n$,
and since they are all contained in $\Delta_1$, the coefficient of $\alpha_p= \delta_{p-1}-\delta_p$ must be $1$. 
This implies that $\mathfrak{m}_1^{'\mathbb{C}}\subseteq\mathfrak{m}_1^{\mathbb{C}}$. Because
\begin{align*} 
 \bigw_s^2(\mu_{1|p-1})
  \cong \bigw_s^2(\mu_{1|0} \oplus \mu_{0|p-1})
  \cong \bigw_s^2(\mu_{0|p-1}) \oplus \mu_{1|0} \otimes \mu_{0|p-1}, 
\end{align*}
the weights of $\mathfrak{m}_2^{'\mathbb{C}}$ are $\delta_i+\delta_j$ with $1\leq i\leq j\leq p-1$ and 
$\varepsilon+\delta_i$ with $1\leq i\leq p-1$. As all these weights are contained in $\Delta_2$, 
it follows that $\mathfrak{m}_2^{'\mathbb{C}}\subseteq \mathfrak{m}_2^{\mathbb{C}}$. 
We similarly have $\mathfrak{m}_{-1}^{'\mathbb{C}}\subseteq \mathfrak{m}_{-1}^{\mathbb{C}}$ 
and $\mathfrak{m}_{-2}^{'\mathbb{C}}\subseteq \mathfrak{m}_{-2}^{\mathbb{C}}$, so
\begin{align*}
 \bigoplus_{k\in I_{\mathfrak{m}^{\mathbb{C}}}} \mathfrak{m}_k^{\mathbb{C}} 
 \subseteq \mathfrak{m}
 =\bigoplus_{k\in I_{\mathfrak{m}^{\mathbb{C}}}}\mathfrak{m}_k^{'\mathbb{C}} 
 \subseteq  \bigoplus_{k\in I_{\mathfrak{m}^{\mathbb{C}}}} \mathfrak{m}_k^{\mathbb{C}}.
\end{align*}
It follows that $\mathfrak{m}_k^{\mathbb{C}}=\mathfrak{m}_k^{'\mathbb{C}}$ for all $k\in I_{\mathfrak{m}^{\mathbb{C}}}$.
\end{proof}
\begin{corollary}\label{cor: mpL}
Let the notation be as in \propref{prop: mdecC}. Then,
\begin{align*}
 \mathfrak{m}_1^{\mathbb{C}}\cong \mu_{1|p-1}\otimes \nu_{0|2(n+1-p)}, \qquad
 \mathfrak{m}_2^{\mathbb{C}}\cong \bigw_s^2(\mu_{1|p-1}),
\end{align*}
with corresponding highest weights
\begin{align*}
 \Lambda_1=\varepsilon+\delta_p, \qquad 
 \Lambda_2=\varepsilon+\delta_1.  
\end{align*}
\end{corollary}
\begin{proposition}\label{prop: rmdecC}
The $\ad_{\mathfrak{k}}$-representation $\mathfrak{m}$ admits a $Q$-orthogonal decomposition of the form
\begin{align*}
   \mathfrak{m}\cong\mathfrak{m}_1\oplus \mathfrak{m}_2,
\end{align*}
where the summands are irreducible $\ad_{\mathfrak{k}}$-representations satisfying
\begin{align*}
 [\mathfrak{m}_1, \mathfrak{m}_1]\subseteq \mathfrak{k}\oplus \mathfrak{m}_2, \qquad 
 [\mathfrak{m}_1, \mathfrak{m}_2]\subseteq \mathfrak{m}_1, \qquad 
 [\mathfrak{m}_2, \mathfrak{m}_2]\subseteq \mathfrak{k}. 
\end{align*}    
\end{proposition}
\begin{proof}
  The proof is analogous to that of \propref{prop: mdec}. 
\end{proof}
By \corref{cor: mpL}, the $\ad_{\mathfrak{k}^{\mathbb{C}}}$-modules $\mathfrak{m}_{1}^{\mathbb{C}}$ and 
$\mathfrak{m}_{2}^{\mathbb{C}}$ have distinct highest weights and are therefore non-isomorphic. It follows from 
\propref{prop: rmdecC} that the flag supermanifold $G/K$ has two inequivalent irreducible isotropy summands. 

It follows from the proof of \propref{prop: rmdecC} that, for each $i=1,2$,
\begin{align*}
 \{A_{\alpha}, B_{\alpha}\,|\, \alpha\in\Delta_{i}\cap \Delta_{\bar{0}}^+\}
  \cup \{ \sqrt{\imath}A_{\alpha},\sqrt{\imath}B_{\alpha}\,|\, \alpha\in\Delta_{i}\cap \Delta_{\bar{1}}^+\}
\end{align*}
is an $\mathbb{R}$-basis for the $\ad_{\mathfrak{k}}$-module $\mathfrak{m}_i$, with
\begin{align}\label{dmi}
 d_i:=\sdim(\mathfrak{m}_i)=2\sdim_{\mathbb{C}}(\mathfrak{m}_i^{\mathbb{C}}).
\end{align}
Explicit expressions for $d_1,d_2$ are readily obtained and are given in \lemref{lem: sdimC} below.

\subsubsection{Structure constants and Ricci curvature}

\begin{proposition}\label{prop: cSp}
Let $K$ be as in \eqref{KUSp} and $\mathfrak{m}$ as in \propref{prop: rmdecC}.
Then, the Casimir eigenvalues \eqref{Cc} are given by
\begin{align*}
 c_1=-n+\tfrac{1}{2}(p-1),\qquad
 c_2=-(p-1).
\end{align*}
\end{proposition}
\begin{proof}
The eigenvalues of $C_{\mathfrak{m}_i, Q|_{\mathfrak{k}}}$ are given by
$c_i=-(\Lambda_i+2\rho_{\mathfrak{k}^{\mathbb{C}}}, \Lambda_i)$, $i=1,2$,
where $\rho_{\mathfrak{k}^{\mathbb{C}}}$ is the Weyl vector of $\mathfrak{k}^{\mathbb{C}}$,
and where $\Lambda_1,\Lambda_2$ are the highest weights given in \corref{cor: mpL}.
To compute these eigenvalues, we use that
\begin{align*}
 (\lambda, \mu)=\sum_{i=1}^{n+1}\lambda(h_i)\mu(\bar{h}_i),\qquad \lambda,\mu\in(\mathfrak{h}^{\mathbb{C}})^\vee,
\end{align*}
where $\{h_1=E_{11}-E_{22},\,h_i=E_{i+1,i+1}-E_{n+1+i,n+1+i}\,|\, 2\leq i\leq n+1\}$ is a basis for the 
Cartan subalgebra $\mathfrak{h}^{\mathbb{C}}$ (cf.~\eqnref{eq: cartanC}).
The right $Q$-dual basis elements are given by
\begin{align*}
  \bar{h}_1=-\tfrac{1}{2}h_1, \qquad \bar{h}_i=\tfrac{1}{2} h_{i}, \qquad i=2,\ldots,n+1,
\end{align*}
so
\begin{align*}
 (\varepsilon,\varepsilon)=\sum_{i=1}^{n+1}\varepsilon(h_i)\varepsilon(\bar{h}_i)=-\frac{1}{2},\qquad
 (\delta_j, \delta_j)=\sum_{i=1}^{n+1}\delta_j(h_i)\delta_j(\bar{h}_i)=\frac{1}{2},\qquad j=1,\ldots,n.
\end{align*}
Since
\begin{align*}
 2\rho_{\mathfrak{k}^{\mathbb{C}}}
  =(1-p)\varepsilon
  +\sum_{j=1}^{p-1}(p-2j+1)\delta_j
  +\sum_{j=p}^n\!\big(2(n+1-j)\big)\delta_j,
\end{align*}
and the highest weights are given by $\Lambda_1=\varepsilon+\delta_p$ and
$\Lambda_2=\varepsilon+\delta_1$, we have
\begin{align*}
 c_1&=-(\Lambda_1+2\rho_{\mathfrak{k}^{\mathbb{C}}},\Lambda_1)
  =(p-2)(\varepsilon,\varepsilon)- (2(n+1-p)+1)(\delta_p, \delta_p)
  =-\tfrac{1}{2}(2n-p+1),
 \\[.15cm]
 c_2&=-(\Lambda_2+2\rho_{\mathfrak{k}^{\mathbb{C}}},\Lambda_2)
  =(p-2)(\varepsilon,\varepsilon)-p(\delta_1,\delta_1)
  =1-p.
\end{align*}
\end{proof}
\begin{corollary}
Let the notation be as in \propref{prop: cSp}. Then,
\begin{align*}
 b=2c_1+c_2.
\end{align*}
\end{corollary}
\begin{lemma}\label{lem: sdimC}
Let the notation be as in \propref{prop: rmdecC} and \eqref{dmi}. 
Then,
\begin{align*}
 d_1=-2(p-2)(b-2c_2),\qquad
 d_2=-(p-2)(b-2c_1).
\end{align*}
\end{lemma}
\begin{proof}
The result follows from \corref{cor: mpL}, \eqref{VW} and \eqref{dmi}.
\end{proof}
\begin{proposition}\label{prop: StrC}
Let $K$ be as in \eqref{KUSp} and $\mathfrak{m}$ as in \propref{prop: rmdecC}.
Then, up to permutation, $[112]$ is the only possibly nonzero structure constant, and it is given as 
\begin{align*}
  [112]=-(p-2)(b-2c_1)(b-2c_2).
\end{align*} 
\end{proposition}
\begin{proof}
\propref{prop: rmdecC} implies that $[112]$ is the only possibly nonzero structure constant up to permutation. 
Using \propref{prop: StrCas} and the symmetry of $[ijk]$, we have $[112]=[211]=d_2(b_2-2c_2)$.
The result now follows using $b_2=b$, \propref{prop: cSp} and \lemref{lem: sdimC}.
\end{proof}
Concretely, $b-2c_1=-(p-1)$ and $b-2c_2=-2(n+1-p)$, and since $2\leq p\leq n$, it follows that $(b-2c_1)(b-2c_2)>0$,
$[112]\leq0$ and that $[112]=0$ if and only if $p=2$.

\begin{proposition}\label{prop: RicScC}
Let $g=(x_1,x_2)$ be a $G$-invariant graded Riemannian metric on $M=G/K$, where $G=\SOSp(2|2n)$ and $K$ 
is given in \eqref{KUSp}. Set $x^2=2x_1^2+x_2^2$. Then, the following holds.
\begin{itemize}
\item The Ricci coefficients in \eqref{RicQm} are given by
\begin{align*}
 r_i=\frac{b}{2}+\frac{(b-2c_i)x_i(2x_i^2-x^2)}{4x_1x_1x_2},\qquad
 i=1,2.
\end{align*}
\item The scalar curvature $S$ is given by 
\begin{align*}
 S=\frac{b}{2}\Big(\frac{d_1}{x_1}+\frac{d_2}{x_2}\Big)
  -\frac{[112]\,x^2}{4x_1x_1x_2},
\end{align*}
with $d_1,d_2$ given in \lemref{lem: sdimC} and $[112]$ in \propref{prop: StrC}.
\end{itemize}
\end{proposition}
\begin{proof}
If $p\neq2$, then, by \lemref{lem: sdimC}, $d_1,d_2$ are nonzero, in which case \thmref{thm: SimRic} implies that 
\begin{align*}
           r_1=\frac{b_1}{2}- \frac{[112]}{2d_1}\frac{x_2}{x_1},\qquad
           r_2=\frac{b_2}{2}+ \frac{[211]}{4d_2}\Big( \frac{x_2^2}{x_1^2}- 2\Big).
\end{align*}
Using $b_1=b_2=b$ and \propref{prop: StrC}, the desired expressions for the Ricci coefficients follow.
If $p=2$, then $d_1=d_2=0$. By \propref{prop: rmdecC}, $(2,1)$ is $1$-selected, and $(1,1)$ is $2$-selected.
The desired expressions for the Ricci coefficients now follow from \propref{prop: ricexc}.
Finally, the expression for $S$ follows from \propref{prop: scalar}. 
\end{proof}

\subsubsection{Classification of Einstein metrics}
\label{subsec: strcon}

Using \propref{prop: RicScC}, we can write the Einstein equations \eqref{eq: Ein} as
\begin{align}\label{cs2}
 c=\frac{b}{2x_i}+\frac{(b-2c_i)(2x_i^2-x^2)}{4x_1x_1x_2},\qquad i=1,2.
\end{align}
Eliminating $c$ by combining the two equations in \eqref{cs2} and using $b=2c_1+c_2$ and $x^2=2x_1^2+x_2^2$, yields
\begin{align*}
 (2x_1-x_2)(c_2x_1-c_1x_2)=0.
\end{align*}
It follows from \propref{prop: cSp} that $c_1,c_2\neq0$, so, up to scaling, there are exactly two solutions to the Einstein 
equations:
\begin{align*}
 &(\mathrm{S1}):\quad (x_1,x_2)
  =(1,2),
   \qquad\ \
   c=c_1,
 \\[.15cm]
 &(\mathrm{S2}):\quad (x_1,x_2)
  =(c_1,c_2),
   \qquad
   c=1+\frac{(2c_1-c_2)c_2}{4c_1^2}.
\end{align*}
This completes our analysis of the Einstein equations and gives rise to the classification result in \thmref{thm: EinC} below.
Mimicking the presentation in \secref{sec: Aclass}, for all $x_1,x_2,c\in\mathbb{R}$, we introduce
\begin{align*}
 [\,x_1:x_2\mid c\,]:=\begin{cases} \big\{(\lambda g,\tfrac{c}{\lambda})\,|\, 
  g=(x_1,x_2);\,\lambda\in\mathbb{R}^\times\big\},\ &\text{if $x_1x_2\neq0$,}\\[.15cm]
   \emptyset,\ &\text{if $x_1x_2=0$.}\end{cases}
\end{align*}
\begin{theorem}\label{thm: EinC}
Let $M=G/K$ be a flag supermanifold, where $G=\SOSp(2|2n)$ and $K$ is given in \eqref{KUSp} with $2\leq p\leq n$. 
Then, the $G$-invariant Einstein metrics $g$ on $M$, with corresponding Einstein constant $c$, are of the form \eqref{gxx} 
and classified as follows:
\begin{align*}
 (g,c)\in[\,1:2\mid -n+\tfrac{1}{2}(p-1)\,]
 \cup\Big[2n+1-p:2(p-1)\mid -\frac{1}{2}-\frac{(p-1)(n+1-p)}{(2n+1-p)^2}\Big].
\end{align*}
\end{theorem}
We highlight the similarities between the solutions in \thmref{thm: EinC} and the ones in \secref{sec: Einstein}
when the former are expressed as
\begin{align*}
 (g,c)\in[\,c_1:b-c_2\mid 1\,]
 \cup\Big[c_1:c_2\mid 1+\frac{(b-2c_1)(b-2c_1)(b-2c_2)}{4c_1c_1c_2}\Big].
\end{align*}
\begin{remark}
The various expressions above, including the ones in \thmref{thm: EinC}, shorten if $p$ is replaced by $q+1$.
This would correspond to enumerating the nodes associated to $\alpha_2,\ldots,\alpha_n$ in the $\SOSp(2|2n)$
Dynkin diagram in \secref{subsubsec: typeC} from $1$ to $n-1$. For consistency with the enumeration in the $\SU(m|n)$ 
cases in \secref{sec: flagA1} and \secref{sec: flagA}, we have opted to use $p$ as indicated.
\end{remark}
When restricting the scaling parameter $\lambda$ to $\mathbb{R}^+$, both solutions in \thmref{thm: EinC} are positive with 
negative Einstein constant.

\addcontentsline{toc}{section}{Acknowledgements}
\section*{Acknowledgements}

This work was supported by the Australian Research Council under the Discovery Project scheme, 
project numbers DP180102185, DP200102316 and DP220102530. The authors thank Andreas Arvanitoyeorgos, Ruibin 
Zhang and Gabriele Tartaglino-Mazzucchelli for discussions and comments.

\end{document}